\documentclass[aps,prl,reprint,superscriptaddress,dvipsnames,nofootinbib,nobibnotes,longbibliography,floatfix]{revtex4-2}
\usepackage{times}
\usepackage[stretch=40]{microtype}
\usepackage{graphicx}
\usepackage{needspace}
\usepackage{dcolumn}
\usepackage{mathrsfs}
\usepackage{pifont}
\usepackage{amsthm}
\usepackage{amsmath}
\usepackage{amssymb}
\usepackage{bm}
\usepackage{bbm}
\usepackage{latexsym}
\usepackage{dsfont}
\usepackage[colorlinks=true,
linkcolor=blue,
urlcolor=blue,
citecolor=blue]{hyperref}
\usepackage[capitalise]{cleveref}
\usepackage{color}
\usepackage{booktabs} 
\makeatletter
\providecommand{\href@noop}[0]{\@secondoftwo}
\makeatother

\newcommand{\ket}[1]{|#1\rangle}
\newcommand{\bra}[1]{\langle#1|}

\newcommand{\dif}{\mathrm{d}}
\newcommand{\Z}{\mathbb{Z}}

\newcommand{\RR}{\mathbb{R}}
\newcommand{\Tr}{\operatorname{Tr}}
\newtheorem{theorem}{Theorem}
\newtheorem{lemma}{Lemma}
\newtheorem{corollary}{Corollary}

\setcitestyle{numbers,square,sort&compress}

\begin{document}
\microtypesetup{expansion=false,protrusion=false}
\title{Radial Coarse Graining Restores Vacuum Majorization in Wigner Phase Space}
\author{Ao-Xiang Liu}
\affiliation{School of Physical Sciences, University of Chinese Academy of Sciences, 1 Yanqihu East Rd, Beijing 101408, China}
\author{Cong-Feng Qiao}
\email[]{qiaocf@ucas.ac.cn}
\affiliation{School of Physical Sciences, University of Chinese Academy of Sciences, 1 Yanqihu East Rd, Beijing 101408, China}
\affiliation{ICTP-AP, University of Chinese Academy of Sciences, Beijing 100190, China}
\date{\today}
\begin{abstract}
Quantum uncertainty limits how strongly Wigner functions can concentrate in phase space. It is commonly conjectured that the vacuum Wigner function continuously majorizes that of any Wigner-positive state. However, by constructing Wigner-positive states whose localized quantum coherence lowers their Wigner entropy below the vacuum value, we show that this conjecture is incorrect. We restore vacuum majorization at finite resolution by integrating single-mode Wigner functions over concentric radial energy shells, obtaining probability vectors majorized by that of the vacuum for all Wigner-positive states and for Wigner-negative states with non-negative shell weights. Anchoring these shells to a fixed oscillator frame breaks affine symplectic invariance, rendering this relation strict for every non-vacuum Gaussian pure state. For Wigner-negative states, the persistence of negative shell weights defines a radial negativity scale with Airy-edge semiclassical asymptotics for highly excited Fock states. Operationally, we combine the majorization relation with partial transposition to obtain a two-mode entanglement criterion directly evaluable from radially binned joint homodyne outcomes without state reconstruction.
\end{abstract}

\maketitle

\textit{Introduction}---The uncertainty principle places a fundamental limit on the phase-space localization of quantum states \cite{heisenberg27}. Standard variance-based relations quantify this bound in the second order of moment expansion \cite{kennard27,robertson29}, while entropic uncertainty relations probe the underlying distributions globally \cite{hirschman57,bialynicki75,beckner75,DD83U,coles17}, underpinning applications from entanglement detection \cite{walborn09,saboia11,li11memory,schneeloch19,haas21b} to quantum cryptography \cite{berta10,tomamichel11,furrer12,tomamichel12,curty14,zhang20cvsdi}. Majorization uncertainty relations establish an even stronger hierarchy by ordering complete distributions to pinpoint states of least disorder; a single relation implies every Schur-concave uncertainty bound \cite{marshall11,partovi11,friedland13,puchala13,rudnicki14,LJ19O,cerf24} while supplying operational criteria for entanglement and quantum steering \cite{zhu23conditional,garttner23prl,garttner23,yang26}. Such comparisons of entire distributions are particularly natural in phase space, where informationally complete quasiprobabilities fully describe the quantum state \cite{moyal49,cahill69a,cahill69b}. For the Husimi $Q$ representation, the Lieb--Solovej theorem establishes that the vacuum distribution continuously majorizes that of every single-mode bosonic state \cite{lieb14}.

\begingroup\microtypesetup{expansion=true,protrusion=true}
\looseness=-1 Among phase-space representations \cite{wigner32,hillery84}, the Wigner distribution stands out as the closest quantum analog to a classical phase-space distribution. Unlike classical probability densities, however, generic Wigner functions can take on negative values \cite{kenfack04,walschaers20,WM23O}, serving as a hallmark of nonclassicality and an indispensable resource for non-Gaussian quantum advantage \cite{mari12,veitch13,pashayan15,rahimi16,albarelli18}. Yet standard continuous majorization orders non-negative integrable densities by comparing the concentration of their decreasing rearrangements \cite{hardy29,ryff65,chong74}, strictly confining its direct application to the Wigner-positive sector. Within this domain, the vacuum is conjectured to continuously majorize all states \cite{vanherstraeten23}, implying the Wigner-entropy conjecture \cite{hertz17,vanherstraeten21}. This ordering has been established for selected Fock-diagonal families and the convex hull of multimode Gaussian states \cite{vanbever21,vanherstraeten23,qian24,boer25}, while the entropy bound holds for broader subclasses and under purity conditions \cite{vanherstraeten21,dias23,vanherstraeten25,qian26}.
\par\endgroup

\Needspace{2\baselineskip}
We construct an explicit family of Wigner-positive states whose Wigner functions have lower R\'enyi entropy than that of the vacuum, demonstrating that the vacuum is not a universal benchmark under continuous majorization even within the Wigner-positive sector. Details are presented in the Supplements attached. This breakdown originates from localized phase-space interference generated by quantum coherence, which alters the density values compared by continuous majorization even after decreasing rearrangement. This prompts the question of whether a universal majorization benchmark can be restored by averaging over finite phase-space regions.

In this work, we show that integrating single-mode Wigner functions over concentric shells bounded by contours of constant oscillator energy in a fixed frame restores vacuum majorization at finite phase-space resolution. For complete radial partitions with bounded shell areas, the shell-probability vector is majorized by that of the vacuum state for every Wigner-positive state and every Wigner-negative state with non-negative shell weights.

While continuous majorization treats all single-mode Gaussian pure states as equivalent to the vacuum under affine symplectic transformations \cite{vanherstraeten23}, anchoring radial shells to a fixed oscillator frame explicitly breaks this symmetry, rendering the coarse-grained relation strict for every non-vacuum Gaussian pure state at every positive shell width. For Wigner-negative states, the persistence of negative shell weights under radial coarse graining defines a radial negativity scale, half of which is asymptotic to the oscillator energy for highly excited Fock states, with an Airy-edge correction \cite{berry77,hanin20}. Operationally, combining the majorization relation with partial transposition yields a two-mode entanglement criterion \cite{peres96,horodecki96,simon00} directly evaluable from radially binned joint homodyne outcomes \cite{braunstein98,leonhardt97}, with finite shell widths treated exactly and no reconstruction of a density matrix or continuous phase-space distribution.

\Needspace{9\baselineskip}
\textit{Coarse-grained phase space}---Let us consider a single-mode bosonic system described by canonical phase-space coordinates $\boldsymbol r=(q,p)$, setting $\hbar=1$. The state $\rho$ is fully characterized by its Wigner distribution \cite{wigner32,hillery84}
\begin{align}
W_\rho(\boldsymbol{r})=\frac{1}{\pi}\int_{-\infty}^{\infty}e^{2ipy}\langle q-y|\rho|q+y\rangle\,\dif y\ ,
\label{eq:wigner}
\end{align}
normalized to unity over phase space. For an arbitrary Fock state $\ket n$, the Wigner function evaluates explicitly to
\begin{align}
W_n(\boldsymbol r)=\pi^{-1}(-1)^nL_n(2|\boldsymbol r|^2)e^{-|\boldsymbol r|^2}\ ,
\end{align}
which for the vacuum reduces to the isotropic Gaussian distribution $W_0(\boldsymbol r)=\pi^{-1}e^{-|\boldsymbol r|^2}$. Here, $L_n$ denotes the Laguerre polynomial of degree $n$. The rotational symmetry shared by all Fock states highlights the radial distance as the natural geometric coordinate for structuring phase-space coarse graining.

\begin{figure}[t]
\centering
\includegraphics[width=\columnwidth]{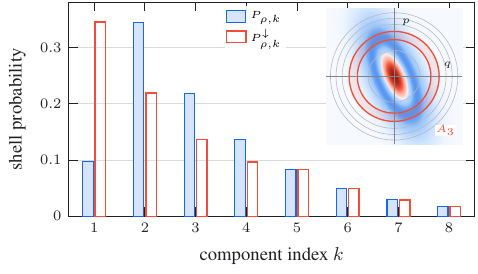}
\caption{Equal-area radial shells at $\Delta=3/2$ for a Wigner-negative, rotated squeezed one-photon state with squeezing magnitude $\varsigma=0.45$. The inset shows $W_\rho(q,p)$ and the shell boundaries, with $A_3$ highlighted. Blue filled bars show the first eight shell weights $P_{\rho,k}$ in radial order; coral open bars show their non-increasing rearrangement $P_{\rho,k}^{\downarrow}$. Here, the negative region lies entirely in $A_1$, whose weight is $P_{\rho,1}=0.096834>0$; $W_\rho$ is non-negative on every later shell, so the complete shell vector is componentwise non-negative.}
\label{fig:shell-construction}
\end{figure}

The harmonic oscillator Hamiltonian $\hat H=(\hat q^2+\hat p^2)/2=\hat n+1/2$ naturally identifies the radial variable $\xi=|\boldsymbol r|^2$ as an energy coordinate, where an interval of width $d$ in $\xi$ spans phase-space area $\pi d$ and corresponds to an energy window $d/2$. Let $\Pi:0=r_0<r_1<r_2<\cdots$ with $r_k\to\infty$ be a complete radial partition whose squared-radius increments $d_k:=r_k^2-r_{k-1}^2$ satisfy $\sup_k d_k<\infty$. The associated concentric shell $A_k(\Pi):=\{\boldsymbol r:r_{k-1}\le|\boldsymbol r|<r_k\}$ has phase-space area $\pi d_k$ and contains the integrated Wigner weight
\begin{align}
P_{\rho,k}(\Pi):=\int_{A_k(\Pi)}W_\rho(\boldsymbol r)\,\dif\boldsymbol r\ .
\label{eq:general-shell-weights}
\end{align}
For any phase-space domain $A$ of finite area, the corresponding region operator $\widehat R_A$ is defined by $\Tr[\rho\widehat R_A]:=\int_AW_\rho(\boldsymbol r)\,\dif\boldsymbol r$ for every state $\rho$ \cite{bracken99,ellinas06,ellinas08}, yielding $P_{\rho,k}(\Pi)=\Tr[\rho\widehat R_{A_k(\Pi)}]$. These region operators are not generally positive and therefore do not define a positive operator-valued measure. Interpreting the shell-weight vector $\boldsymbol P_\rho(\Pi):=(P_{\rho,1},P_{\rho,2},\ldots)$ as a probability vector requires componentwise non-negativity, a condition that depends on both the state and the chosen partition. \cref{fig:shell-construction} illustrates this feature for a rotated squeezed single-photon state: although its Wigner function exhibits pronounced local negativity near the origin, choosing equal-area shells with width $\Delta=3/2$ in $\xi$ yields an integrated shell vector that is componentwise strictly positive.

Integrating $W_0$ over the $k$th shell of the partition $\Pi$ gives
\begin{align}
P_{0,k}(\Pi)=e^{-r_{k-1}^2}-e^{-r_k^2}\ .
\label{eq:vacuum-general-weights}
\end{align}
Strict monotonic radial decay ensures that the integral of $W_0$ over any phase-space region of fixed area is uniquely maximized when that region is an origin-centered disk. We use these disk integrals to construct the majorization benchmark below from the total areas of finite selections of shells.

\textit{Vacuum majorization for radial partitions}---The theory of majorization provides a natural framework for quantifying phase-space localization by imposing a preorder on discrete probability distributions: for two probability vectors $\boldsymbol x$ and $\boldsymbol y$, one writes $\boldsymbol x\prec\boldsymbol y$ if their non-increasing rearrangements satisfy $\sum_{j=1}^m x_j^\downarrow\le\sum_{j=1}^m y_j^\downarrow$ for all $m\ge1$, with equality holding for the total sum. To bound the combined weight of any $m$ shells, we introduce the supremum of their total phase-space area divided by $\pi$,
\begin{align}
D_m(\Pi):=\sup_{\substack{I\subset\mathbb N\\|I|=m}}\sum_{i\in I}d_i\ ,
\quad D_0(\Pi):=0\ .
\label{eq:Dm-main}
\end{align}
Bounded increments $\sup_k d_k<\infty$ ensure that each $D_m(\Pi)$ is finite, while completeness makes the sequence grow without bound. These values form a concave sequence and determine the vacuum benchmark vector
\begin{align}
\label{eq:benchmark}
\boldsymbol P_0(\Pi^\downarrow):=(P_{0,1}(\Pi^\downarrow),P_{0,2}(\Pi^\downarrow),\ldots,P_{0,m}(\Pi^\downarrow),\ldots)\ ,
\end{align}
with components $P_{0,m}(\Pi^\downarrow):=e^{-D_{m-1}(\Pi)}-e^{-D_m(\Pi)}$ that decrease monotonically with $m$ by concavity of $D_m(\Pi)$.

Consequently, $\boldsymbol P_0(\Pi^\downarrow)$ forms a non-increasing probability vector whose cumulative sums satisfy
\begin{align}
\sum_{j=1}^m P_{0,j}(\Pi^\downarrow)=1-e^{-D_m(\Pi)}\ ,
\end{align}
matching the exact integral of the vacuum distribution $W_0$ over an origin-centered disk of area $\pi D_m(\Pi)$. The following theorem establishes this vector as a universal majorization benchmark whenever the shell weights are non-negative.

\begin{theorem}
\label{thm:radial-majorization}
Let $\rho$ be any single-mode quantum state and let $\Pi$ be a complete radial partition with bounded increments ($\sup_k d_k<\infty$). If $P_{\rho,k}(\Pi)\ge0$ for every $k$, then $\boldsymbol P_\rho(\Pi)$ is a normalized probability vector majorized by the vacuum benchmark defined in Eq.~\eqref{eq:benchmark}:
\begin{align}
\label{eq:general-Lorenz-main}
\boldsymbol P_\rho(\Pi)\prec\boldsymbol P_0(\Pi^\downarrow)\ .
\end{align}
Equivalently, for every integer $m\ge1$,
\begin{align}
\sum_{j=1}^{m}P_{\rho,j}^{\downarrow}(\Pi)\le 1-e^{-D_m(\Pi)}\ .
\end{align}
\end{theorem}

The proof rests on a geometric tail inequality in phase space. For the origin-centered open disk $D_\xi:=\{\boldsymbol r:|\boldsymbol r|^2<\xi\}$, we define the exterior tail weight $\mathcal T_\rho(\xi):=\Tr[\rho(\mathbb I-\widehat R_{D_\xi})]$, which determines individual shell weights through finite differences $P_{\rho,k}(\Pi)=\mathcal T_\rho(r_{k-1}^2)-\mathcal T_\rho(r_k^2)$. Rotational invariance of the disk removes all off-diagonal Fock contributions from the integral, giving $\mathcal T_\rho(\xi)=\sum_{n\ge0}\rho_n\mathcal T_n(\xi)$ with populations $\rho_n=\langle n|\rho|n\rangle$. For each Fock state $\ket n$, the Mehler expansion \cite[Eq.~18.18.28]{dlmf24} reveals that the rescaled tail $e^\xi\mathcal T_n(\xi)$ is nondecreasing on $[0,\infty)$. Positivity and unit trace of physical density operators ($\rho_n\ge0$, $\sum_n\rho_n=1$) directly preserve this monotonicity for arbitrary states, yielding the tail inequality
\begin{align}
\mathcal T_\rho(\xi+d)\ge e^{-d}\mathcal T_\rho(\xi),
\qquad \xi\ge0,\quad d>0,
\label{eq:state-tail-main}
\end{align}
saturated identically by the vacuum.

Non-negativity of all shell weights makes the sequence $\mathcal T_\rho(r_k^2)$ non-increasing. For a complete partition with bounded increments $\sup_k d_k<\infty$, combining the tail inequality with a Laplace-transform argument gives $\lim_{k\to\infty}\mathcal T_\rho(r_k^2)=0$ and establishes unit normalization $\sum_{k=1}^\infty P_{\rho,k}(\Pi)=1$ without requiring compact support in phase space. Applying Eq.~\eqref{eq:state-tail-main} inductively to any collection of $m$ shells with ordered indices $i_1<\cdots<i_m$ bounds their cumulative weight by the integral of $W_0$ over a central disk of the same total area,
\begin{align}
\sum_{\ell=1}^mP_{\rho,i_\ell}(\Pi)\le1-e^{-\sum_{\ell=1}^m d_{i_\ell}}\le1-e^{-D_m(\Pi)}\ .
\label{eq:subset-sum-main}
\end{align}
Taking the supremum over all index choices of size $m$ proves Eq.~\eqref{eq:general-Lorenz-main}. Pointwise Wigner positivity is not required, since negative values within individual shells are compatible with majorization whenever every integrated shell weight is non-negative. The complete analytical derivation is provided in the Supplemental Material.

\begin{figure}[t]
\centering
\includegraphics[width=\columnwidth]{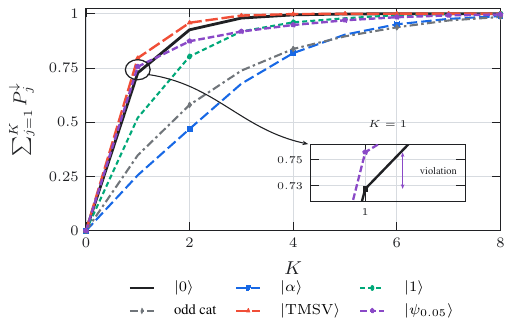}
\caption{Ordered cumulative shell weights $\sum_{j=1}^m P_{\rho,j}^\downarrow$ at resolution $\Delta=1.3$. Coherent ($\alpha=1.5$), single-photon Fock ($\ket1$), and odd cat (amplitude $1.1$) states lie strictly below the vacuum Lorenz curve, with the latter two remaining Wigner-negative yet admitting non-negative shell weights. For the two entangled states discussed below, the shell probabilities of the joint EPR outcomes exceed the vacuum benchmark at $m=1$. The inset resolves the violation for $\ket{\psi_{0.05}}=\sqrt{0.95}\ket{0,0}+\sqrt{0.05}\ket{2,2}$.}
\label{fig:lorenz-comparison}
\end{figure}

The equal-area partition $\Pi_\Delta$ has boundaries $r_k^2=k\Delta$ with resolution parameter $\Delta>0$. With the corresponding shells and weights denoted by $A_k(\Delta):=\{\boldsymbol r:(k-1)\Delta\le|\boldsymbol r|^2<k\Delta\}$ and $P_{\rho,k}(\Delta):=P_{\rho,k}(\Pi_\Delta)$, each shell spans an oscillator-energy window $\Delta/2$. At $\Delta=2$, the shells have area $2\pi$ and unit energy width, corresponding to the Planck--Bohr--Sommerfeld bands \cite{schleich88,lewis16}. For equal-area shells, $D_m(\Pi_\Delta)=m\Delta$, and the vacuum benchmark reduces to the geometric probability vector
\begin{align}
\boldsymbol P_0(\Delta)
=(1-e^{-\Delta})(1,e^{-\Delta},e^{-2\Delta},\ldots)\ .
\label{eq:pvac}
\end{align}
Specializing Theorem~\ref{thm:radial-majorization} to this uniform geometry yields the following corollary.

\Needspace{7\baselineskip}
\begin{corollary}
\label{cor:equal-area}
Let $\rho$ be any single-mode state and let $\Delta>0$. If $P_{\rho,k}(\Delta)\ge0$ for every $k$, then $\boldsymbol P_\rho(\Delta)$ is a normalized probability vector satisfying
\begin{align}
\boldsymbol P_\rho(\Delta)\prec\boldsymbol P_0(\Delta)\ .
\label{eq:main}
\end{align}
Equivalently, for every $m\ge1$,
\begin{align}
\sum_{j=1}^{m}P_{\rho,j}^{\downarrow}(\Delta)
\le 1-e^{-m\Delta}\ .
\end{align}
\end{corollary}

\begingroup\microtypesetup{expansion=true,protrusion=true}
Notably, these shell-majorization relations hold more broadly for weights and vacuum benchmarks evaluated in the same Cahill--Glauber $s$-ordered representation with $s\in[-1,0]$, smoothly connecting the Husimi $Q$ representation ($s=-1$) with the Wigner function ($s=0$). In what follows, we use the Wigner relation in Eq.~\eqref{eq:main}, for which the oscillator frame specifies the shell geometry and $\Delta$ sets the radial resolution.
\par\endgroup

The majorization relation in Eq.~\eqref{eq:main} is strict for every non-vacuum single-mode Gaussian pure state. In continuous phase space, area-preserving affine symplectic transformations render all Gaussian pure states level-equivalent to the vacuum \cite{vanherstraeten23}. Anchoring the radial partition to a fixed oscillator frame makes the integrated weights sensitive to displacement and squeezing, lifting the degeneracy of continuous majorization. Consequently, every non-vacuum state $\ket\psi=D(\alpha)S(\zeta)\ket0$ obeys $\boldsymbol P_\psi(\Delta)\prec\boldsymbol P_0(\Delta)$ for all $\Delta>0$, whereas the converse relation is strictly excluded.
The proof rules out reverse majorization by comparing the large-$k$ decay of Gaussian shell probabilities with the vacuum geometric progression in Eq.~\eqref{eq:pvac}. The coherent-state curve in \cref{fig:lorenz-comparison} illustrates this strict relation.

Optical loss and amplification raise the question of whether the output generated from a vacuum input remains a concentration benchmark. Existing majorization relations for phase-insensitive Gaussian channels compare output spectra \cite{mari14,jabbour16,vanherstraeten23ladder}. Such spectral comparisons are invariant under phase-space displacements, while radial shell weights depend on displacement relative to the fixed oscillator origin. For any single-mode phase-insensitive Gaussian channel $\Phi$, an arbitrary input $\rho$ satisfies the output-majorization relation
\begin{align}
\boldsymbol P_{\Phi(\rho)}(\Delta)\prec \boldsymbol P_{\Phi(\ket0\bra0)}(\Delta)
\end{align}
whenever the output shell weights remain non-negative (see Supplemental Material).

\textit{Wigner negativity scale}---The non-negativity of the shell weights depends on the phase-space resolution. To characterize how far negative shell weights persist under radial coarse graining, we define the radial negativity scale
\begin{align}
\Delta_\rho^*:=\sup\{\Delta>0:\exists\,k\ge 1,\,P_{\rho,k}(\Delta)<0\}\ ,
\label{eq:deltastar}
\end{align}
with $\Delta_\rho^*=0$ if the underlying set is empty. When this scale is finite, all shell weights are non-negative for every $\Delta>\Delta_\rho^*$, although non-negativity may also hold at smaller resolutions. Unlike the continuous negativity volume \cite{kenfack04}, which measures the total negative weight, $\Delta_\rho^*$ characterizes the persistence of negativity under radial averaging. This scale is finite for Gaussian-unitary images of finite Fock superpositions, finite mixtures of such states, and even and odd coherent cat states. It can nevertheless be infinite even for finite-energy states, as shown in the Supplemental Material.

For highly excited Fock states $\ket n$, this scale reveals a universal semiclassical crossover in phase space. Inside the classically allowed region $\xi < \xi_n = 2n+1$, quantum interference produces rapid oscillations with negative values, while the Wigner function decays exponentially in the classically forbidden region. The asymptotic relation $\Delta_n^*/2\sim E_n=n+1/2$ shows that the resolution beyond which all shell weights remain non-negative grows with the oscillator energy. Near the turning boundary, uniform Laguerre--Airy asymptotics approximate the first-shell weight by the Airy integral $\mathcal I_{\rm Ai}(t):=\int_{-\infty}^{t}\mathrm{Ai}(u)\,\dif u$ on the characteristic scale $\ell_n:=2^{-1/3}(4n+2)^{1/3}$ \cite{berry77,hanin20}. Here $t=(\Delta-\xi_n)/\ell_n$ measures the displacement of the first-shell boundary from the classical turning point in units of $\ell_n$. The largest zero of this integral, $t_0\simeq-1.38418$, determines the leading correction to the classical boundary,
\begin{align}
\Delta_n^*=\xi_n+t_0\ell_n+O(n^{-1/3})\ ,\quad n\to\infty\ .
\label{eq:airy-law}
\end{align}
The classical orbital energy thus dictates the leading linear growth, while quantum corrections are governed by universal Airy-edge scaling. As shown in \cref{fig:negativity-airy}, exact numerical evaluations rapidly converge to this asymptotic scaling law.

\begin{figure}[!t]
\centering
\includegraphics[width=\columnwidth]{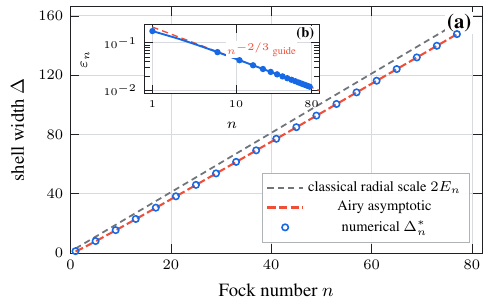}
\caption{Fock-state radial negativity scale $\Delta_n^*$. (a) Numerically evaluated finite-$n$ values of $\Delta_n^*$ (blue circles), the asymptotic Airy expansion in Eq.~\eqref{eq:airy-law} (coral dashed curve), and the classical turning boundary $\xi_n=2n+1$ (gray dashed line). (b) Normalized residual $\varepsilon_n=[\Delta_n^*-(\xi_n+t_0\ell_n)]/\ell_n$, consistent with the predicted $O(n^{-2/3})$ decay.}
\label{fig:negativity-airy}
\end{figure}

\begingroup\microtypesetup{expansion=true,protrusion=true}
\textit{Coarse-grained entanglement criterion}---Combining equal-area majorization with partial transposition yields an entanglement criterion that can be evaluated directly from radially binned joint homodyne outcomes. In a dual-homodyne arrangement, mixing two modes at a balanced beam splitter allows the simultaneous detection of the commuting EPR quadratures $\hat u=(\hat q_1-\hat q_2)/\sqrt2$ and $\hat v=(\hat p_1+\hat p_2)/\sqrt2$ \cite{braunstein98,leonhardt97}. Radial binning of paired outcomes $(u,v)$ into concentric shells $(k-1)\Delta\le u^2+v^2<k\Delta$ directly samples discrete shell probabilities, bypassing any reconstruction of density operators or continuous phase-space distributions.
\par\endgroup

Under partial transposition of the second mode, the joint EPR probability density satisfies $F_\rho(u,v)=W_{\tau_\rho}(u,v)$, where $\tau_\rho=\Tr_+(U_{\rm BS}\rho_{12}^\Gamma U_{\rm BS}^\dagger)$ \cite{peres96,horodecki96}. Here, $\Tr_+$ traces out the symmetric output mode of the balanced beam splitter, leaving $\tau_\rho$ as a single-mode operator. For every separable state $\rho_{12}$, the positivity of the partial transpose ($\rho_{12}^\Gamma\ge0$) ensures that $\tau_\rho$ is a physical density operator. Because $F_\rho$ is a joint measurement probability density, its shell weights are non-negative, and Corollary~\ref{cor:equal-area} gives the following entanglement criterion.

\begin{corollary}
\label{cor:epr-criterion}
For any separable two-mode state $\rho_{12}$ and any resolution $\Delta>0$, the shell-probability vector obtained from the EPR distribution $F_\rho$ satisfies
\begin{align}
\boldsymbol P_{\tau_\rho}(\Delta)\prec\boldsymbol P_0(\Delta)
\ .
\label{eq:epr-criterion}
\end{align}
Any violation of this majorization relation certifies entanglement.
\end{corollary}

For a two-mode squeezed vacuum, both EPR quadratures have variance $\nu=e^{-2\zeta}/2$. The resulting first-shell probability $1-e^{-\Delta/(2\nu)}$ exceeds the vacuum benchmark $1-e^{-\Delta}$ for every $\zeta>0$ and $\Delta>0$, giving an infimum detectable squeezing of $\zeta_{\min}=0$ at every positive shell width. As shown in \cref{fig:epr-nongaussian-benchmarks}(a), the criteria of Tasca et al.\ \cite{tasca13} and G{\"a}rttner et al.\ \cite{garttner23} give positive detection thresholds when evaluated using Cartesian bins of the same phase-space area.

\Needspace{7\baselineskip}
\begingroup\microtypesetup{expansion=true,protrusion=true}
\predisplaypenalty=10000\postdisplaypenalty=10000
The criterion also provides an operational witness for non-Gaussian entanglement. For the Schmidt-correlated state $\ket{\psi_p}=\sqrt{1-p}\ket{0,0}+\sqrt p\ket{2,2}$, the second moments do not violate the Mancini--Giovannetti--Vitali--Tombesi (MGVT) variance criterion \cite{mancini02,giovannetti03}. The EPR marginal distributions satisfy the Saboia--Toscano--Walborn (STW) entropic bound \cite{saboia11} because they coincide with the conjugate-quadrature distributions of the same physical single-mode state (see Supplemental Material). The balanced G{\"a}rttner--Haas--Noll (GHN) criterion detects entanglement only for $p\in(0,16/41)$ \cite{garttner23}. In contrast, violation of the vacuum bound on the first-shell probability certifies entanglement in the region $0<p<p_*(\Delta)$ shown in \cref{fig:epr-nongaussian-benchmarks}(b), with boundary
\begin{align}
p_*(\Delta)=
\frac{(2-\Delta)^2}
{(2-\Delta)^2+\Delta^2(2-\Delta+\Delta^2/4)^2}
\label{eq:schmidt-shell-boundary-main}
\end{align}
for $0<\Delta<2$. Since $p_*(\Delta)\to1$ as $\Delta\to0$, every fixed state with $0<p<1$ is detected at sufficiently small $\Delta$. More generally, for every entangled Schmidt-correlated pure state with finite Fock support, suitable local phase rotations yield a violation of the first-shell bound at sufficiently fine resolution (see Supplemental Material for related R\'enyi shell-entropy and moment inequalities).
\par\endgroup

\begin{figure}[t]
\centering
\includegraphics[width=\columnwidth]{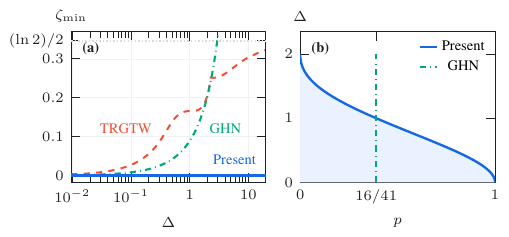}
\caption{EPR entanglement benchmarks. (a) Infimum detectable squeezing versus resolution $\Delta$. The infimum $\zeta_{\min}=0$ (solid blue) persists across all $\Delta$, outperforming the bounds of Tasca et al.\ (TRGTW) \cite{tasca13} (coral dashed) and G{\"a}rttner et al.\ \cite{garttner23} (green dash-dotted, matched area $\Delta=\delta^2/(2\pi)$). Here $\delta$ denotes the square-bin width in the coordinate convention of Ref.~\cite{garttner23}. (b) Detection domain for $\ket{\psi_p}=\sqrt{1-p}\ket{0,0}+\sqrt p\ket{2,2}$. The first-shell boundary in Eq.~\eqref{eq:schmidt-shell-boundary-main} (solid blue) extends well beyond the balanced GHN limit $p=16/41$ (green dash-dotted line).}
\label{fig:epr-nongaussian-benchmarks}
\end{figure}

\textit{Conclusions}---Our counterexamples show that Wigner positivity alone does not guarantee continuous majorization by the vacuum. By integrating Wigner functions over radial oscillator-energy shells, we establish a finite-resolution majorization relation with the vacuum as a universal reference for phase-space concentration, and find that the relation applies to all Wigner-positive states and to Wigner-negative states whose shell weights are non-negative. Notably, the same construction is extendable to the Cahill--Glauber family $s\in[-1,0]$. Although continuous majorization treats the Wigner functions of all single-mode Gaussian pure states as equivalent, the vacuum shell-probability vector strictly majorizes that of every other such state for every $\Delta>0$ in a fixed oscillator frame. For Wigner-negative states, the radial negativity scale characterizes how far negative shell weights persist under coarse graining, complementing the negativity volume, which quantifies the total negative weight of the Wigner function \cite{kenfack04}. For highly excited Fock states, half of the radial negativity scale is asymptotic to the oscillator energy, with an Airy correction. We find that after partial transposition, the majorization relation yields a two-mode entanglement criterion directly evaluated from radially binned joint homodyne outcomes without reconstructing a density matrix or continuous phase-space distribution. Whether the vacuum bound extends to multimode shells of fixed total oscillator energy remains open.

\section*{Acknowledgements}

We appreciate the discussions with Ma-Cheng Yang on this subject. This work was supported in part by the National Natural Science Foundation of China (NSFC) under Grants 12475087 and 12235008.

\setlength{\bibsep}{0pt plus 1pt}
\setlength{\bibsep}{0pt}

\onecolumngrid
\clearpage

\begin{center}
{\large\bfseries Supplemental Material for\\
``Radial Coarse Graining Restores Vacuum Majorization in Wigner Phase Space''}\\[0.75em]
Ao-Xiang Liu$^{1}$ and Cong-Feng Qiao$^{1,2}$\\[0.35em]
{\small $^{1}$School of Physical Sciences, University of Chinese Academy of
Sciences, 1 Yanqihu East Rd, Beijing 101408, China}\\[-0.1em]
{\small $^{2}$ICTP-AP, University of Chinese Academy of Sciences,
Beijing 100190, China}
\end{center}

\setcounter{secnumdepth}{1}
\setcounter{section}{0}
\renewcommand{\thesection}{S\arabic{section}}
\renewcommand{\theHsection}{SM.section.\arabic{section}}
\setcounter{equation}{0}
\renewcommand{\theequation}{S\arabic{equation}}
\renewcommand{\theHequation}{SM.equation.\arabic{equation}}
\setcounter{figure}{0}
\renewcommand{\thefigure}{S\arabic{figure}}
\renewcommand{\theHfigure}{SM.figure.\arabic{figure}}
\setcounter{table}{0}
\renewcommand{\thetable}{S\arabic{table}}
\renewcommand{\theHtable}{SM.table.\arabic{table}}
\setcounter{lemma}{0}
\renewcommand{\thelemma}{S\arabic{lemma}}
\small
\setlength{\abovedisplayskip}{5pt plus 2pt minus 2pt}
\setlength{\belowdisplayskip}{5pt plus 2pt minus 2pt}
\setlength{\abovedisplayshortskip}{3pt plus 2pt}
\setlength{\belowdisplayshortskip}{3pt plus 2pt minus 2pt}
\setlength{\jot}{2pt}

\section{Wigner-positive counterexamples to continuous vacuum majorization}
\label{app:continuous-counterexample}

We construct Wigner-positive states that are incomparable with the vacuum
under continuous majorization. The entropy and positive-part calculations below identify the angular
coherence responsible for the failure, while the final part connects these
states to the radial relation in the main text.

\subsection*{Construction and global Wigner positivity}

The construction begins with a weak coherence
between the vacuum and an $n$-photon Fock state. For a fixed integer $n\ge3$ and
$0<\epsilon<1$, the pure state is
\begin{align}
\ket{\psi_{\epsilon,n}}
=\sqrt{1-\epsilon^2}\ket0+\epsilon\ket n\ .
\label{eq:counterexample-pure}
\end{align}
The Wigner matrix elements make the coherence contribution explicit.
In the convention of \cref{eq:wigner}, the Fock operator
$\ket m\bra k$ has, for $k\ge m$, the transform
\begin{align}
W_{\ket m\bra k}(x,\theta)
&=\frac{e^{-x}}{\pi}(-1)^m\sqrt{\frac{m!}{k!}}
(2x)^{(k-m)/2}L_m^{(k-m)}(2x)e^{i(k-m)\theta}\ ,
\label{eq:counterexample-wigner-matrix}
\end{align}
where $L_m^{(k-m)}$ is an associated Laguerre polynomial
\cite{cahill69a,cahill69b}. The reverse matrix element is its complex
conjugate. For $m=0$, the associated Laguerre polynomial equals one.
Combining the two off-diagonal matrix elements with the populations
$1-\epsilon^2$ and $\epsilon^2$ gives the pure-state Wigner function
in the polar coordinates
$x=|\boldsymbol r|^2$ and $\theta=\arg(q+ip)$, is
\begin{align}
W_{\psi_{\epsilon,n}}(x,\theta)
&=\frac{e^{-x}}{\pi}\Big[(1-\epsilon^2)
+\epsilon^2P_n(x)
+2\epsilon\sqrt{1-\epsilon^2}\,\beta_nx^{n/2}\cos(n\theta)\Big]\ ,
\label{eq:counterexample-pure-wigner}
\end{align}
where $P_n(x):=(-1)^nL_n(2x)$ and
$\beta_n:=2^{n/2}/\sqrt{n!}$. The angular term is the Wigner image of the
coherence $\ket0\bra n+\ket n\bra0$. It may make the pure-state Wigner
function negative, but only in a region that moves into the Gaussian tail as
$\epsilon\to0$.

A phase rotation $e^{-i\phi\hat N}$ multiplies the operator
$\ket0\bra n$ by $e^{in\phi}$. Its Hermitian combination therefore
has angular harmonic $n$, while the diagonal populations are radial.
Here $\hat N=\hat a^\dagger\hat a$ is the number operator.

The negative tail can be removed without changing the leading entropy
effect of the coherence. The thermal state with fixed mean photon number
$\nu>0$ is
\begin{align}
\tau_\nu
&=\sum_{k=0}^\infty\frac{\nu^k}{(\nu+1)^{k+1}}\ket k\bra k\ .
\label{eq:counterexample-thermal-state}
\end{align}
Substituting its diagonal Wigner matrix elements and applying
$\sum_{k\ge0}z^kL_k(w)=(1-z)^{-1}e^{-wz/(1-z)}$ at
$z=-\nu/(\nu+1)$ gives
\begin{align}
W_{\tau_\nu}(x)
&=\frac{e^{-x}}{\pi(\nu+1)}
\sum_{k=0}^\infty\left(-\frac{\nu}{\nu+1}\right)^kL_k(2x)
=\frac{e^{-x/(2\nu+1)}}{\pi(2\nu+1)}\ .
\label{eq:counterexample-thermal-wigner}
\end{align}
We consider the mixed state
\begin{align}
\rho_{\epsilon,n}
&=(1-\delta_\epsilon)
\ket{\psi_{\epsilon,n}}\bra{\psi_{\epsilon,n}}
+\delta_\epsilon\tau_\nu\ .
\label{eq:counterexample-mixed}
\end{align}
The thermal weight $\delta_\epsilon$ can be chosen to make
$W_{\rho_{\epsilon,n}}$ non-negative while remaining smaller than every
power of $\epsilon$. We prove this statement next.

The polynomial $P_n$ has a positive leading coefficient. The finite constant
$M_n:=\max\{0,-\min_{x\ge0}P_n(x)\}$ therefore satisfies
$P_n(x)\ge-M_n$. For $n=3$, one has
$P_3(x)=\tfrac43x^3-6x^2+6x-1$ and $M_3=1+\sqrt3$,
obtained at $x=(3+\sqrt3)/2$.
The exact minimum over the angle is $\pi^{-1}e^{-x}B_\epsilon(x)$,
where
\begin{align}
B_\epsilon(x)
&=1+\epsilon^2[P_n(x)-1]
-2\epsilon\sqrt{1-\epsilon^2}\,\beta_nx^{n/2}\ .
\label{eq:counterexample-worst-bracket}
\end{align}
The polynomial lower bound gives
\begin{align}
(1-\epsilon^2)+\epsilon^2P_n(x)
-2\epsilon\sqrt{1-\epsilon^2}\,\beta_nx^{n/2}
\ge1-(1+M_n)\epsilon^2-2\epsilon\beta_nx^{n/2}\ .
\label{eq:counterexample-core-bound}
\end{align}
For sufficiently small $\epsilon$, the right-hand side is positive whenever
$x<x_-(\epsilon)$, where
\begin{align}
x_-(\epsilon)
&:=\left[
\frac{1-(1+M_n)\epsilon^2}{2\epsilon\beta_n}
\right]^{2/n}\ .
\label{eq:counterexample-xminus}
\end{align}
Possible negativity is consequently confined to $x\ge x_-(\epsilon)$, and
the pure-state Wigner function satisfies the global lower bound
\begin{align}
W_{\psi_{\epsilon,n}}(x,\theta)
&\ge-\frac{2\epsilon\beta_n}{\pi}x^{n/2}e^{-x}\ .
\label{eq:counterexample-pure-lower}
\end{align}

The thermal distribution has a broader Gaussian tail. The parameters
$s_\nu:=2\nu+1$ and $c:=1-s_\nu^{-1}>0$ give a convenient sufficient
thermal weight,
\begin{align}
\delta_\epsilon
&:=K\,2\epsilon\beta_ns_\nu
x_-(\epsilon)^{n/2}e^{-cx_-(\epsilon)}\ ,
\qquad K>1\ .
\label{eq:counterexample-delta}
\end{align}
For small enough $\epsilon$, one has
$x_-(\epsilon)>n/(2c)$, and the function
$x^{n/2}e^{-cx}$ decreases throughout $[x_-(\epsilon),\infty)$. It follows
from \cref{eq:counterexample-delta} that
\begin{align}
\delta_\epsilon\frac{e^{-x/s_\nu}}{s_\nu}
&\ge2\epsilon\beta_nx^{n/2}e^{-x}
\qquad\bigl(x\ge x_-(\epsilon)\bigr)\ .
\label{eq:counterexample-thermal-dominates}
\end{align}
On the inner region, both terms in \cref{eq:counterexample-mixed} are
non-negative. On the outer region,
\cref{eq:counterexample-pure-lower,eq:counterexample-thermal-dominates}
show that the thermal contribution dominates the largest possible negative
value of the pure-state contribution. Hence
$W_{\rho_{\epsilon,n}}(\boldsymbol r)\ge0$ over the full phase space.

The identity
$2\epsilon\beta_nx_-(\epsilon)^{n/2}
=1-(1+M_n)\epsilon^2$ reduces \cref{eq:counterexample-delta} to
\begin{align}
\delta_\epsilon
&=Ks_\nu\bigl[1-(1+M_n)\epsilon^2\bigr]
e^{-cx_-(\epsilon)}\ .
\label{eq:counterexample-delta-reduced}
\end{align}
Since $x_-(\epsilon)\to\infty$, this choice satisfies
$\delta_\epsilon\to0$. In particular, $0<\delta_\epsilon<1$ for sufficiently
small $\epsilon$, and \cref{eq:counterexample-mixed} defines a valid quantum
state.
Since $x_-(\epsilon)\asymp\epsilon^{-2/n}$, there is a constant $c'>0$
such that
\begin{align}
\delta_\epsilon
&=O\!\left(e^{-c'\epsilon^{-2/n}}\right)
=o(\epsilon^m)
\qquad\text{for every fixed }m>0\ .
\label{eq:counterexample-delta-asymptotic}
\end{align}
Moreover, $|\ln\delta_\epsilon|=O(\epsilon^{-2/n})$, which implies
$\delta_\epsilon|\ln\delta_\epsilon|=o(\epsilon^m)$ for every fixed
$m>0$. The strict margin $K>1$ also gives
$W_{\rho_{\epsilon,n}}\ge(1-K^{-1})\delta_\epsilon W_{\tau_\nu}>0$.
On the core this follows from positivity of the pure component.
On the exterior, the chosen thermal term is at least $K$ times the
negative envelope in \cref{eq:counterexample-pure-lower}.

For finite parameters, the least admissible thermal weight can be
expressed without the envelope approximation. With
$d_\epsilon(x)=\pi^{-1}e^{-x}[-B_\epsilon(x)]_+$,
the worst-angle condition is
$\delta W_{\tau_\nu}(x)\ge(1-\delta)d_\epsilon(x)$.
It holds at every radius precisely when
\begin{align}
\delta\ge\delta_{\min}(\epsilon;n,\nu)
&=\sup_{x\ge0}
\frac{d_\epsilon(x)}{W_{\tau_\nu}(x)+d_\epsilon(x)}\ .
\label{eq:counterexample-minimal-weight}
\end{align}
The ratio vanishes when the pure Wigner function is non-negative at
every angle. The sufficient choice in \cref{eq:counterexample-delta}
bounds this minimum from above, and is not assumed to attain it.
This distinction allows a smaller weight to be verified directly at a
finite parameter point.

\Cref{fig:counterexample-phase-space-app} shows the angular coherence and a
finite Wigner-positive member of the family with $n=3$.
\begin{figure}[tb]
\centering
\begin{tabular}{@{}c@{\hspace{0.04\textwidth}}c@{}}
\textbf{(a)} & \textbf{(b)}\\[-2pt]
\includegraphics[height=0.23\textwidth]{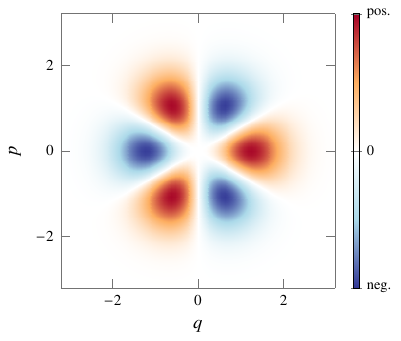}
&
\includegraphics[height=0.23\textwidth]{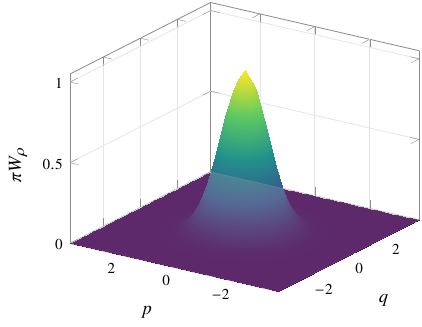}
\end{tabular}
\caption{\label{fig:counterexample-phase-space-app}
Phase-space structure of the $n=3$ counterexample. (a) Interference
contribution proportional to $e^{-x}x^{3/2}\cos(3\theta)$. (b) The
Wigner-positive density $\pi W_{\rho}(q,p)$ for $\epsilon=10^{-2}$,
$\nu=3$, and $\delta=5.16\times10^{-6}$. The grey plane marks $W=0$.}
\end{figure}

\subsection*{Shannon entropy and control of the full phase space}

The entropy decrease must be established for the repaired, non-negative
Wigner function. In particular, the small thermal weight alone does not
justify expanding the logarithm over the full plane, since the ratio of the
thermal distribution to the vacuum grows exponentially with $x$. We first
control this ratio on a growing disk and then estimate the exterior without
a Taylor expansion.

The normalization of $W$ and $W_0=\pi^{-1}e^{-x}$ gives the exact identity
\begin{align}
h(W)-h(W_0)
&=\langle x\rangle_W-1-D(W\Vert W_0)\ ,
\label{eq:counterexample-entropy-identity}
\end{align}
where $h(W)=-\int W\ln W\,\dif^2\boldsymbol r$ and
$D(W\Vert W_0)=\int W\ln(W/W_0)\,\dif^2\boldsymbol r$.
Indeed, substituting $\ln W_0=-x-\ln\pi$ in the definition of $D$
gives \cref{eq:counterexample-entropy-identity}. The moment relation
$\langle x\rangle_W=1+2\Tr(\rho\hat N)$, with
$\hat N=\hat a^\dagger\hat a$, makes the first term explicit for our family,
\begin{align}
\langle x\rangle_{W_{\rho_{\epsilon,n}}}-1
&=2(1-\delta_\epsilon)n\epsilon^2+2\delta_\epsilon\nu
=2n\epsilon^2+o(\epsilon^2)\ .
\label{eq:counterexample-radius-moment}
\end{align}
The same identity follows directly by integrating the finite Laguerre
expansion, which gives
$\int_0^\infty e^{-x}P_n(x)\,\dif x=1$ and
$\int_0^\infty xe^{-x}P_n(x)\,\dif x=2n+1$.

The relative perturbation is
$r_\epsilon=W_{\rho_{\epsilon,n}}/W_0-1$. The scalar function
$\varphi(r)=(1+r)\ln(1+r)-r$, with $\varphi(-1)=1$, is non-negative
for $r\ge-1$. Since $\int W_0r_\epsilon\,\dif^2\boldsymbol r=0$,
\begin{align}
D(W_{\rho_{\epsilon,n}}\Vert W_0)
&=\int W_0\varphi(r_\epsilon)\,\dif^2\boldsymbol r\ .
\label{eq:counterexample-relative-entropy}
\end{align}
For fixed $0<\gamma<2/n$, the growing disk
$x\le X_\epsilon=\epsilon^{-\gamma}$ lies inside the positive core for
sufficiently small $\epsilon$. The thermal correction satisfies
\begin{align}
\sup_{x\le X_\epsilon}
\delta_\epsilon\frac{W_{\tau_\nu}}{W_0}
&\le\frac{\delta_\epsilon}{s_\nu}e^{cX_\epsilon}
=o(\epsilon^m)
\qquad\text{for every fixed }m>0\ ,
\label{eq:counterexample-core-thermal}
\end{align}
because $X_\epsilon=o(\epsilon^{-2/n})$ and
\cref{eq:counterexample-delta-asymptotic} controls the negative exponential
in $\delta_\epsilon$. Expanding only the square-root amplitude in
\cref{eq:counterexample-pure-wigner} now gives
\begin{align}
\frac{W_{\rho_{\epsilon,n}}}{W_0}
&=1+\epsilon u_n(x)\cos(n\theta)
+\epsilon^2[P_n(x)-1]+R_\epsilon(x,\theta)\ ,
\label{eq:counterexample-core-expansion}
\end{align}
where $u_n(x)=2\beta_nx^{n/2}$. The coherence remainder has size
$O(\epsilon^3x^{n/2})$, while all thermal corrections are smaller than
every fixed power of $\epsilon$ on this disk. Consequently,
\begin{align}
\sup_{x\le X_\epsilon}|R_\epsilon|
&=O(\epsilon^{3-\gamma n/2})+o(\epsilon^m)
=o(\epsilon^2)\ .
\label{eq:counterexample-core-remainder}
\end{align}
Here and below $m$ may be chosen larger than two. The relative perturbation
tends uniformly to zero because both
$\epsilon X_\epsilon^{n/2}$ and $\epsilon^2X_\epsilon^n$ tend to zero.
Gaussian integrability of the polynomial factors also gives
$\int_{x\le X_\epsilon}W_0r_\epsilon^2\,\dif^2\boldsymbol r
=O(\epsilon^2)$.

On this disk, $\varphi(r_\epsilon)=r_\epsilon^2/2+O(|r_\epsilon|^3)$.
The integrated cubic remainder is controlled by
\begin{align}
\int_{x\le X_\epsilon}W_0|r_\epsilon|^3\,\dif^2\boldsymbol r
&\le\sup_{x\le X_\epsilon}|r_\epsilon|
\int_{x\le X_\epsilon}W_0r_\epsilon^2\,\dif^2\boldsymbol r
=o(\epsilon^2)\ .
\label{eq:counterexample-cubic-remainder}
\end{align}
The angular cross term between the coherence and the radial population
change vanishes. The population square contributes $O(\epsilon^4)$,
and the remaining quadratic moment is
\begin{align}
\frac12\int W_0u_n(x)^2\cos^2(n\theta)\,\dif^2\boldsymbol r
&=\beta_n^2\int_0^\infty x^ne^{-x}\,\dif x=2^n\ .
\label{eq:counterexample-shannon-coherence}
\end{align}
Its omitted tail is a polynomial in $X_\epsilon$ times
$e^{-X_\epsilon}$. Therefore the disk contribution to
\cref{eq:counterexample-relative-entropy} is
$2^n\epsilon^2+o(\epsilon^2)$.

To bound the exterior, the explicit Wigner functions give
\begin{align}
0\le W_{\rho_{\epsilon,n}}(x,\theta)
&\le C_ne^{-x}(1+x^n)
+C_\nu\delta_\epsilon e^{-x/s_\nu}\ ,
\label{eq:counterexample-tail-bound}
\end{align}
with constants independent of $\epsilon$. In terms of the density ratio
$y=W_{\rho_{\epsilon,n}}/W_0$, this implies
$0\le y\le A(x)=C_n(1+x^n)+C_\nu\delta_\epsilon e^{cx}$.
For $0\le y\le1$, one has $y|\ln y|\le e^{-1}$, while for
$1\le y\le A$ one has $y\ln y\le A\ln(1+A)$.
Since $\delta_\epsilon\le1$, the logarithmic factor obeys
$\ln(1+A(x))\le C_{n,\nu}(1+x)$ after increasing the constant.
It follows that, for a fixed polynomial $\mathcal P$ with non-negative
coefficients,
\begin{align}
0\le W_0\varphi(r_\epsilon)
&\le C e^{-x}\mathcal P(x)
+C\delta_\epsilon e^{-x/s_\nu}\mathcal P(x)\ .
\label{eq:counterexample-shannon-tail-pointwise}
\end{align}
The affine terms in $\varphi$ satisfy the same bound. Repeated integration
by parts of each monomial in $\mathcal P$ gives
\begin{align}
\int_{x>X_\epsilon}W_0\varphi(r_\epsilon)\,\dif^2\boldsymbol r
&\le C e^{-X_\epsilon}\mathcal P_1(X_\epsilon)
+C\delta_\epsilon e^{-X_\epsilon/s_\nu}\mathcal P_2(X_\epsilon)
=o(\epsilon^2)\ ,
\label{eq:counterexample-shannon-tail}
\end{align}
where $\mathcal P_1$ and $\mathcal P_2$ are fixed polynomials.
These estimates also establish finiteness of the entropy and relative
entropy used above. Combining the disk and exterior contributions gives
\begin{align}
D(W_{\rho_{\epsilon,n}}\Vert W_0)
&=2^n\epsilon^2+o(\epsilon^2)\ .
\label{eq:counterexample-shannon-relative-expansion}
\end{align}
Substitution into \cref{eq:counterexample-entropy-identity} proves
\begin{align}
h(W_{\rho_{\epsilon,n}})-h(W_0)
&=(2n-2^n)\epsilon^2+o(\epsilon^2)<0
\qquad(n\ge3)\ .
\label{eq:counterexample-shannon-expansion}
\end{align}
The population change contributes $2n\epsilon^2$ to the mean squared
radius, whereas the angular coherence contributes $2^n\epsilon^2$ to
the relative entropy. The latter contribution is larger for $n\ge3$.

The zero quadratic coefficients at $n=1,2$ agree with the tangent
directions of the Gaussian orbit through the vacuum. The linear generators
$\hat q$ and $\hat p$ couple the vacuum to $\ket1$, while the quadratic
generators $\hat q^2-\hat p^2$ and
$\hat q\hat p+\hat p\hat q$ couple it to $\ket2$.
Exact displacement and squeezing preserve the continuous Wigner entropy.
This tangent-space statement does not make the truncated superpositions
Gaussian and does not determine their higher-order entropy response.
Thus $n=3$ is the lowest sector with a negative quadratic coefficient in
\cref{eq:counterexample-shannon-expansion}.

\subsection*{A finite Wigner-positive counterexample}

The asymptotic construction is complemented by the finite state shown in
\cref{fig:counterexample-phase-space-app}. Its parameters are
\begin{align}
n=3\ ,\qquad \epsilon=\frac1{100}\ ,\qquad
\nu=3\ ,\qquad \delta=\frac{129}{25000000}\ .
\label{eq:counterexample-certified-parameters}
\end{align}
Here $\delta$ is a fixed thermal weight, rather than the sufficient
asymptotic choice in \cref{eq:counterexample-delta}. Global positivity and
the entropy sign are verified separately at this point.

The angle dependence can be written as
\begin{align}
\pi W_\rho(x,\theta)
&=P_\pi(x)+Q_\pi(x)\cos(3\theta)\ ,
\label{eq:counterexample-certified-density}
\end{align}
where the radial part is
\begin{align}
P_\pi(x)
&=(1-\delta)e^{-x}[1+\epsilon^2(P_3(x)-1)]
+\frac{\delta}{7}e^{-x/7}\ ,
\label{eq:counterexample-certified-radial}
\end{align}
and the non-negative angular amplitude is
\begin{align}
Q_\pi(x)
&=(1-\delta)2\epsilon\sqrt{1-\epsilon^2}\,
\beta_3x^{3/2}e^{-x}\ .
\label{eq:counterexample-certified-angular}
\end{align}
Positivity reduces to $P_\pi(x)-Q_\pi(x)>0$ for every $x\ge0$.
Removing the common Gaussian decay gives the equivalent function
\begin{align}
H(x)
&=e^x[P_\pi(x)-Q_\pi(x)]
=(1-\delta)B_\epsilon(x)+\frac{\delta}{7}e^{6x/7}\ .
\label{eq:counterexample-certified-H}
\end{align}
Arb interval arithmetic through \texttt{python-flint}~0.8.0, with
$256$-bit working precision, encloses $H$ on the $384$ intervals
$[j/16,(j+1)/16]$, $j=0,\ldots,383$. The smallest certified lower
endpoint exceeds $0.150224$. This bounds the function on each complete
interval, rather than only at sampled points.

For $x\ge24$, the polynomial $P_3(x)$ is positive.
For example, $P_3(24)>0$ and $P_3'(x)=4x^2-12x+6>0$ on this
half-line. The radial pure-state contribution to $P_\pi$ is therefore
non-negative, and
\begin{align}
\frac{Q_\pi(x)}{P_\pi(x)}
&\le\frac{14\epsilon\beta_3}{\delta}
x^{3/2}e^{-6x/7}
\le0.004288<1\ .
\label{eq:counterexample-certified-tail}
\end{align}
The last bound is evaluated at $x=24$, since
$x^{3/2}e^{-6x/7}$ decreases for $x\ge7/4$.
Together with the interval calculation, it proves $W_\rho>0$ over
the full phase space.

The entropy can be certified without truncating an entropy integral of
unknown sign. For $r\ge-1$, the scalar function in
\cref{eq:counterexample-relative-entropy} obeys
\begin{align}
\varphi(r)
&\ge\frac{r^2}{2}-\frac{|r|^3}{6}\ .
\label{eq:counterexample-cubic-lower}
\end{align}
For $r\ge0$, the second derivative of the difference between the two
sides is $r^2/(1+r)\ge0$, and its value and first derivative vanish at
zero. For $-1<r<0$, the stronger inequality
$\varphi(r)\ge r^2/2$ follows by the same argument, since the second
derivative of $\varphi(r)-r^2/2$ is $-r/(1+r)>0$.
Continuity includes $r=-1$.

For the state in \cref{eq:counterexample-certified-parameters},
$r=W_\rho/W_0-1$ and the disk integral
\begin{align}
J_X
&=\int_{x\le X}W_0
\left(\frac{r^2}{2}-\frac{|r|^3}{6}\right)
\dif^2\boldsymbol r
\label{eq:counterexample-certified-core-integral}
\end{align}
is a lower bound on $D(W_\rho\Vert W_0)$, because
$\varphi\ge0$ on the exterior. At $X=10$, interval integration on
$400$ radial intervals of width $1/40$ and $360$ angular intervals of
width $\pi/180$ gives the bounds in
\cref{tab:counterexample-certificate}. The integration uses
$\dif^2\boldsymbol r=\tfrac12\dif x\,\dif\theta$ and encloses the
integrand on every rectangle. The exact mean-radius shift is
$\langle x\rangle_W-1=0.000630956904$.

\begin{table}[tb]
\centering
\caption{Verified bounds for the finite counterexample in
\cref{eq:counterexample-certified-parameters}. Decimal inequalities
are rounded outward. The last row combines the preceding entropy bound
with the exact mean-radius shift.}
\label{tab:counterexample-certificate}
\begin{tabular}{ll}
\toprule
Quantity & Verified bound\\
\midrule
$H$ on $[0,24]$ & $H>0.150224$\\
$Q_\pi/P_\pi$ on $[24,\infty)$ & $Q_\pi/P_\pi<0.004288$\\
$1+r$ on $x\le10$ & $1+r>0.351666$\\
$J_{10}$ & $J_{10}>0.0006968757138$\\
$h(W_\rho)-h(W_0)$ & $h(W_\rho)-h(W_0)<-0.0000659188$\\
\bottomrule
\end{tabular}
\end{table}

For an independent numerical evaluation, the angular integral can be
performed exactly. With $P=P_\pi/\pi$, $Q=Q_\pi/\pi$,
$b=Q/P$ and $g=\sqrt{1-b^2}$, positivity gives $|b|<1$.
The factorization
\begin{align}
1+b\cos\phi
&=\frac{1+g}{2}\left|1+\frac{b}{1+g}e^{i\phi}\right|^2
\label{eq:counterexample-angular-factorization}
\end{align}
gives the angular averages
$\langle\ln(1+b\cos\phi)\rangle_\phi=\ln[(1+g)/2]$ and
$\langle\cos\phi\ln(1+b\cos\phi)\rangle_\phi=b/(1+g)$
by expanding the logarithm in its convergent Fourier series.
The replacement $\phi=3\theta$ covers three identical periods.
The entropy therefore reduces to
\begin{align}
h(W_\rho)
&=-\pi\int_0^\infty P(x)
\left[\ln P(x)+\ln\frac{1+g(x)}{2}+1-g(x)\right]\dif x\ .
\label{eq:counterexample-angular-entropy}
\end{align}
High-precision radial quadrature gives
$h(W_\rho)-h(W_0)\simeq-1.7427877032\times10^{-4}$.
This value is a numerical estimate, whereas
\cref{tab:counterexample-certificate} already establishes the strict
entropy inequality with controlled errors.

\subsection*{An entropy bound at finite parameters}

An explicit remainder bound connects the quadratic mechanism to a finite
parameter range. In this subsection $n\ge3$ and $\nu>0$ are fixed, and
$\delta$ may be any weight for which the state in
\cref{eq:counterexample-mixed} is Wigner-positive. The pure relative
perturbation is
\begin{align}
a(x,\theta)
&=\frac{W_{\psi_{\epsilon,n}}}{W_0}-1
=\epsilon\sqrt{1-\epsilon^2}\,u_n(x)\cos(n\theta)
+\epsilon^2[P_n(x)-1]\ .
\label{eq:counterexample-effective-a}
\end{align}
The repaired ratio is $r=a+\delta b$, where
$b=(W_{\tau_\nu}-W_{\psi_{\epsilon,n}})/W_0$.
For a fixed $0<\vartheta<1$, the region
$\Omega_\vartheta=\{1+a\ge\vartheta\}$ contains the disk
$x\le x_\vartheta$, with
\begin{align}
x_\vartheta
&=\left[\frac{1-\vartheta-(1+M_n)\epsilon^2}
{2\beta_n\epsilon}\right]^{2/n}\ ,
\label{eq:counterexample-effective-cutoff}
\end{align}
provided $(1+M_n)\epsilon^2<1-\vartheta$.
The Gaussian tail fraction needed below is
\begin{align}
\mathcal T_n(X)
&=\frac1{n!}\int_X^\infty x^ne^{-x}\,\dif x
=e^{-X}\sum_{k=0}^n\frac{X^k}{k!}\ .
\label{eq:counterexample-effective-tail}
\end{align}

All constants in the bound are finite Gaussian moments. With
$F_n(x)=P_n(x)-1$, the population variance is
$\mathcal V_n=\int_0^\infty e^{-x}F_n(x)^2\,\dif x$.
The cubic coherence coefficient is
\begin{align}
\kappa_n
&=\frac{16\beta_n^3}{9\pi}
\Gamma\!\left(\frac{3n}{2}+1\right)\ .
\label{eq:counterexample-effective-kappa}
\end{align}
The cubic terms containing two coherence factors are controlled by
\begin{align}
\mathcal J_{4,n}
&=2\beta_n^2\int_0^\infty e^{-x}x^n|F_n(x)|\,\dif x\ .
\label{eq:counterexample-effective-J4}
\end{align}
The terms containing one coherence factor involve
\begin{align}
\mathcal J_{5,n}
&=\frac{4\beta_n}{\pi}\int_0^\infty
 e^{-x}x^{n/2}F_n(x)^2\,\dif x\ ,
\label{eq:counterexample-effective-J5}
\end{align}
while the population term uses
\begin{align}
\mathcal J_{6,n}
&=\int_0^\infty e^{-x}|F_n(x)|^3\,\dif x\ .
\label{eq:counterexample-effective-J6}
\end{align}
The first absolute moments are
$g_{1,n}=(4\beta_n/\pi)\Gamma(n/2+1)$,
$g_{2,n}=\int_0^\infty e^{-x}|F_n(x)|\,\dif x$, and
$g_{2,n}^{\tau}=s_\nu^{-1}\int_0^\infty
e^{-x/s_\nu}|F_n(x)|\,\dif x$.

In this notation, every admissible state with
$(1+M_n)\epsilon^2<1-\vartheta$ obeys
\begin{align}
h(W_\rho)-h(W_0)
&\le(2n-2^n)\epsilon^2+\kappa_n\epsilon^3
+2^{n+1}\epsilon^2\mathcal T_n(x_\vartheta)
+R_n(\epsilon)
+\delta[2\nu+\Lambda_n(\epsilon,\vartheta)]\ .
\label{eq:counterexample-effective-bound}
\end{align}
The population and mixed remainders are collected in
\begin{align}
R_n(\epsilon)
&=\left(2^n+\frac{\mathcal V_n}{2}
+\frac{\mathcal J_{4,n}}{2}\right)\epsilon^4
+\frac{\mathcal J_{5,n}}{2}\epsilon^5
+\frac{\mathcal J_{6,n}}{6}\epsilon^6\ ,
\label{eq:counterexample-effective-R}
\end{align}
while the cost of thermal mixing is bounded by
\begin{align}
\Lambda_n(\epsilon,\vartheta)
&=g_{1,n}(1+s_\nu^{n/2})\epsilon
+(g_{2,n}+g_{2,n}^{\tau}+2^{n+1})\epsilon^2
+\mathcal V_n\epsilon^4
\nonumber\\
&\quad+\ln(1/\vartheta)
(2+g_{1,n}\epsilon+g_{2,n}\epsilon^2)\ .
\label{eq:counterexample-effective-Lambda}
\end{align}

To prove the bound, non-negativity of $\varphi$ allows its integral to be
restricted to $\Omega_\vartheta$. Convexity then gives
\begin{align}
\varphi(a+\delta b)
&\ge\varphi(a)+\delta b\ln(1+a)
\qquad\text{on }\Omega_\vartheta\ .
\label{eq:counterexample-effective-tangent}
\end{align}
The logarithm satisfies
$|\ln(1+a)|\le |a|+\ln(1/\vartheta)$ on this region, and
$|b|\le W_{\tau_\nu}/W_0+1+|a|$.
Consequently,
\begin{align}
\left|\int_{\Omega_\vartheta}W_0b\ln(1+a)\,\dif^2\boldsymbol r\right|
&\le\int W_{\tau_\nu}|a|\,\dif^2\boldsymbol r
+\int W_0|a|\,\dif^2\boldsymbol r
+\int W_0a^2\,\dif^2\boldsymbol r
\nonumber\\
&\quad+\ln(1/\vartheta)
\left(2+\int W_0|a|\,\dif^2\boldsymbol r\right)
\le\Lambda_n(\epsilon,\vartheta)\ .
\label{eq:counterexample-effective-mixing}
\end{align}
The last step uses
$\int W_0|a|\le g_{1,n}\epsilon+g_{2,n}\epsilon^2$ and
$\int W_{\tau_\nu}|a|
\le g_{1,n}s_\nu^{n/2}\epsilon+g_{2,n}^{\tau}\epsilon^2$.
The quadratic integral is exact,
\begin{align}
\int W_0a^2\,\dif^2\boldsymbol r
&=2^{n+1}\epsilon^2(1-\epsilon^2)
+\mathcal V_n\epsilon^4\ ,
\label{eq:counterexample-effective-square}
\end{align}
since the angular cross term vanishes.

The cubic lower bound in \cref{eq:counterexample-cubic-lower} applies
to $a$ throughout $\Omega_\vartheta$. Its integral satisfies
\begin{align}
\int_{\Omega_\vartheta}W_0\varphi(a)\,\dif^2\boldsymbol r
&\ge\frac12\int W_0a^2\,\dif^2\boldsymbol r
-\frac12\int_{x>x_\vartheta}W_0a^2\,\dif^2\boldsymbol r
-\frac16\int W_0|a|^3\,\dif^2\boldsymbol r\ .
\label{eq:counterexample-effective-restriction}
\end{align}
The missing quadratic contribution is bounded using
$(u+v)^2\le2u^2+2v^2$,
\begin{align}
\frac12\int_{x>x_\vartheta}W_0a^2\,\dif^2\boldsymbol r
&\le2^{n+1}\epsilon^2\mathcal T_n(x_\vartheta)
+\mathcal V_n\epsilon^4\ .
\label{eq:counterexample-effective-square-tail}
\end{align}
Expanding
$|a|^3\le[\epsilon u_n|\cos(n\theta)|+\epsilon^2|F_n|]^3$
and using the angular averages
$\langle|\cos(n\theta)|\rangle=2/\pi$,
$\langle\cos^2(n\theta)\rangle=1/2$, and
$\langle|\cos(n\theta)|^3\rangle=4/(3\pi)$ gives
\begin{align}
\frac16\int W_0|a|^3\,\dif^2\boldsymbol r
&\le\kappa_n\epsilon^3
+\frac{\mathcal J_{4,n}}2\epsilon^4
+\frac{\mathcal J_{5,n}}2\epsilon^5
+\frac{\mathcal J_{6,n}}6\epsilon^6\ .
\label{eq:counterexample-effective-cubic}
\end{align}
Combining these estimates yields a lower bound on $D(W_\rho\Vert W_0)$.
The identity in \cref{eq:counterexample-entropy-identity} and
$\langle x\rangle_W-1\le2n\epsilon^2+2\nu\delta$ then give
\cref{eq:counterexample-effective-bound}.

For $n=3$ and $\nu=3$, the constants can be evaluated without numerical
quadrature. Here $\mathcal V_3=244$,
$\kappa_3=140\sqrt3/(3\sqrt\pi)$, and
$g_{1,3}=2\sqrt3/\sqrt\pi$.
The polynomial $F_3(x)=\tfrac43x^3-6x^2+6x-2$ has a unique positive
root $x_r=3.2611666966796562484\ldots$.
Splitting the absolute-value integrals at this root reduces all remaining
moments to polynomial Gaussian integrals. In particular,
\begin{align}
\int_R^\infty x^k e^{-a x}\,\dif x
&=e^{-aR}\sum_{j=0}^k
\frac{k!}{j!}\frac{R^j}{a^{k-j+1}}
\qquad(a>0)\ ,
\label{eq:counterexample-polynomial-tail}
\end{align}
and the half-integer moments in $\mathcal J_{5,3}$ are Gamma values.
Rational enclosure of $x_r$ and outward-rounded evaluation give
$\kappa_3<45.602902$, $g_{1,3}<1.954411$,
$g_{2,3}<2.415356$, $g_{2,3}^{\tau}<2197.356553$,
$\mathcal J_{4,3}<1025.388$,
$\mathcal J_{5,3}<9638.174$, and
$\mathcal J_{6,3}<173404.860$.

At the parameters in \cref{eq:counterexample-certified-parameters},
choosing $\vartheta=1/4$ gives $x_\vartheta=10.175689\ldots$.
Insertion in \cref{eq:counterexample-effective-bound} yields the strict
upper bound
\begin{align}
h(W_\rho)-h(W_0)&<-8.44\times10^{-5}\ .
\label{eq:counterexample-effective-finite}
\end{align}
This establishes the entropy sign using the analytic remainder bound,
independently of the disk integration in
\cref{tab:counterexample-certificate}.

The same estimate provides a finite parameter window. For $n=3$,
$\nu=3$, $\vartheta=1/4$, and $\delta\le6\epsilon^3$, replacing
$\delta$ by $6\epsilon^3$ bounds the last term from above.
After division by $\epsilon^2$, every positive contribution in
\cref{eq:counterexample-effective-bound} is nondecreasing in $\epsilon$.
For the tail term this follows because $x_\vartheta$ decreases and
$\mathcal T_n'(X)=-e^{-X}X^n/n!<0$.
The upper bound at $\epsilon=1/80$ is less than
$-3.08\times10^{-5}$. Thus
\begin{align}
h(W_\rho)<h(W_0)
\quad\text{if}\quad
0<\epsilon\le\frac1{80}\ ,\quad
0\le\delta\le6\epsilon^3\ ,\quad W_\rho\ge0\ .
\label{eq:counterexample-effective-window}
\end{align}
Wigner positivity remains a separate hypothesis in this window.
It holds at the certified point, and admissible weights exist for all
sufficiently small $\epsilon$ by
\cref{eq:counterexample-delta-asymptotic}.

\subsection*{R\'enyi entropy and the order-two boundary}

We now evaluate the Wigner--R\'enyi entropy of the repaired state. For a
non-negative Wigner function, the order-$\alpha$ entropy is determined by
\begin{align}
J_\alpha(W)
&:=\int_{\mathbb R^2}W(\boldsymbol r)^\alpha\,\dif^2\boldsymbol r\ ,
\qquad
h_\alpha(W):=\frac{\ln J_\alpha(W)}{1-\alpha}\ .
\label{eq:counterexample-renyi-definition}
\end{align}
The vacuum values are
\begin{align}
J_\alpha(W_0)
&=\frac{\pi^{1-\alpha}}{\alpha}\ ,
\qquad
h_\alpha(W_0)
=\ln\pi+\frac{\ln\alpha}{\alpha-1}\ .
\label{eq:counterexample-vacuum-renyi}
\end{align}
The Shannon value at $\alpha=1$ is understood by continuity.

The expansion in \cref{eq:counterexample-core-expansion} applies
on $x\le X_\epsilon$ for each fixed $\alpha>0$.
Its remainder remains integrable against the Gaussian weight used below.

The normalized Gaussian measure associated with $W_0^\alpha$ is
\begin{align}
\dif\mu_\alpha
&:=\frac{W_0^\alpha}{J_\alpha(W_0)}\,
\dif^2\boldsymbol r
=\frac{\alpha}{2\pi}e^{-\alpha x}\,\dif x\,\dif\theta\ .
\label{eq:counterexample-gaussian-measure}
\end{align}
Taylor expansion gives
$(1+r_\epsilon)^\alpha=1+\alpha r_\epsilon
+\alpha(\alpha-1)r_\epsilon^2/2+O_\alpha(|r_\epsilon|^3)$
on the growing disk. Gaussian moments bound the integrated quadratic
term by $O(\epsilon^2)$, and its cubic remainder is
$o(\epsilon^2)$ by the argument in
\cref{eq:counterexample-cubic-remainder}, now with the measure
$\mu_\alpha$. The population term enters linearly and the coherence
enters quadratically, since its first angular moment vanishes.
After controlling the exterior as below, this gives
\begin{align}
\frac{J_\alpha(W_{\rho_{\epsilon,n}})}{J_\alpha(W_0)}
&=1+c_n(\alpha)\epsilon^2+o(\epsilon^2)\ ,
\label{eq:counterexample-J-expansion}
\end{align}
where only two Gaussian moments enter. They are
\begin{align}
\bigl\langle P_n-1\bigr\rangle_{\mu_\alpha}
&=\left(\frac{2-\alpha}{\alpha}\right)^n-1\ ,
\label{eq:counterexample-population-moment}
\end{align}
and
\begin{align}
\bigl\langle u_n^2\cos^2(n\theta)\bigr\rangle_{\mu_\alpha}
&=\frac{2^{n+1}}{\alpha^n}\ .
\label{eq:counterexample-coherence-moment}
\end{align}
The first identity follows from
$\int_0^\infty e^{-\alpha x}L_n(2x)\,\dif x
=(\alpha-2)^n/\alpha^{n+1}$, and the second uses
$\int_0^\infty x^ne^{-\alpha x}\,\dif x=n!/\alpha^{n+1}$.
Their substitution gives
\begin{align}
c_n(\alpha)
&=\alpha^{1-n}\left[(2-\alpha)^n
+2^n(\alpha-1)\right]-\alpha\ .
\label{eq:counterexample-cn}
\end{align}

For completeness, the passage from the growing core to the full phase space
is uniform at every fixed $\alpha>0$. The envelope in
\cref{eq:counterexample-tail-bound} controls the repaired density.
For non-negative $a$ and $b$,
$(a+b)^\alpha\le C_\alpha(a^\alpha+b^\alpha)$, with
$C_\alpha=1$ for $0<\alpha\le1$ and
$C_\alpha=2^{\alpha-1}$ for $\alpha>1$. Integration of
\cref{eq:counterexample-tail-bound} over $x>X_\epsilon$ yields
\begin{align}
\int_{x>X_\epsilon}W_{\rho_{\epsilon,n}}^\alpha
\,\dif^2\boldsymbol r
&\le C_{\alpha,n}e^{-\alpha X_\epsilon/2}
+C_{\alpha,\nu}\delta_\epsilon^\alpha
e^{-\alpha X_\epsilon/s_\nu}
=o(\epsilon^m)
\label{eq:counterexample-tail-integral}
\end{align}
for every fixed $m>0$. The vacuum tail satisfies the same conclusion.
This proves the full-space expansion in
\cref{eq:counterexample-J-expansion}, including for $0<\alpha<1$.

For $\alpha\ne1$, taking the logarithm of
\cref{eq:counterexample-J-expansion} gives
\begin{align}
h_\alpha(W_{\rho_{\epsilon,n}})-h_\alpha(W_0)
&=K_n(\alpha)\epsilon^2+o(\epsilon^2)\ ,
\qquad
K_n(\alpha):=\frac{c_n(\alpha)}{1-\alpha}\ .
\label{eq:counterexample-renyi-expansion}
\end{align}
The coefficient has the continuous factorization
\begin{align}
K_n(\alpha)
&=2\left[
\sum_{k=0}^{n-1}t_\alpha^k-(1+t_\alpha)^{n-1}
\right]\ ,
\qquad
t_\alpha:=\frac{2-\alpha}{\alpha}\ .
\label{eq:counterexample-K-factorization}
\end{align}
When $0<\alpha<2$, $t_\alpha>0$, and the binomial theorem gives
\begin{align}
K_n(\alpha)
&=-2\sum_{k=1}^{n-2}
\left[\binom{n-1}{k}-1\right]t_\alpha^k<0
\qquad(n\ge3)\ .
\label{eq:counterexample-K-negative}
\end{align}
The independent Shannon proof in
\cref{eq:counterexample-shannon-expansion} gives
$K_n(1)=2n-2^n<0$, in agreement with the continuous value of
\cref{eq:counterexample-K-factorization}.
In particular, the lowest sector $n=3$ satisfies
\begin{align}
h_\alpha(W_{\rho_{\epsilon,3}})-h_\alpha(W_0)
&=\frac{2(\alpha-2)}{\alpha}\epsilon^2+o(\epsilon^2)\ .
\label{eq:counterexample-n3}
\end{align}
The precise quantifiers are
\begin{align}
\forall\alpha\in(0,2)\ \exists\epsilon_0(\alpha)>0:\qquad
0<\epsilon<\epsilon_0(\alpha)
\ \Longrightarrow\
h_\alpha(W_{\rho_{\epsilon,n}})<h_\alpha(W_0)\ .
\label{eq:counterexample-quantifiers}
\end{align}
The fixed parameters $n$, $\nu$, and $K$ are suppressed in
$\epsilon_0(\alpha)$. The estimates are not uniform as
$\alpha\downarrow0$ or $\alpha\uparrow2$.

At order two, Moyal's identity \cite{moyal49} gives the exact relation
\begin{align}
h_2(W_\rho)-h_2(W_0)
&=-\ln\Tr\rho^2\ge0\ .
\label{eq:counterexample-order-two}
\end{align}
For the repaired family, the thermal overlap is
$A_\epsilon=\langle\psi_{\epsilon,n}|\tau_\nu|\psi_{\epsilon,n}\rangle
=(1-\epsilon^2)/(\nu+1)
+\epsilon^2\nu^n/(\nu+1)^{n+1}$, and the purity is
\begin{align}
\Tr\rho_{\epsilon,n}^2
&=(1-\delta_\epsilon)^2
+2\delta_\epsilon(1-\delta_\epsilon)A_\epsilon
+\frac{\delta_\epsilon^2}{2\nu+1}\ .
\label{eq:counterexample-purity}
\end{align}
Since $A_\epsilon\le1/(\nu+1)<1$, the mixed state has a strictly
positive order-two entropy shift,
\begin{align}
h_2(W_{\rho_{\epsilon,n}})-h_2(W_0)
&=2\delta_\epsilon(1-A_\epsilon)+O(\delta_\epsilon^2)>0\ .
\label{eq:counterexample-order-two-shift}
\end{align}
This shift is smaller than every algebraic order in $\epsilon$.
At each fixed nonzero $\epsilon$, continuity in $\alpha$ leaves a
nonviolating interval immediately below two. This is consistent with
\cref{eq:counterexample-quantifiers}, which fixes $\alpha$ before taking
the small-$\epsilon$ limit.

For each fixed finite $\alpha>2$, the same full-space expansion has
$K_n(\alpha)>0$. Indeed, $u=-t_\alpha$ lies in $(0,1)$.
For even $n$, $\sum_{k=0}^{n-1}(-u)^k
=(1-u)\sum_{j=0}^{n/2-1}u^{2j}\ge1-u$, and for odd $n$,
the sum equals $(1+u^n)/(1+u)>1-u$.
Both exceed $(1-u)^{n-1}$ when $n\ge3$.
Thus \cref{eq:counterexample-K-factorization} is positive above two.
The entropy violation for this family is therefore specific to orders
below two, with the quantifiers stated in
\cref{eq:counterexample-quantifiers}.

Wigner--R\'enyi entropies are Schur concave. If the vacuum continuously
majorized $W_{\rho_{\epsilon,n}}$, then every entropy in
\cref{eq:counterexample-quantifiers} would be no smaller than its vacuum
value. Thus, for every fixed $\alpha\in(0,2)$ and sufficiently small
$\epsilon$, one has
\begin{align}
W_{\rho_{\epsilon,n}}
&\nprec_{\mathrm{cont}}W_0\ .
\label{eq:counterexample-no-cont-majorization}
\end{align}

\subsection*{Continuous-majorization incomparability}

The reverse ordering can be excluded through the positive-part
characterization of continuous majorization. For a non-negative density
$W$, let $\Phi_t(W):=\int(W-t)_+\,\dif^2\boldsymbol r$, where
$(y)_+:=\max\{y,0\}$. Two normalized
densities obey $W\prec_{\mathrm{cont}}V$ precisely when
$\Phi_t(W)\le\Phi_t(V)$ for every $t\ge0$ \cite{ryff65,chong74}. At the fixed vacuum level
$t=e^{-s}/\pi$, the counterexample satisfies
\begin{align}
\Phi_t(W_{\rho_{\epsilon,n}})-\Phi_t(W_0)
&=\epsilon^2G_n(s)+o(\epsilon^2)\ ,
\label{eq:counterexample-positive-part-expansion}
\end{align}
where
\begin{align}
G_n(s)
&=\int_0^s e^{-x}\bigl[P_n(x)-1\bigr]\,\dif x
+\frac{2^n}{n!}s^ne^{-s}\ .
\label{eq:counterexample-Gn}
\end{align}
For fixed $s>0$, the boundary calculation can be carried out first for
the pure Wigner function. The functional $\Phi_t$ is well-defined for this
integrable signed function, although a probability-density entropy is not.
On a fixed compact region,
$W_{\psi_{\epsilon,n}}=W_0+\epsilon W_1+\epsilon^2W_2+O(\epsilon^3)$,
where $W_1=W_0u_n(x)\cos(n\theta)$ and $W_2=W_0[P_n(x)-1]$.

The superlevel set has a single radial boundary for sufficiently small
$\epsilon$. To establish this, a fixed $0<\eta<s/2$ separates the inner
disk $x\le s-\eta$, where $W_0>t$ by a positive margin, from the annulus
$s+\eta\le x\le R$, where $W_0<t$ by a positive margin.
A sufficiently large fixed $R$ also makes
$C_ne^{-x}(1+x^n)<t$ throughout $x\ge R$.
These margins persist under the perturbation. In the remaining annulus
$|x-s|<\eta$, the derivative $\partial_xW_0=-W_0$ is bounded away from
zero, and the perturbed radial derivative remains negative.
The implicit-function theorem gives a unique boundary
$x=\xi_\epsilon(\theta)=s+\epsilon\xi_1(\theta)+O(\epsilon^2)$,
with
\begin{align}
\xi_1(\theta)
&=-\frac{W_1(s,\theta)}{W_0'(s)}
=u_n(s)\cos(n\theta)\ .
\label{eq:counterexample-boundary-displacement}
\end{align}
The positive-part integral is consequently
\begin{align}
\Phi_t(W_{\psi_{\epsilon,n}})
&=\frac12\int_0^{2\pi}\dif\theta
\int_0^{\xi_\epsilon(\theta)}
[W_{\psi_{\epsilon,n}}(x,\theta)-t]\,\dif x\ .
\label{eq:counterexample-moving-integral}
\end{align}
Because $W_0(s)=t$, the first-order boundary term vanishes.
At second order, expansion of the inner integral gives
\begin{align}
\int_0^{\xi_\epsilon}(W_{\psi_{\epsilon,n}}-t)\,\dif x
&=\int_0^s(W_0-t)\,\dif x
+\epsilon\int_0^sW_1\,\dif x
\nonumber\\
&\quad+\epsilon^2\left[
\int_0^sW_2\,\dif x
+W_1(s,\theta)\xi_1
+\frac12W_0'(s)\xi_1^2\right]+O(\epsilon^3)\ .
\label{eq:counterexample-boundary-expansion}
\end{align}
The two terms arising from boundary displacement combine to
$W_1(s,\theta)^2/[2W_0(s)]$.
The angular average removes the first-order term, while the fixed-domain
second-order term becomes
$\int_0^s e^{-x}[P_n(x)-1]\,\dif x$.
The boundary contribution is
\begin{align}
\frac14\int_0^{2\pi}
\frac{W_1(s,\theta)^2}{W_0(s)}\,\dif\theta
&=\beta_n^2s^ne^{-s}
=\frac{2^n}{n!}s^ne^{-s}\ ,
\label{eq:counterexample-boundary-contribution}
\end{align}
which proves the coefficient in \cref{eq:counterexample-Gn} for the
pure curve. Finally, the map $z\mapsto(z-t)_+$ is $1$-Lipschitz,
and the $L^1$ norm of the pure Wigner function is uniformly bounded
by its polynomial Gaussian envelope. Hence
\begin{align}
\left|\Phi_t(W_{\rho_{\epsilon,n}})
-\Phi_t(W_{\psi_{\epsilon,n}})\right|
&\le\delta_\epsilon
\int|W_{\tau_\nu}-W_{\psi_{\epsilon,n}}|\,\dif^2\boldsymbol r
=O(\delta_\epsilon)=o(\epsilon^2)\ .
\label{eq:counterexample-positive-part-repair}
\end{align}
Thus the boundary coefficient is inherited by the Wigner-positive family.
The expansion is at fixed $s$, and does not require uniform control as
$s$ tends to zero or infinity.

Near the origin, the Laguerre expansion gives
\begin{align}
G_n(s)
&=\begin{cases}
-2s+O(s^2),&n\ \text{odd},\\
-ns^2+O(s^3),&n\ \text{even},
\end{cases}
\qquad s\downarrow0\ .
\label{eq:counterexample-Gn-small}
\end{align}
For even $n$, the construction has $n\ge4$.
At large $s$, normalization permits the integral in
\cref{eq:counterexample-Gn} to be written as
$-\int_s^\infty e^{-x}[P_n(x)-1]\,\dif x$.
The leading polynomial coefficients are
$a_n=2^n/n!$ and $a_{n-1}=-n^2 2^{n-1}/n!$.
Integration by parts gives a leading tail term
$a_ns^ne^{-s}$, which cancels the boundary contribution, and a next
coefficient $-(na_n+a_{n-1})=n(n-2)2^{n-1}/n!$.
Consequently,
\begin{align}
G_n(s)
&=\frac{n(n-2)2^{n-1}}{n!}
s^{n-1}e^{-s}+O\!\left(s^{n-2}e^{-s}\right)>0\ .
\label{eq:counterexample-Gn-large}
\end{align}
The coefficient is therefore negative at a fixed sufficiently small value
of $s$ and positive at a fixed sufficiently large value. Both signs persist
for all sufficiently small $\epsilon$ at every fixed $n\ge3$. The small-$s$ sign excludes
$W_0\prec_{\mathrm{cont}}W_{\rho_{\epsilon,n}}$, while the large-$s$ sign
excludes the reverse ordering. Hence
\begin{align}
W_{\rho_{\epsilon,n}}&\nprec_{\mathrm{cont}}W_0\ ,
&
W_0&\nprec_{\mathrm{cont}}W_{\rho_{\epsilon,n}}\ .
\label{eq:counterexample-incomparability}
\end{align}
For $n=3$, the crossing is explicit because
$G_3(s)=2s(s-1)e^{-s}$. The fixed choices $s=1/2$ and $s=2$
give $G_3(1/2)=-e^{-1/2}/2<0$ and $G_3(2)=4e^{-2}>0$.
Only two fixed levels are needed to exclude both orderings.

\Cref{fig:counterexample-tests-app} illustrates the entropy violation and
the sign change of the positive-part coefficient.
\begin{figure}[tb]
\centering
\begin{tabular}{@{}c@{\hspace{0.03\textwidth}}c@{}}
\textbf{(a)} & \textbf{(b)}\\[-2pt]
\includegraphics[height=0.22\textwidth]{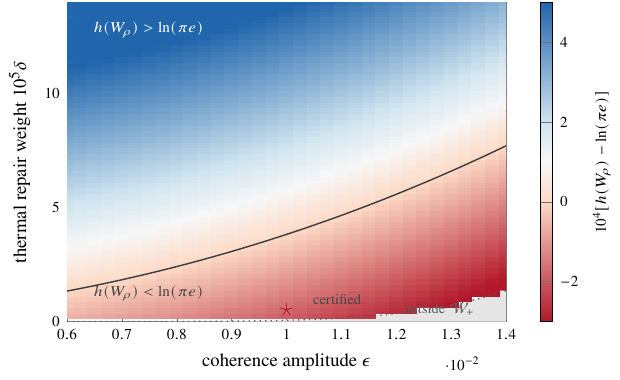}
&
\includegraphics[height=0.22\textwidth]{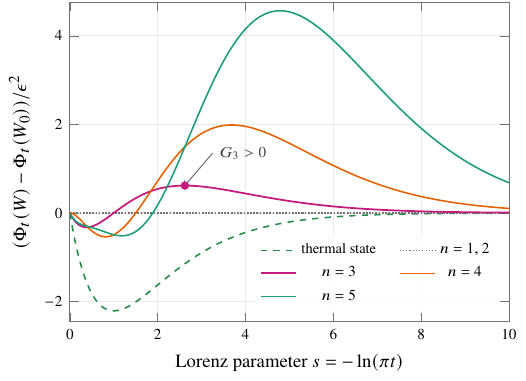}
\end{tabular}
\caption{\label{fig:counterexample-tests-app}
Entropy and continuous-majorization tests. (a) Numerical entropy landscape
for $n=3$ and $\nu=3$. The color scale shows
$10^4[h(W_\rho)-\ln(\pi e)]$. The solid curve marks equality with the vacuum,
the dotted curve is the threshold for Wigner positivity, and the star marks
the parameters used in \cref{fig:counterexample-phase-space-app}(b).
(b) The coefficient $G_n(s)$ in
\cref{eq:counterexample-positive-part-expansion}. The curves for
$n=3,4,5$ illustrate the negative and positive regions that exclude the two
directions of continuous majorization.}
\end{figure}

\subsection*{Physical properties and radial coarse graining}

The counterexamples approach the vacuum in trace norm and energy. The
projector difference for the pure state in
\cref{eq:counterexample-pure} has trace norm $2\epsilon$, because its
two nonzero eigenvalues are $\pm\epsilon$. Convexity of the trace norm
therefore gives
\begin{align}
\left\|\rho_{\epsilon,n}-\ket0\bra0\right\|_1
&\le2(1-\delta_\epsilon)\epsilon+2\delta_\epsilon
\longrightarrow0\ .
\label{eq:counterexample-trace-convergence}
\end{align}
The mean photon number is exactly
\begin{align}
E_\epsilon
&=\Tr(\rho_{\epsilon,n}\hat N)
=(1-\delta_\epsilon)n\epsilon^2+\delta_\epsilon\nu
\longrightarrow0\ .
\label{eq:counterexample-energy-convergence}
\end{align}
Thus the failure of continuous vacuum majorization occurs arbitrarily
close to the vacuum, with a vanishing absolute entropy difference.

For $n\ge3$, the vacuum--Fock coherence contributes neither
$\langle\hat a\rangle$ nor $\langle\hat a^2\rangle$, since the
ladder operators change the photon number by only one or two.
The thermal state has the same vanishing moments. The covariance matrix
of $\hat q=(\hat a+\hat a^\dagger)/\sqrt2$ and
$\hat p=(\hat a-\hat a^\dagger)/(i\sqrt2)$ is consequently
\begin{align}
V_{\rho_{\epsilon,n}}
&=\left(\frac12+E_\epsilon\right)I_2\ .
\label{eq:counterexample-covariance}
\end{align}
Its determinant exceeds $1/4$ for every nonzero member of the family.
The entropy decrease therefore coexists with quadrature variances larger
than the vacuum values.

The quadrature entropies also distinguish the joint phase-space effect
from marginal concentration. For the real $n=3$ superposition,
the vacuum quadrature density is
$f_0(y)=\pi^{-1/2}e^{-y^2}$, and the ratio of oscillator wave functions is
$R_3(y)=\psi_3(y)/\psi_0(y)=(2y^3-3y)/\sqrt3$.
On a growing central interval, the repaired quadrature densities obey
\begin{align}
f_q(y)
&=f_0(y)\left[1+2\epsilon R_3(y)
+\epsilon^2(R_3(y)^2-1)\right]+o(\epsilon^2)\ ,
\label{eq:counterexample-q-marginal}
\end{align}
and
\begin{align}
f_p(y)
&=f_0(y)\left[1+\epsilon^2(R_3(y)^2-1)\right]
+o(\epsilon^2)\ .
\label{eq:counterexample-p-marginal}
\end{align}
The linear interference vanishes in the momentum marginal because the
three-photon wave function acquires the phase $(-i)^3$ under Fourier
transformation. The Gaussian moments
$\int f_0R_3^2\,\dif y=1$ and
$\int y^2f_0(R_3^2-1)\,\dif y=3$ give
\begin{align}
h(f_q)-h(f_0)&=\epsilon^2+o(\epsilon^2)\ ,
\label{eq:counterexample-q-entropy}
\end{align}
where the relative-entropy correction is $2\epsilon^2$, and
\begin{align}
h(f_p)-h(f_0)&=3\epsilon^2+o(\epsilon^2)\ ,
\label{eq:counterexample-p-entropy}
\end{align}
where that correction begins at order $\epsilon^4$.
These are full-line entropy expansions. Their justification uses
$y^2\le\epsilon^{-\gamma}$ with $0<\gamma<2/3$ and the same
polynomial Gaussian and thermal tail estimates as
\cref{eq:counterexample-core-thermal,eq:counterexample-shannon-tail}.
The exponentially small thermal mixing does not change the displayed
coefficients. Both marginal entropies increase, while the joint Wigner
entropy decreases by $2\epsilon^2+o(\epsilon^2)$.

The Shannon violation also places the counterexamples outside convex
mixtures of Gaussian states. A one-mode Gaussian state with covariance
matrix $V$ has
$h(W)=\ln(2\pi e\sqrt{\det V})\ge\ln(\pi e)$ because
$\det V\ge1/4$. Concavity of $-w\ln w$ then gives, for any Gaussian
mixture $W=\int W_\lambda\,\dif\mu(\lambda)$,
\begin{align}
h(W)
&\ge\int h(W_\lambda)\,\dif\mu(\lambda)
\ge\ln(\pi e)\ .
\label{eq:counterexample-gaussian-mixture}
\end{align}
This argument also applies to continuous probability mixtures.
Each Gaussian Wigner density is bounded by $1/\pi$, making its
entropy integrand non-negative and allowing the order of integration
to be exchanged. The states with a strict Shannon violation are
therefore quantum non-Gaussian in this convex-mixture sense, despite
having non-negative Wigner functions.

The entropy and continuous-majorization conclusions are preserved by
Gaussian unitaries. If $S$ is symplectic and $\boldsymbol d$ is a
displacement, the transformed density is
$W'(\boldsymbol r)=W(S^{-1}(\boldsymbol r-\boldsymbol d))$.
The Jacobian is one, preserving both $J_\alpha$ and every
$\Phi_t$. The Gaussian image of the vacuum is itself level-equivalent
to the vacuum. The transformed counterexample therefore has the same
entropy deficit and remains incomparable with $W_0$.
Independent copies also preserve the entropy violation, since
\begin{align}
h_\alpha(W_{\rho^{\otimes m}})-h_\alpha(W_{0^{\otimes m}})
&=m\,[h_\alpha(W_\rho)-h_\alpha(W_0)]\ .
\label{eq:counterexample-tensor-entropy}
\end{align}
For $\alpha\ne1$ this follows by factorization of $J_\alpha$,
and for $\alpha=1$ by Shannon additivity. This statement concerns the
entropy comparison and does not assume a multimode radial-shell theorem.

The same angular structure explains the connection to the finite-resolution
relation in the main text. Averaging over phase rotations produces
\begin{align}
\overline\rho_{\epsilon,n}
&=\frac1{2\pi}\int_0^{2\pi}
e^{-i\phi\hat N}\rho_{\epsilon,n}e^{i\phi\hat N}\,\dif\phi
\nonumber\\
&=(1-\delta_\epsilon)
\left[(1-\epsilon^2)\ket0\bra0+\epsilon^2\ket n\bra n\right]
+\delta_\epsilon\tau_\nu\ .
\label{eq:counterexample-phase-average}
\end{align}
Every radial shell has the same integrated weight for
$\rho_{\epsilon,n}$ and $\overline\rho_{\epsilon,n}$, because its full
angular integral removes $\cos(n\theta)$ exactly.
All these weights are non-negative for the Wigner-positive family.
Applying \cref{eq:main} in the chosen oscillator frame gives
\begin{align}
\boldsymbol P_{\rho_{\epsilon,n}}(\Delta)
&=\boldsymbol P_{\overline\rho_{\epsilon,n}}(\Delta)
\prec\boldsymbol P_0(\Delta)
\qquad\text{for every }\Delta>0\ .
\label{eq:counterexample-shell-majorization}
\end{align}
The counterexample states are therefore incomparable with the vacuum under
continuous majorization while obeying the finite-resolution radial relation
at every shell width.

\section{Proof of the majorization relation for radial shell weights}
\label{app:radial-majorization}

We give the complete proof of \cref{thm:radial-majorization}, repeating the
partition notation so that the argument is self-contained.
Let $\Pi:0=r_0<r_1<\cdots$ be a complete radial partition with
$r_k\to\infty$ and bounded squared-radius widths
$d_k:=r_k^2-r_{k-1}^2$. Its shells are
$A_k(\Pi):=\{\boldsymbol r:r_{k-1}\le|\boldsymbol r|<r_k\}$.
Their Wigner weights are
\begin{align}
P_{\rho,k}(\Pi):=\int_{A_k(\Pi)}W_\rho(\boldsymbol r)\,\dif\boldsymbol r\ .
\label{eq:general-shell-weights-app}
\end{align}
The complete shell vector is
$\boldsymbol P_\rho(\Pi):=(P_{\rho,1},P_{\rho,2},\ldots)$.
The equal-area partition in the main text is the specialization
$r_k^2=k\Delta$.

Phase averaging first reduces the shell weights to those of a
Fock-diagonal state. After introducing the radial profiles and their
tails, we give a general subset-sum argument that converts a tail
inequality into majorization. To establish this inequality for physical
states, the Mehler formula for Hermite polynomials
\cite[Eq.~18.18.28]{dlmf24} expresses the relevant alternating Laguerre
sum as a non-negative sum of squared Hermite polynomials. A
Laplace-transform argument then establishes normalization whenever a
complete partition with bounded squared-radius widths has componentwise
non-negative shell weights. Combining these results yields majorization
by $\boldsymbol P_0(\Pi^\downarrow)$, the decreasing vacuum benchmark
defined by the order statistics $D_m$. This benchmark is the vacuum
shell vector of an actual reordered partition only when the widths admit
a non-increasing enumeration. The relation for equal-area shells follows
directly.

Let $U_\theta=e^{-i\theta\hat n}$ denote a phase rotation, and let
$\mathsf R_\theta$ be the corresponding planar rotation matrix. Define
$\bar\rho:=(2\pi)^{-1}\int_0^{2\pi}
U_\theta\rho U_\theta^\dagger\,\dif\theta$.
Its Wigner function is
$W_{\bar\rho}(\boldsymbol r)=(2\pi)^{-1}\int_0^{2\pi}
W_\rho(\mathsf R_{-\theta}\boldsymbol r)\,\dif\theta$. Since every
$A_k(\Pi)$ is rotation invariant, Fubini's theorem and the change of
variables $\boldsymbol r\mapsto\mathsf R_\theta\boldsymbol r$ give
\begin{align}
\int_{A_k(\Pi)}W_{\bar\rho}(\boldsymbol r)\,\dif\boldsymbol r
&=\frac{1}{2\pi}\int_0^{2\pi}\int_{A_k(\Pi)}
W_\rho(\mathsf R_{-\theta}\boldsymbol r)
\,\dif\boldsymbol r\,\dif\theta
=\int_{A_k(\Pi)}W_\rho(\boldsymbol r)\,\dif\boldsymbol r\ .
\label{eq:phase-average-general-app}
\end{align}
Thus, $\boldsymbol P_\rho(\Pi)=
\boldsymbol P_{\bar\rho}(\Pi)$ for every radial partition. The phase
average is complete dephasing in the Fock basis,
$\bar\rho=\sum_{n=0}^\infty\rho_n\ket n\bra n$, where
$\rho_n\ge0$ and $\sum_n\rho_n=1$.

With $\xi=|\boldsymbol r|^2$, radial integration converts the Fock Wigner
function into the radial Wigner profile
\begin{align}
g_n(\xi):=\pi W_n(\sqrt{\xi},0)
&=(-1)^nL_n(2\xi)e^{-\xi}\ .
\label{eq:fock-radial-density}
\end{align}
This profile is normalized,
$\int_0^\infty g_n(\xi)\,\dif \xi=1$. For the phase-averaged state, define
$g_{\bar\rho}(\xi):=\sum_{n=0}^\infty\rho_ng_n(\xi)$.
The shell weights then become
\begin{align}
P_{\rho,k}(\Pi)
=\int_{r_{k-1}^2}^{r_k^2}g_{\bar\rho}(\xi)\,\dif \xi\ .
\label{eq:general-shell-radial}
\end{align}
Introduce the exterior Fock tail
$\mathcal T_n(\xi):=\int_\xi^\infty g_n(u)\,\dif u
=(-1)^n\int_\xi^\infty L_n(2u)e^{-u}\,\dif u$
and its state average
$\mathcal T_\rho(\xi):=\sum_{n=0}^\infty\rho_n\mathcal T_n(\xi)$.
These identities require no energy constraint. To justify this state-average
expression, let
$D_\xi=\{\boldsymbol r:|\boldsymbol r|^2<\xi\}$ and
$\widehat Q_\xi=\mathbb I-\widehat R_{D_\xi}$. The disk region operator is
Hilbert--Schmidt and hence bounded, so $\widehat Q_\xi$ is bounded. The finite
Fock truncations of $\bar\rho$ converge in
trace norm. Therefore their expectations of $\widehat Q_\xi$ converge and give
$\mathcal T_\rho(\xi)=\Tr(\bar\rho\,\widehat Q_\xi)
=\sum_n\rho_n\bra n\widehat Q_\xi\ket n
=\sum_n\rho_n\mathcal T_n(\xi)$.
In particular, the shell weights are the tail differences
$P_{\rho,k}(\Pi)=
\mathcal T_\rho(r_{k-1}^2)-\mathcal T_\rho(r_k^2)$.

\smallskip
\noindent\textit{From tail inequalities to majorization.} Let
$\boldsymbol q=(q_1,q_2,\ldots)$ be a probability vector with tails
$T_k:=\sum_{j\ge k}q_j$. Let $d_k>0$, $\sum_kd_k=\infty$, and suppose
\begin{align}
T_{k+1}\ge e^{-d_k}T_k
\quad\Longleftrightarrow\quad
q_k\le(1-e^{-d_k})T_k
\qquad(k\ge1)\ .
\label{eq:geometric-tail-condition}
\end{align}
For $m\ge0$, define
\begin{align}
D_m:=\sup_{\substack{I\subset\mathbb N\\|I|=m}}\sum_{i\in I}d_i\ ,
\qquad D_0:=0\ .
\label{eq:Dm-app}
\end{align}
The decreasing vacuum-benchmark components are
$P_{0,m}(\Pi^\downarrow):=e^{-D_{m-1}}-e^{-D_m}$.
If the widths are bounded, then
$\boldsymbol P_0(\Pi^\downarrow)=
(P_{0,1}(\Pi^\downarrow),P_{0,2}(\Pi^\downarrow),\ldots)$ is a
non-increasing probability vector and
\begin{align}
\boldsymbol q\prec\boldsymbol P_0(\Pi^\downarrow)\ .
\label{eq:variable-tail-majorization}
\end{align}

Since $\sup_kd_k<\infty$, every $D_m$ is finite. For the first $N$
widths, let $D_m^{(N)}$ be the sum of their $m$ largest members. Writing those
members in non-increasing order gives
$D_m^{(N)}-D_{m-1}^{(N)}\ge D_{m+1}^{(N)}-D_m^{(N)}$.
For each fixed $m$, $D_m^{(N)}$ increases to $D_m$ as $N\to\infty$.
Passing to the limit shows that $D_m$ is a concave sequence. Hence its
increments $\delta_m:=D_m-D_{m-1}$ are non-increasing and non-negative. Moreover,
$D_m\ge\sum_{k=1}^md_k\to\infty$, so
$\sum_{m=1}^\infty P_{0,m}(\Pi^\downarrow)
=\lim_{M\to\infty}(1-e^{-D_M})=1$.
The components are themselves non-increasing because
\begin{align}
P_{0,m+1}(\Pi^\downarrow)&=e^{-D_m}(1-e^{-\delta_{m+1}})
\le e^{-D_{m-1}}e^{-\delta_m}(1-e^{-\delta_m})
\le P_{0,m}(\Pi^\downarrow)\ .
\label{eq:vacuum-benchmark-decreasing}
\end{align}
Thus, $\boldsymbol P_0(\Pi^\downarrow)$ is a decreasing probability vector.

\smallskip
\noindent\textit{Interpretation of the decreasing benchmark.}
For the first $N$ shell widths, placing the $m$ largest members in
non-increasing order from the origin gives the cumulative vacuum weight
$1-e^{-D_m^{(N)}}$ in the first $m$ shells. Since
$D_m^{(N)}\uparrow D_m$ for every fixed $m$, these cumulative weights
converge to those of $\boldsymbol P_0(\Pi^\downarrow)$. If the full width
sequence admits a non-increasing enumeration, this is the vacuum shell vector
of the corresponding reordered partition. Otherwise,
$\boldsymbol P_0(\Pi^\downarrow)$ is interpreted as the
limiting benchmark; for example, $d_k=1-1/(k+1)$ gives $D_m=m$ but has no
largest width. The majorization argument below depends only on $D_m$ and does
not require the suprema to be attained.

Now choose arbitrary indices $i_1<\cdots<i_m$ and set
$R_\ell:=\sum_{j=1}^{\ell}q_{i_j}$, with $R_0=0$. The previously selected
components lie before the tail beginning at $i_\ell$. Since all components
are non-negative,
$T_{i_\ell}\le1-R_{\ell-1}$. Using
\cref{eq:geometric-tail-condition},
$R_\ell=R_{\ell-1}+q_{i_\ell}
\le R_{\ell-1}+(1-e^{-d_{i_\ell}})(1-R_{\ell-1})$.
Induction on $\ell$, followed by the supremum over all $m$-element index
sets, yields
\begin{align}
R_m\le1-\exp\!\left(-\sum_{\ell=1}^m d_{i_\ell}\right)
\le1-e^{-D_m}\ .
\label{eq:subset-tail-bound}
\end{align}
Taking the supremum over the index sets gives
\begin{align*}
\sum_{j=1}^m q_j^\downarrow
\le1-e^{-D_m}
=\sum_{j=1}^m P_{0,j}(\Pi^\downarrow)\ ,
\end{align*}
which is \cref{eq:variable-tail-majorization}.

For equal widths $d_k=d$, one has $D_m=md$. Setting $a=e^{-d}$ gives the
geometric probability vector obtained from the vacuum state
$\boldsymbol g_a=(1-a)(1,a,a^2,\ldots)$, with
$q_k\le(1-a)T_k$ as the sufficient condition. This is the equal-area
specialization used below.

The analytic input is a positive decomposition of the alternating Laguerre
sum
\begin{align}
S_N(\xi):=\sum_{j=0}^N(-1)^jL_j(2\xi),\qquad N\ge0\ .
\label{eq:S-def-app}
\end{align}

For every $N\ge0$ and $\xi\ge0$, the Mehler formula gives
\begin{align}
S_N(\xi)=
\sum_{\ell=0}^{\lfloor N/2\rfloor}
\frac{\binom{2\ell}{\ell}}{4^\ell}
\frac{H_{N-2\ell}(\sqrt{\xi})^2}
{2^{N-2\ell}(N-2\ell)!}\ge0\ .
\label{eq:mehler-sos-app}
\end{align}
If $N$ is even, $S_N(\xi)>0$ for every $\xi\ge0$. If $N$ is odd,
$S_N(0)=0$ and $S_N(\xi)>0$ for every $\xi>0$.

The ordinary Laguerre generating function is \cite[Eq.~18.12.13]{dlmf24}
\begin{align}
\sum_{n=0}^\infty L_n(2u)z^n
=\frac{1}{1-z}\exp\!\left(-\frac{2uz}{1-z}\right)\ ,
\qquad |z|<1\ .
\label{eq:laguerre-generating-app}
\end{align}
Set $\Psi_n(\xi):=e^\xi\mathcal T_n(\xi)$. Substituting $z=-t$ in
\cref{eq:laguerre-generating-app}, multiplying by $e^{-u}$, and integrating
over $u\in[\xi,\infty)$ gives
\begin{align}
\sum_{n=0}^\infty\Psi_n(\xi)t^n
=\frac{e^\xi}{1+t}\int_\xi^\infty
\exp\!\left[-u\frac{1-t}{1+t}\right]\dif u
=\frac{1}{1-t}\exp\!\left(\frac{2\xi t}{1+t}\right)\ ,\quad |t|<1\ .
\label{eq:Psi-genfun}
\end{align}

For completeness, the exchange of the Laguerre series and the improper
integral can be justified uniformly on compact subsets of $|t|<1$. Fix
$0<r<R<1$. Cauchy's coefficient estimate applied to
\cref{eq:laguerre-generating-app} on $|z|=R$ gives
$|L_n(2u)|\le R^{-n}\exp[2Ru/(1+R)]/(1-R)$, since
$\operatorname{Re}[z/(1-z)]\ge-R/(1+R)$ on this circle. Therefore, uniformly
for $|t|\le r$,
\begin{align}
\sum_{n=0}^\infty |t|^n|L_n(2u)|e^{-u}
&\le
\frac{1}{(1-R)(1-r/R)}
\exp\!\left[-u\frac{1-R}{1+R}\right]\ .
\label{eq:laguerre-dominated}
\end{align}
The right-hand side is integrable on $[\xi,\infty)$. Dominated convergence
proves \cref{eq:Psi-genfun} and its locally uniform convergence in $t$, which
also permits the coefficient comparisons and differentiations below.

Differentiating \cref{eq:Psi-genfun} with respect to $\xi$ yields
$\sum_{n=0}^\infty\Psi_n'(\xi)t^n=
2t\exp[2\xi t/(1+t)]/[(1-t)(1+t)]$.
On the other hand, \cref{eq:S-def-app,eq:laguerre-generating-app} give
$\sum_{n\ge1}2S_{n-1}(\xi)t^n
=2t\exp[2\xi t/(1+t)]/[(1-t)(1+t)]$.
Coefficient comparison proves $\Psi_0'(\xi)=0$ and
$\Psi_n'(\xi)=2S_{n-1}(\xi)$ for $n\ge1$.

It remains to establish the sign of $S_N$. Summing
\cref{eq:S-def-app} first over $N$ gives
\begin{align}
\sum_{N=0}^\infty S_N(\xi)t^N
=\frac{1}{1-t^2}
\exp\!\left(\frac{2\xi t}{1+t}\right)\ .
\label{eq:S-generating}
\end{align}
The Mehler formula for the physicists' Hermite polynomials reads \cite[Eq.~18.18.28]{dlmf24}
\begin{align}
\sum_{m=0}^\infty
\frac{H_m(q)H_m(q')}{2^m m!}t^m
=\frac{1}{\sqrt{1-t^2}}
\exp\!\left[
\frac{2qq't-(q^2+q'^2)t^2}{1-t^2}
\right]\ .
\label{eq:mehler-app}
\end{align}
Taking $q=q'=\sqrt{\xi}$ reduces the exponent to $2\xi t/(1+t)$, so
\begin{align}
\sum_{m=0}^\infty
\frac{H_m(\sqrt{\xi})^2}{2^m m!}t^m
=\frac{1}{\sqrt{1-t^2}}
\exp\!\left(\frac{2\xi t}{1+t}\right)\ .
\label{eq:mehler-diagonal}
\end{align}
The remaining factor is
$(1-t^2)^{-1/2}=\sum_{\ell\ge0}\binom{2\ell}{\ell}t^{2\ell}/4^\ell$.
Multiplying this expansion by \cref{eq:mehler-diagonal} reproduces
\cref{eq:S-generating}. Extracting the coefficient of $t^N$ proves
\cref{eq:mehler-sos-app}.

Every summand in \cref{eq:mehler-sos-app} is non-negative. If $N$ is even,
the term $\ell=N/2$ contains $H_0^2=1$ and is strictly positive. If $N$ is
odd, the term $\ell=(N-1)/2$ contains
$H_1(\sqrt{\xi})^2=4\xi$, which is positive for $\xi>0$. At $\xi=0$, every Hermite
polynomial appearing has odd degree and vanishes. This proves the stated
strictness.

\smallskip
\noindent\textit{Tail inequality on a radial interval.}\par\nopagebreak[4]

Integrating $\Psi_n'=2S_{n-1}\ge0$ for $n\ge1$, with $\Psi_0'=0$
for the vacuum, gives, for every $n\ge0$, $\xi\ge0$, and $d>0$,
\begin{align}
\mathcal T_n(\xi+d)\ge e^{-d}\mathcal T_n(\xi)\ .
\label{eq:finite-shell-inequality}
\end{align}
Equality holds for the vacuum. For every $n\ge1$, the difference has the
explicit positive representation
\begin{align}
\mathcal T_n(\xi+d)-e^{-d}\mathcal T_n(\xi)
&=2e^{-(\xi+d)}
\sum_{\ell=0}^{\lfloor(n-1)/2\rfloor}
\frac{\binom{2\ell}{\ell}}{4^\ell}
\int_\xi^{\xi+d}
\frac{H_{n-1-2\ell}(\sqrt{u})^2}
{2^{n-1-2\ell}(n-1-2\ell)!}\,\dif u
>0\ .
\label{eq:finite-shell-squares}
\end{align}

Indeed, $\mathcal T_0(\xi)=e^{-\xi}$. For $n\ge1$, integrating the coefficient identity
above between $\xi$ and $\xi+d$ gives
$e^{\xi+d}\mathcal T_n(\xi+d)-e^\xi\mathcal T_n(\xi)
=2\int_\xi^{\xi+d}S_{n-1}(u)\,\dif u$.
Substitution of \cref{eq:mehler-sos-app} and multiplication by
$e^{-(\xi+d)}$ give \cref{eq:finite-shell-squares}. By
\cref{eq:mehler-sos-app}, $S_{n-1}$ is positive throughout the integration
interval except possibly at its single endpoint $u=0$. The integral is
therefore strictly positive.

Averaging \cref{eq:finite-shell-inequality} over the Fock populations gives
\begin{align}
\mathcal T_\rho(\xi+d)\ge e^{-d}\mathcal T_\rho(\xi)
\qquad(\xi\ge0,\ d>0)\ .
\label{eq:state-finite-shell}
\end{align}
No pointwise positivity of $W_\rho$, and no non-negativity of any shell
weight, is needed here.

\begin{lemma}
\label{lem:shell-normalization-app}
Let $\Pi$ be a complete radial partition, set $x_k=r_k^2$, and assume
$D:=\sup_k(x_k-x_{k-1})<\infty$. If
$P_{\rho,k}(\Pi)\ge0$ for every $k$, then
$\lim_{k\to\infty}\mathcal T_\rho(x_k)=0$.
Consequently,
\begin{align}
\sum_{k=1}^\infty P_{\rho,k}(\Pi)=1\ .
\label{eq:shell-normalization-app}
\end{align}
\end{lemma}

\begin{proof}
Taking the initial point to be zero and $d=\xi$ in
\cref{eq:finite-shell-inequality} gives
$\mathcal T_n(\xi)\ge e^{-\xi}\mathcal T_n(0)=e^{-\xi}$, so every Fock
tail is non-negative. For $0<a_{\rm L}<1$, the Laguerre Laplace transform and
\cref{eq:fock-radial-density} give
\begin{align}
\widehat g_\rho(a_{\rm L})
&:=\int_0^\infty e^{-a_{\rm L}\xi}g_{\bar\rho}(\xi)\,\dif\xi
=\frac{1}{1+a_{\rm L}}\sum_{n=0}^\infty\rho_n
\left(\frac{1-a_{\rm L}}{1+a_{\rm L}}\right)^n
\xrightarrow[a_{\rm L}\downarrow0]{}1\ .
\label{eq:radial-laplace-app}
\end{align}
The bound $|g_n(\xi)|\le1$ \cite[Eq.~18.14.8]{dlmf24} justifies exchanging
the Fock sum and the integral, while dominated convergence in $n$ gives the
limit. Since $\mathcal T_n\ge0$, Tonelli's theorem and integration by parts
for each Fock tail yield
\begin{align}
a_{\rm L}\int_0^\infty e^{-a_{\rm L}\xi}\mathcal T_\rho(\xi)\,\dif\xi
=1-\widehat g_\rho(a_{\rm L})
\xrightarrow[a_{\rm L}\downarrow0]{}0\ .
\label{eq:tail-abel-app}
\end{align}

Because
$P_{\rho,k}=\mathcal T_\rho(x_{k-1})-\mathcal T_\rho(x_k)\ge0$,
the boundary tails decrease to some $L\ge0$. For
$\xi\in[x_{k-1},x_k]$, \cref{eq:state-finite-shell} implies
$\mathcal T_\rho(\xi)\ge
e^{-(\xi-x_{k-1})}\mathcal T_\rho(x_{k-1})\ge e^{-D}L$.
Hence the left-hand side of \cref{eq:tail-abel-app} is at least $e^{-D}L$
for every $0<a_{\rm L}<1$, and its limit forces $L=0$. Finally,
$\sum_{k=1}^N P_{\rho,k}(\Pi)=
\mathcal T_\rho(0)-\mathcal T_\rho(x_N)\to1$ as $N\to\infty$, which
proves the claim.
\end{proof}

\noindent\textit{Majorization of the radial shell weights.} Let $\Pi$ be a
complete radial partition with bounded squared-radius widths $d_k$, and write
$D_m:=D_m(\Pi)$ as defined in \cref{eq:Dm-app}. If
$\boldsymbol P_\rho(\Pi)$ is componentwise non-negative, then
\begin{align}
\boldsymbol P_\rho(\Pi)
&\prec\boldsymbol P_0(\Pi^\downarrow)\ .
\label{eq:radial-shell-app}
\end{align}
The ordered vacuum components are
\begin{align}
P_{0,m}(\Pi^\downarrow)
&=e^{-D_{m-1}}-e^{-D_m}\ .
\label{eq:ordered-vacuum-app}
\end{align}
Equivalently,
\begin{align}
\sum_{j=1}^mP_{\rho,j}^\downarrow(\Pi)
\le1-e^{-D_m}
\qquad(m\ge1)\ .
\label{eq:general-Lorenz-bound-app}
\end{align}

By \cref{eq:shell-normalization-app}, write
$T_k:=\sum_{j\ge k}P_{\rho,j}(\Pi)=\mathcal T_\rho(r_{k-1}^2)$.
Taking
$\xi=r_{k-1}^2$ and $d=d_k$ in \cref{eq:state-finite-shell} gives
$T_{k+1}=\mathcal T_\rho(r_k^2)
\ge e^{-d_k}\mathcal T_\rho(r_{k-1}^2)=e^{-d_k}T_k$.
The componentwise non-negativity assumption makes
$\boldsymbol P_\rho(\Pi)$ a probability vector. The preceding subset-sum
argument then proves
\cref{eq:radial-shell-app,eq:general-Lorenz-bound-app}.

The right-hand side has an exact concentration interpretation. Select any
$m$ shells with total radial width
$D(I)=\sum_{i\in I}d_i$. Since a width $d$ in $\xi=|\boldsymbol r|^2$ has
phase-space area $\pi d$, placing the selected widths consecutively from the
origin produces the disk $0\le \xi<D(I)$. The integral of the vacuum Wigner function over this disk
is $\int_0^{D(I)}e^{-\xi}\,\dif \xi=1-e^{-D(I)}$.
Because $e^{-\xi}$ is strictly decreasing, this is the largest integral of $W_0$
over a region of the selected total area. Maximizing over all
$m$-element selections gives $1-e^{-D_m}$.

\smallskip
\noindent\textit{Admissibility check for \cref{fig:shell-construction}.}
For $\rho=S(\varsigma)\ket1\bra1S^\dagger(\varsigma)$, choose principal
axes $(u,v)$ and set
$R_\varsigma^2=e^{2\varsigma}u^2+e^{-2\varsigma}v^2$. The Wigner function is
$W_\rho(u,v)=\pi^{-1}(2R_\varsigma^2-1)e^{-R_\varsigma^2}$.
The negative region is $R_\varsigma^2<1/2$, on which
$u^2+v^2\le e^{2\varsigma}/2$. For $\varsigma=0.45$ this upper bound is
$1.229802<\Delta=3/2$, so every shell after $A_1$ lies in the pointwise
non-negative region. It remains only to check the first shell. Setting
$a=\cosh(2\varsigma)$ and $b=\sinh(2\varsigma)$, angular integration followed
by integration by parts gives
\begin{align}
P_{\rho,1}(\Delta)
&=\int_0^\Delta e^{-a\xi}I_0(b\xi)\,\dif\xi
-2\Delta e^{-a\Delta}I_0(b\Delta)\notag\\
&=0.096833984\ldots>0
\qquad(\varsigma=0.45,\ \Delta=3/2),
\label{eq:fig1-admissibility}
\end{align}
where $I_0$ is the modified Bessel function. Hence the complete shell vector
used in \cref{fig:shell-construction}, not only its displayed components, is
componentwise non-negative.

\section{Strict majorization for Gaussian pure states}

For every non-vacuum single-mode Gaussian pure state $\ket\psi$, the
probabilities obtained from equal-area shells in a fixed oscillator frame
satisfy
\begin{align}
\boldsymbol P_\psi(\Delta)&\prec\boldsymbol P_0(\Delta),
\qquad
\boldsymbol P_0(\Delta)\nprec\boldsymbol P_\psi(\Delta)
\qquad(\Delta>0)\ .
\label{eq:strict-gaussian}
\end{align}

This section proves the reverse exclusion in \cref{eq:strict-gaussian};
together with the forward relation \cref{eq:main}, this establishes strictness
for non-vacuum Gaussian pure states. The ratios of consecutive shell probabilities show how
displacement and squeezing separate every non-vacuum Gaussian pure state from
the vacuum at finite radial resolution.
Fix $\Delta>0$ and use
$\boldsymbol r=(q,p)=\sqrt{\xi}(\cos\theta,\sin\theta)$, where
$\xi=|\boldsymbol r|^2\ge0$. For a phase-invariant state $\bar\sigma$, define
$g(\xi):=\pi W_{\bar\sigma}(\boldsymbol r)|_{|\boldsymbol r|^2=\xi}$, so that
$\int_0^\infty g(\xi)\,\dif \xi=1$. Consider the displaced squeezed state
$\ket{\psi} = D(\alpha) S(\zeta)\ket0$, where $\zeta = \varsigma e^{i\phi_0}$
and $\varsigma=|\zeta|\ge0$. Its Wigner function is \cite{weedbrook12}
\begin{align}
W_\psi(\boldsymbol{r}) = \frac{1}{\pi}
\exp\Bigl[-\tfrac{1}{2}(\boldsymbol{r}-\boldsymbol r_\alpha)^{\mathsf T}
\Sigma^{-1}(\boldsymbol{r}-\boldsymbol r_\alpha)\Bigr]\ ,
\label{eq:wigner_gauss}
\end{align}
where $\boldsymbol r_\alpha \in \RR^2$ is the displacement. With
$[\hat q,\hat p]=i$, the complex amplitude in $D(\alpha)$ and the real
phase-space displacement are related by
$\boldsymbol r_\alpha=(\sqrt2\,\operatorname{Re}\alpha,
\sqrt2\,\operatorname{Im}\alpha)$, and hence
$|\boldsymbol r_\alpha|^2=2|\alpha|^2$.
Moreover,
$\Sigma=\tfrac12\mathsf R_\phi\operatorname{diag}
(e^{-2\varsigma},e^{2\varsigma})\mathsf R_\phi^{\mathsf T}$, with
$\phi=\phi_0/2$ and $\mathsf R_\phi$ the rotation by $\phi$. Equivalently,
$\tfrac12\boldsymbol u^{\mathsf T}\Sigma^{-1}\boldsymbol u
=e^{2\varsigma}(\boldsymbol u\cdot\boldsymbol e_\phi)^2
+e^{-2\varsigma}(\boldsymbol u\cdot\boldsymbol e_\phi^\perp)^2$,
where $\boldsymbol{e}_\phi = (\cos\phi, \sin\phi)$ is the squeezing direction
and $\boldsymbol{e}_\phi^\perp = (-\sin\phi, \cos\phi)$. Choose coordinates
so that $\boldsymbol r_\alpha = (a, 0)$ with $a := |\boldsymbol r_\alpha| \ge 0$,
and introduce the squeezed-frame components
$A:=a\cos\phi$ and $B:=a\sin\phi$.
Averaging \cref{eq:wigner_gauss} over $\theta$ and substituting
$\vartheta=\theta-\phi$ yields
\begin{align}
g_\psi(\xi) &= \frac{1}{2\pi}\int_0^{2\pi} e^{-Q(\vartheta;\, \xi)}\,\dif\vartheta\ ,
\label{eq:gpsi}
\end{align}
Here, the exponent is
\begin{align}
Q(\vartheta;\, \xi) &= e^{2\varsigma}\bigl(\sqrt{\xi}\cos\vartheta - A\bigr)^2
+ e^{-2\varsigma}\bigl(\sqrt{\xi}\sin\vartheta + B\bigr)^2\ .
\label{eq:Q}
\end{align}

For coherent states, let $\zeta=0$ and $\alpha\neq0$. The Wigner
function reduces to $W_\alpha(\boldsymbol{r}) = \pi^{-1}\exp(-|\boldsymbol{r}
- \boldsymbol r_\alpha|^2)$, and using polar coordinates $\boldsymbol{r}
= \sqrt{\xi}(\cos\theta, \sin\theta)$ with $\boldsymbol r_\alpha = (a, 0)$ gives
the identity $|\boldsymbol{r}-\boldsymbol r_\alpha|^2=
\xi+a^2-2a\sqrt{\xi}\cos\theta$.
Phase averaging then yields, via the integral representation
$I_0(z) = (2\pi)^{-1}\int_0^{2\pi} e^{z\cos\theta}\,\dif\theta$
 \cite[Eq.~10.32.1]{dlmf24},
\begin{align}
g_\alpha(\xi) = \pi\,\overline{W_\alpha}(\boldsymbol{r})
= e^{-(\xi+a^2)}\,I_0(2a\sqrt{\xi}), \quad a > 0\ .
\label{eq:gcoh}
\end{align}
Normalization $\int_0^\infty g_\alpha\,\dif \xi = 1$ follows from termwise
integration of the series $I_0(2a\sqrt{\xi}) = \sum_{j\ge 0} (a^2 \xi)^j/(j!)^2$,
giving $\int_0^\infty e^{-\xi}\,\allowbreak I_0(2a\sqrt{\xi})\,\dif \xi = e^{a^2}$.

Here, $I_\nu$ denotes the modified Bessel function of the first kind. For the
non-negative orders used below, it is defined by \cite[Eq.~10.25.2]{dlmf24}
as $I_\nu(z)=\sum_{j=0}^\infty
\bigl(\frac{z}{2}\bigr)^{2j+\nu}/[j!\,\Gamma(j+\nu+1)]$ for
$z\in\mathbb C$ and $\nu\ge0$.
Since $I_0'(z)=I_1(z)>0$ for $z>0$ \cite[Eq.~10.29.3]{dlmf24},
$I_0$ is strictly increasing on $[0,\infty)$. For fixed $\nu$, its
large-argument expansion is
$I_\nu(z)=e^z(2\pi z)^{-1/2}[1+O(z^{-1})]$ as
$z\to\infty$ \cite[Eq.~10.40.1]{dlmf24}.

A straightforward calculation gives
$g_\alpha(y+\Delta)/g_\alpha(y)=
e^{-\Delta}I_0(2a\sqrt{y+\Delta})/I_0(2a\sqrt y)$.
Since $a>0$ and $\Delta>0$, strict monotonicity of $I_0$ gives
$I_0(2a\sqrt{y+\Delta})>I_0(2a\sqrt y)$. Therefore
\begin{align}
\label{eq:pointwise}
\frac{g_\alpha(y+\Delta)}{g_\alpha(y)} > e^{-\Delta} \quad \text{for all } y \ge 0,\ \Delta > 0\ .
\end{align}
Writing $p_k:=P_{\alpha,k}(\Delta)$ and translating the integral for
$p_{k+1}$ by $\Delta$ gives
$p_{k+1}-e^{-\Delta}p_k=\int_{(k-1)\Delta}^{k\Delta}
[g_\alpha(y+\Delta)-e^{-\Delta}g_\alpha(y)]\,\dif y$.
The integrand is continuous and strictly positive by \cref{eq:pointwise}, and
$p_k>0$. Hence
\begin{align}
\label{eq:shellbound}
\frac{p_{k+1}}{p_k}>e^{-\Delta}
\qquad\text{for every }k\ge1\ .
\end{align}
 
The large-argument behavior of the same shell ratios is
\begin{align}
\lim_{k\to\infty}\frac{p_{k+1}}{p_k}=e^{-\Delta}\ .
\label{eq:coherent-shell-ratio-limit}
\end{align}
To verify this, write the exact decomposition
$g_\alpha(\xi)=h(\xi)e^{-\xi}$, where
$h(\xi)=e^{-a^2}I_0(2a\sqrt{\xi})$ and
$\dif\log h(\xi)/\dif\xi=
[a/\sqrt{\xi}]I_1(2a\sqrt{\xi})/\allowbreak I_0(2a\sqrt{\xi})=O(\xi^{-1/2})$.
For fixed $\Delta$ and $u\in[0,\Delta]$, integration of this logarithmic
derivative gives, uniformly in $u$,
$h((k-1)\Delta+u)=h((k-1)\Delta)[1+O(k^{-1/2})]$. It follows that
\begin{align*}
p_k &=h\bigl((k-1)\Delta\bigr)(1-e^{-\Delta})e^{-(k-1)\Delta}
[1+O(k^{-1/2})]\ .
\end{align*}
The adjacent-shell ratio therefore satisfies
\begin{align*}
\frac{p_{k+1}}{p_k}
&=\frac{h(k\Delta)}{h((k-1)\Delta)}e^{-\Delta}
[1+O(k^{-1/2})]\longrightarrow e^{-\Delta}\ .
\end{align*}

\Cref{eq:main} gives
$\boldsymbol P_\alpha(\Delta)\prec\boldsymbol P_0(\Delta)$. Suppose, for contradiction,
that the converse relation also holds. Antisymmetry then implies
$\boldsymbol{P}_\alpha^{\downarrow}(\Delta) = \boldsymbol{P}_0(\Delta)$.
Because the vacuum shell probabilities are all distinct, there is a bijection
$\pi:\mathbb N\to\mathbb N$ such that
$p_k=(1-e^{-\Delta})e^{-[\pi(k)-1]\Delta}$. Consequently,
\begin{align}
\frac{p_{k+1}}{p_k}
=e^{-[\pi(k+1)-\pi(k)]\Delta}
\in \{e^{-m\Delta}:m\in\Z\}\quad\forall k\ .
\label{eq:cohdiscrete}
\end{align}
By \cref{eq:shellbound,eq:coherent-shell-ratio-limit}, the ratio is larger than $e^{-\Delta}$ for
every $k$ and converges to $e^{-\Delta}$. Thus, there is an integer $K$ such that
$p_{k+1}/p_k < (1+e^{-\Delta})/2 < 1$ for all $k \ge K$. These facts imply that
$p_{k+1}/p_k \in (e^{-\Delta}, 1)$ whenever $k \ge K$.
 
No factor $e^{-m\Delta}$ with integer $m$ lies strictly between
$e^{-\Delta}$ and $1$. This contradicts \cref{eq:cohdiscrete}. Hence
$\boldsymbol P_\alpha(\Delta)\prec\boldsymbol P_0(\Delta)$ and
$\boldsymbol P_0(\Delta)\nprec\boldsymbol P_\alpha(\Delta)$.

For squeezed states with $\varsigma>0$ and arbitrary
$\boldsymbol r_\alpha$, the same contradiction follows once the ratios
of consecutive shell probabilities are shown to converge to
$\exp[-e^{-2\varsigma}\Delta]$. Because $0<e^{-2\varsigma}<1$, this limit
lies strictly between $e^{-\Delta}$ and $1$, where no ratio allowed by a
permutation of the vacuum vector can occur. We obtain the required
limit from a Laplace-type analysis of \cref{eq:gpsi,eq:Q} at large radius.

For sufficiently large $\xi$, the function
$\vartheta\mapsto Q(\vartheta;\xi)$ has two
minima in $[0,2\pi)$, located at
$\vartheta_* = \pi/2+\epsilon$ and
$\vartheta_{**}=3\pi/2-\epsilon$, with the
asymptotic deviation
\begin{align}
\epsilon = -\frac{e^{2\varsigma} A}{2\sqrt{\xi}\,\sinh 2\varsigma} + O(\xi^{-1})\ .
\label{eq:eps}
\end{align}

To see this, expand $Q(\vartheta;\xi)$ in powers of $\sqrt{\xi}$ as
$Q(\vartheta;\xi)=\xi F(\vartheta)+2\sqrt{\xi}G(\vartheta)+H$, with
$F(\vartheta):=e^{2\varsigma}\cos^2\vartheta
+e^{-2\varsigma}\sin^2\vartheta$,
$G(\vartheta):=-e^{2\varsigma}A\cos\vartheta
+e^{-2\varsigma}B\sin\vartheta$, and
$H:=e^{2\varsigma}A^2+e^{-2\varsigma}B^2$.
The critical points of $F$ satisfy
$F'(\vartheta)=-2\sinh 2\varsigma\,\sin 2\vartheta=0$, yielding
$\vartheta=0,\pi/2,\pi,3\pi/2$. Since
$F''(\vartheta)=-4\sinh 2\varsigma\,\cos 2\vartheta$,
the minima of $F$ lie at $\pi/2$ and $3\pi/2$ with common value
$F(\pi/2) = F(3\pi/2) = e^{-2\varsigma} =: c$ and $F''(\pi/2) = 4\sinh 2\varsigma > 0$.

Apply the implicit function theorem \cite[Theorem 9.28]{rudin76} to
$\xi^{-1}\partial_\vartheta Q(\vartheta;\xi)=0$ with parameter
$t=\xi^{-1/2}$. Since
$F''(\pi/2) \neq 0$, there exists $\xi_0 > 0$ such that for all $\xi > \xi_0$ the
equation $\partial_\vartheta Q(\vartheta;\xi)=0$ has a unique smooth solution
$\vartheta_*(\xi)$ in a neighborhood of $\pi/2$, with
$\vartheta_*(\xi)\to\pi/2$ as
$\xi \to \infty$. The same holds near $3\pi/2$, yielding a second minimum
$\vartheta_{**}(\xi)$. Moreover,
$\xi^{-1}\partial_\vartheta Q=F'+2\xi^{-1/2}G'$ converges uniformly to $F'$.
Outside fixed neighborhoods of the four simple zeros of $F'$, it is therefore
bounded away from zero for sufficiently large $\xi$. The two remaining critical
points, near $0$ and $\pi$, are maxima. Thus, $\vartheta_*$ and
$\vartheta_{**}$ are the only minima.

To compute the asymptotic expansion, write
$\vartheta_*=\pi/2+\epsilon$. Then
$\cos\vartheta_*=-\epsilon+O(\epsilon^3)$ and
$\sin\vartheta_*=1+O(\epsilon^2)$. Substituting into
$\partial_\vartheta Q=0$ gives
$0=4\xi\epsilon\sinh 2\varsigma
+2\sqrt{\xi}\,e^{2\varsigma}A+O(1)$
to leading order,
yielding \cref{eq:eps}. An analogous expansion at
$\vartheta_{**}=3\pi/2+\epsilon'$ gives $\epsilon'=-\epsilon$.

Let $c=e^{-2\varsigma}$. At the first minimum, the exponent has the asymptotic
value
\begin{align}
Q(\vartheta_*;\xi) = c\,\xi + 2cB\sqrt{\xi} + C_0 + O(\xi^{-1/2})\ .
\label{eq:Qstar}
\end{align}
At the second minimum,
\begin{align*}
Q(\vartheta_{**};\xi) = c\,\xi - 2cB\sqrt{\xi} + C_0 + O(\xi^{-1/2})\ .
\end{align*}
Here, $C_0=e^{-2\varsigma} B^2-A^2/(2\sinh 2\varsigma)$.

To obtain these expressions, use
$\cos\vartheta_*=-\epsilon+O(\epsilon^3)$, which gives
$(\sqrt{\xi}\cos\vartheta_*-A)=e^{-2\varsigma}A/(2\sinh 2\varsigma)
+O(\xi^{-1/2})$ and hence
$e^{2\varsigma}(\sqrt{\xi}\cos\vartheta_*-A)^2=e^{-2\varsigma}A^2/
(4\sinh^2 2\varsigma)+O(\xi^{-1/2})$.
For the second term, the $O(\epsilon^2)$ contribution from
$\sin\vartheta_*=1-\epsilon^2/2+O(\epsilon^4)$ must be retained.
Thus $\sqrt{\xi}\sin\vartheta_*+B=\sqrt{\xi}+B-\sqrt{\xi}\epsilon^2/2+
O(\xi^{-3/2})$ and
$e^{-2\varsigma}(\sqrt{\xi}\sin\vartheta_*+B)^2=e^{-2\varsigma}
(\xi+2B\sqrt{\xi}+B^2-\xi\epsilon^2)+O(\xi^{-1/2})$.
The order-one contribution is
$-e^{-2\varsigma}\xi\epsilon^2=-e^{2\varsigma}A^2/
(4\sinh^2 2\varsigma)+O(\xi^{-1/2})$. Thus
$C_0=e^{-2\varsigma}B^2+
(e^{-2\varsigma}A^2-e^{2\varsigma}A^2)/(4\sinh^2 2\varsigma)
=e^{-2\varsigma}B^2-A^2/(2\sinh 2\varsigma)$, proving \cref{eq:Qstar}. For
$\vartheta_{**}=3\pi/2-\epsilon$, substituting
$\sin\vartheta_{**}=-1+\epsilon^2/2$ flips the sign of the $O(\sqrt{\xi})$
cross-term, yielding the displayed asymptotic formula for
$Q(\vartheta_{**};\xi)$.

The curvatures at both minima satisfy
$Q''(\vartheta_*;\xi)=Q''(\vartheta_{**};\xi)
=4\xi\sinh 2\varsigma[1+O(\xi^{-1/2})]$.
Indeed, differentiating \cref{eq:Q} gives
$Q''(\vartheta;\xi)=-4\xi\sinh 2\varsigma\cos2\vartheta+
2\sqrt{\xi}[e^{2\varsigma}A\cos\vartheta
-e^{-2\varsigma}B\sin\vartheta]$.
At $\vartheta_*=\pi/2+\epsilon$, one has
$\cos2\vartheta_*=-1+O(\xi^{-1})$, while the second term is
$O(\sqrt{\xi})$. Therefore
$Q''(\vartheta_*;\xi)=4\xi\sinh 2\varsigma+O(\sqrt{\xi})$. The same
estimate holds at $\vartheta_{**}$.

The Laplace estimate is uniform over the two wells. Indeed, choose disjoint
fixed neighborhoods $U_*$ and $U_{**}$ of $\pi/2$ and $3\pi/2$. Since $F$
has precisely two strict non-degenerate minima with value $c$, there are
$\eta,c_0>0$ such that, for all sufficiently large $\xi$,
\begin{align*}
Q(\vartheta;\xi)&\ge c\xi+\frac{\eta}{2}\xi
\qquad(\vartheta\notin U_*\cup U_{**})\ .
\end{align*}
Within the two neighborhoods,
\begin{align*}
Q''(\vartheta;\xi)&\ge c_0\xi
\qquad(\vartheta\in U_*\cup U_{**})\ .
\end{align*}
The first bound follows from
$Q=\xi F+2\sqrt{\xi}G+H$ and the positive gap of $F$ away from its minima.
In either neighborhood, set
$x=\sqrt{\xi}\,[\vartheta-\vartheta_\bullet(\xi)]$, where
$\vartheta_\bullet$ is the corresponding minimum. Taylor expansion about
that minimum, together with
$\partial_\vartheta^jQ=O(\xi)$ for $j=3,4$, gives a Gaussian integrand with
a relative $O(\xi^{-1/2})$ remainder, uniformly for the two minima.
Gaussian domination supplied by the second bound gives
\begin{align*}
\int_{U_\bullet}e^{-Q(\vartheta;\xi)}\,\dif\vartheta
=e^{-Q(\vartheta_\bullet;\xi)}
\sqrt{\frac{2\pi}{Q''(\vartheta_\bullet;\xi)}}
\bigl[1+O(\xi^{-1/2})\bigr].
\end{align*}
The contribution outside $U_*\cup U_{**}$ is exponentially smaller than
the sum of these terms. The same estimates remain uniform when $\xi$ is
replaced by $\xi+u$ with $u$ in a fixed interval.

Applying this estimate at the two non-degenerate minima gives
$g_\psi(\xi)\sim(2\pi)^{-1}
\sqrt{2\pi/(4\xi\sinh 2\varsigma)}
[e^{-Q(\vartheta_*;\xi)}+e^{-Q(\vartheta_{**};\xi)}]
[1+O(\xi^{-1/2})]$.
Inserting the two asymptotic formulas above gives
$e^{-Q(\vartheta_*)}+e^{-Q(\vartheta_{**})}
=2e^{-c\xi-C_0}\cosh(2cB\sqrt{\xi})[1+O(\xi^{-1/2})]$ and hence
\begin{align}
g_\psi(\xi) = \frac{e^{-C_0-c\xi}}{\sqrt{2\pi \xi \sinh 2\varsigma}}\,
\cosh(2cB\sqrt{\xi})\bigl[1 + O(\xi^{-1/2})\bigr]\ .
\label{eq:asympA}
\end{align}

No differentiation of the asymptotic remainder is needed. Write
$g_\psi(\xi)=e^{-c\xi}h_\psi(\xi)$. By \cref{eq:asympA},
\begin{align*}
h_\psi(\xi)=
\frac{e^{-C_0}}{\sqrt{2\pi\sinh 2\varsigma}}\,
\xi^{-1/2}\cosh(2cB\sqrt{\xi})
\bigl[1+O(\xi^{-1/2})\bigr]\ .
\end{align*}
For fixed $u\in[0,\Delta]$,
$\sqrt{\xi+u}-\sqrt{\xi}=O(\xi^{-1/2})$ uniformly in $u$.
Since $|\dif(\log\cosh x)/\dif x|=|\tanh x|\le1$, the explicit leading
factor satisfies
\begin{align*}
\frac{(\xi+u)^{-1/2}\cosh(2cB\sqrt{\xi+u})}
{\xi^{-1/2}\cosh(2cB\sqrt{\xi})}
=1+O(\xi^{-1/2})
\end{align*}
uniformly on this interval. The relative remainder in \cref{eq:asympA}
therefore gives
$h_\psi(\xi+u)/h_\psi(\xi)=1+O(\xi^{-1/2})$, again uniformly.

For the squeezed state, set $p_k:=P_{\psi,k}(\Delta)$. Then
$\lim_{k\to\infty}p_{k+1}/p_k=e^{-c\Delta}$.
To see this, set $\xi_k=(k-1)\Delta$. Substituting
$g_\psi(\xi)=h_\psi(\xi)e^{-c\xi}$ into the shell integral gives
\begin{align*}
p_k &= e^{-c\,\xi_k}\int_0^\Delta h_\psi(\xi_k+u)e^{-c\,u}\,\dif u
=h_\psi(\xi_k)\frac{1-e^{-c\Delta}}{c}e^{-c\,\xi_k}
\bigl[1+O(k^{-1/2})\bigr]\ .
\end{align*}
The adjacent-shell ratio is
\begin{align*}
\frac{p_{k+1}}{p_k}
&=\frac{h_\psi(k\Delta)}{h_\psi((k-1)\Delta)}
e^{-c\Delta}\bigl[1+O(k^{-1/2})\bigr]
\longrightarrow e^{-c\Delta}\ .
\end{align*}
Here, the preceding ratio estimate gives
$h_\psi(\xi_k+u)=h_\psi(\xi_k)[1+O(k^{-1/2})]$ uniformly for
$u\in[0,\Delta]$, while
$h_\psi(k\Delta)/h_\psi((k-1)\Delta)=1+O(k^{-1/2})\to1$.

To exclude the converse relation, write $\ket\psi=D(\alpha)S(\zeta)\ket0$ with
$(\alpha,\zeta)\neq(0,0)$. The coherent case $\zeta=0$ is established earlier in this section. Suppose now that $|\zeta|>0$. \Cref{eq:main} gives
$\boldsymbol P_\psi(\Delta)\prec\boldsymbol P_0(\Delta)$. If the converse
majorization relation also held, antisymmetry
would imply $\boldsymbol P_\psi^\downarrow(\Delta)=\boldsymbol
P_0(\Delta)$. As in the coherent case, this would require
$p_{k+1}/p_k\in\{e^{-m\Delta}:m\in\mathbb Z\}$ for every $k$.
On the other hand, the preceding limit is $e^{-c\Delta}$, where
$c=e^{-2|\zeta|}\in(0,1)$.
The limit belongs to $(e^{-\Delta},1)$, so the ratios themselves lie in this interval
for all sufficiently large $k$. This is impossible because
$\{e^{-m\Delta}:m\in\mathbb Z\}\cap(e^{-\Delta},1)=\varnothing$. Hence
$\boldsymbol P_0(\Delta)\nprec\boldsymbol P_\psi(\Delta)$, proving the strict
relation for every non-vacuum Gaussian pure state and every $\Delta>0$.
 
As a consistency check, for a squeezed vacuum ($a=0$), \cref{eq:asympA} reduces to
$g_{\varsigma}(\xi)\sim(2\pi \xi\sinh 2\varsigma)^{-1/2}e^{-c\xi}$, in agreement
with the large-argument expansion of the exact expression
$g_{\varsigma}(\xi)=e^{-\xi\cosh 2\varsigma}I_0(\xi\sinh 2\varsigma)$.
 
\smallskip
\noindent\textit{General radial partitions.}
The strictness extends to every complete radial partition $\Pi$ with bounded
squared-radius widths. For a non-vacuum single-mode Gaussian pure state in the
fixed oscillator frame, we prove
\begin{align}
\sum_{j=1}^{m}P_{\psi,j}^{\downarrow}(\Pi)
<1-e^{-D_m(\Pi)},\qquad m\ge1\ .
\label{eq:strict-gaussian-general}
\end{align}
The benchmark is $\boldsymbol P_0(\Pi^\downarrow)$ as defined in
\cref{eq:benchmark}. For $\Pi=\Pi_\Delta$, this gives a strict inequality
at every finite partial sum of the equal-area Lorenz curves.

The Gaussian Wigner function is strictly positive and normalized, and a
complete partition therefore gives a positive summable sequence of shell
probabilities. For any $c>0$, only finitely many entries can be at least $c$.
Choosing $c$ to be one of the remaining entries shows that their largest
value is attained. Repeating this argument after removing each selected
entry yields a finite index set $I_m$ that attains the sum of the $m$ largest
probabilities. Define
\begin{align}
A_{I_m}:=\bigcup_{i\in I_m}A_i(\Pi),\qquad
 a_m:=\sum_{i\in I_m}d_i\le D_m(\Pi)\ .
\end{align}
The union has area $\pi a_m$ and is invariant under rotations about the
oscillator origin. No attainment of the supremum defining $D_m(\Pi)$ is
assumed.

Write the pure Gaussian Wigner function in \cref{eq:wigner_gauss} as
\begin{align}
W_\psi(\boldsymbol r)
=\pi^{-1}\exp[-(\boldsymbol r-\boldsymbol r_\alpha)^{\mathsf T}
G(\boldsymbol r-\boldsymbol r_\alpha)],\qquad
G=(2\Sigma)^{-1}>0,\quad \det G=1\ .
\end{align}
For $a>0$, the ellipse
\begin{align}
E_a:=\{\boldsymbol r:(\boldsymbol r-\boldsymbol r_\alpha)^{\mathsf T}
G(\boldsymbol r-\boldsymbol r_\alpha)<a\}
\end{align}
has area $\pi a$. The change of variables
$\boldsymbol u=G^{1/2}(\boldsymbol r-\boldsymbol r_\alpha)$ has unit Jacobian
and gives $\int_{E_a}W_\psi\,\dif\boldsymbol r=1-e^{-a}$.
Moreover, $E_a$ uniquely maximizes the integral of $W_\psi$ over regions
of this area, up to sets of measure zero. Indeed, for any measurable $A$
with $|A|=|E_a|=\pi a$, put $t_a=\pi^{-1}e^{-a}$ and write
\begin{align}
\int_{E_a}W_\psi\,\dif\boldsymbol r-\int_AW_\psi\,\dif\boldsymbol r
=\int_{\mathbb R^2}(\mathbf 1_{E_a}-\mathbf 1_A)
(W_\psi-t_a)\,\dif\boldsymbol r\ge0\ .
\end{align}
The integrand is positive on the symmetric difference, except on the
ellipse boundary, which has measure zero. Equality therefore requires
$A=E_a$ almost everywhere.

For $A=A_{I_m}$, such equality would make $E_{a_m}$ invariant under
rotations about the origin up to null sets. Its centroid is
$\boldsymbol r_\alpha$, whereas rotational invariance forces that centroid
to vanish. With $\boldsymbol r_\alpha=0$, its normalized second-moment
matrix is
\begin{align}
\frac{1}{\pi a_m}\int_{E_{a_m}}
\boldsymbol r\boldsymbol r^{\mathsf T}\,\dif\boldsymbol r
=\frac{a_m}{4}G^{-1}\ .
\end{align}
Rotational invariance requires this matrix to be proportional to the
identity. Since $G>0$ and $\det G=1$, it follows that $G=\mathbb I$, which
identifies the state as the vacuum. Every non-vacuum Gaussian pure state
therefore satisfies
\begin{align}
\sum_{j=1}^{m}P_{\psi,j}^{\downarrow}(\Pi)
=\int_{A_{I_m}}W_\psi\,\dif\boldsymbol r
<1-e^{-a_m}\le1-e^{-D_m(\Pi)}\ .
\end{align}
This proves \cref{eq:strict-gaussian-general}, including partitions for
which the area supremum $D_m(\Pi)$ is not attained.

\section{Majorization under phase-insensitive Gaussian channels}\label{app:channel}

This section proves \cref{eq:channel-general-app} for a complete radial
partition with bounded squared-radius widths. Its equal-area specialization is
\cref{eq:channel-maj}. Up to a phase rotation, every single-mode
phase-insensitive Gaussian channel admits the canonical decomposition \cite{mari14}
\begin{align}
\Phi=\mathcal A_\kappa\circ\mathcal E_\eta\ ,
\label{eq:channel-decomp}
\end{align}
where $\mathcal E_\eta$ is a quantum-limited attenuator with
$0\le\eta\le1$ and $\mathcal A_\kappa$ is a quantum-limited amplifier with
$\kappa\ge1$. Both are covariant under phase rotations \cite{weedbrook12}.
Set $c_\Phi:=(2\kappa-1)^{-1}\in(0,1]$. Since
$\mathcal E_\eta(\ket0\bra0)=\ket0\bra0$, the vacuum output has radial
Wigner density
\begin{align}
g_{\Phi(\ket0\bra0)}(\xi)=c_\Phi e^{-c_\Phi\xi}\ .
\label{eq:channel-vacuum-profile}
\end{align}
For $D_m(\Pi)$ defined in \cref{eq:Dm-app}, define the output-vacuum
benchmark components by
$P_{\Phi(\ket0\bra0),m}(\Pi^\downarrow):=
e^{-c_\Phi D_{m-1}(\Pi)}-e^{-c_\Phi D_m(\Pi)}$.
The complete benchmark is
\begin{align}
\boldsymbol P_{\Phi(\ket0\bra0)}(\Pi^\downarrow)
&:=\bigl(P_{\Phi(\ket0\bra0),m}(\Pi^\downarrow)\bigr)_{m\ge1}\ .
\label{eq:channel-output-benchmark}
\end{align}
This definition remains meaningful when the width suprema are not attained.
Indeed, let $\delta_m=D_m-D_{m-1}$. Monotonicity and concavity of $D_m$
imply $\delta_m\ge0$ and $\delta_{m+1}\le\delta_m$, while completeness
gives $D_m\to\infty$.
Consequently,
\begin{align*}
P_{\Phi(\ket0\bra0),m+1}(\Pi^\downarrow)
&=e^{-c_\Phi D_m}(1-e^{-c_\Phi\delta_{m+1}})\\
&\le e^{-c_\Phi D_{m-1}}(1-e^{-c_\Phi\delta_m})
=P_{\Phi(\ket0\bra0),m}(\Pi^\downarrow)\ .
\end{align*}
The components are normalized because
\begin{align}
\sum_{j=1}^{M}P_{\Phi(\ket0\bra0),j}(\Pi^\downarrow)
&=1-e^{-c_\Phi D_M}\xrightarrow[M\to\infty]{}1\ .
\label{eq:channel-benchmark-properties}
\end{align}
Thus \cref{eq:channel-output-benchmark} is a non-negative decreasing
probability vector. For the first $N$ widths, let $D_m^{(N)}$ be the sum of
their $m$ largest members. Placing these widths consecutively from the origin
gives the vacuum-output cumulative weight $1-e^{-c_\Phi D_m^{(N)}}$.
Since $D_m^{(N)}\uparrow D_m$, this converges to the cumulative sum in
\cref{eq:channel-benchmark-properties}. Hence the benchmark is the limit of
finite reorderings; it is the shell vector of an actual reordered partition
whenever the widths admit a non-increasing enumeration.

Because the shells are rotation invariant, their weights depend only on
the phase-averaged input. Pure loss requires no additional estimate:
$\mathcal E_\eta$ fixes the vacuum, so the main majorization relation applied to
$\mathcal E_\eta(\rho)$ gives the desired majorization relation directly.
It remains to
prove the corresponding statement for the quantum-limited amplifier.

\smallskip
\noindent\textit{Output majorization.} Let $\Phi$ be a
single-mode phase-insensitive Gaussian channel and let
$\Pi$ be a complete radial partition with bounded squared-radius widths. If
$\boldsymbol P_{\Phi(\rho)}(\Pi)$ is componentwise non-negative, then it is a
probability vector and
\begin{align}
\boldsymbol P_{\Phi(\rho)}(\Pi)
\prec
\boldsymbol P_{\Phi(\ket0\bra0)}(\Pi^\downarrow)\ .
\label{eq:channel-general-app}
\end{align}
For the equal-area partition $\Pi_\Delta$, this reduces to
\begin{align}
\boldsymbol P_{\Phi(\rho)}(\Delta)
\prec\boldsymbol P_{\Phi(\ket0\bra0)}(\Delta)\ .
\label{eq:channel-maj}
\end{align}
Here $D_m=m\Delta$ and the output-vacuum vector on the original partition is
\begin{align}
P_{\Phi(\ket0\bra0),m}(\Delta)
=(1-e^{-c_\Phi\Delta})e^{-c_\Phi(m-1)\Delta}\ .
\label{eq:channel-equal-vacuum}
\end{align}

The Fock-state Wigner convention is
$g_m(\xi):=\pi W_m(\boldsymbol r)|_{|\boldsymbol r|^2=\xi}
=(-1)^mL_m(2\xi)e^{-\xi}$, with $\int_0^\infty g_m(\xi)\,\dif\xi=1$.
The calculation repeatedly uses the Laguerre generating
function \cite[Eq.~18.12.13]{dlmf24}
\begin{align}
\sum_{m=0}^{\infty}w^mL_m(z)=\frac{1}{1-w}
\exp\!\left(\frac{zw}{w-1}\right),\qquad |w|<1\ .
\label{eq:channel-LGF}
\end{align}

\smallskip
\noindent\textit{Amplifier output.}
For the quantum-limited amplifier $\mathcal A_\kappa$ with $\kappa>1$, set
$\lambda=(\kappa-1)/\kappa$ and
$c:=c_\Phi=(1-\lambda)/(1+\lambda)=(2\kappa-1)^{-1}\in(0,1)$. Its transition
probabilities in the Fock basis are
$p_{n|m}^{\mathcal A_\kappa}
=\binom{n}{m}(1-\lambda)^{m+1}\lambda^{n-m}$ for $n\ge m$.
This is the negative-binomial Fock kernel of the quantum-limited
amplifier \cite{ivan11}.
Consider the double generating function
$G_{\mathcal A}(\xi,t):=\sum_{m=0}^\infty
g_{\mathcal A_\kappa(\ket{m}\bra{m})}(\xi)t^m$. For $|t|<1$, the Wigner
bound $|g_n(\xi)|\le1$ and normalization of the amplifier transition
probabilities give
\begin{align*}
\sum_{m=0}^\infty |t|^m\sum_{n=m}^\infty
p_{n|m}^{\mathcal A_\kappa}|g_n(\xi)|
\le\frac{1}{1-|t|}.
\end{align*}
Thus the double series is absolutely convergent and its order may be
exchanged. Moreover,
$|\lambda+(1-\lambda)t|\le\lambda+(1-\lambda)|t|<1$, so the resulting
Laguerre series remains inside the convergence disk of
\cref{eq:channel-LGF}. Exchanging the order of
summation gives
\begin{align}
G_{\mathcal A}(\xi,t)
&=e^{-\xi}\sum_{n=0}^{\infty}(-1)^nL_n(2\xi)
\sum_{m=0}^{n}\binom{n}{m}(1-\lambda)^{m+1}\lambda^{n-m}t^m \notag\\
&=(1-\lambda)e^{-\xi}
\sum_{n=0}^{\infty}[-\lambda-(1-\lambda)t]^nL_n(2\xi)\ .
\end{align}
Applying \cref{eq:channel-LGF} with
$w=-[\lambda+(1-\lambda)t]$ yields
$G_{\mathcal A}(\xi,t)=c(1+ct)^{-1}
\exp[-c\xi(1-t)/(1+ct)]$.
The $t=0$ term is $g_{\mathcal A_\kappa(\ket0\bra0)}(\xi)=ce^{-c\xi}$.
Dividing by this term and applying \cref{eq:channel-LGF} again, now with
$w=-ct$ and $z=(1+c)\xi$, gives
\begin{align}
g_{\mathcal A_\kappa(\ket{m}\bra{m})}(\xi)
=(-c)^mL_m((1+c)\xi)\,ce^{-c\xi}\ .
\label{eq:channel-amp-profile}
\end{align}
Thus, after the rescaling $y=c\xi$, the vacuum output becomes $e^{-y}$
and the $m$th output becomes
$\tilde g_m(y)=(-c)^mL_m[(1+c)y/c]e^{-y}$.

\smallskip
\noindent\textit{Cumulative weights and tail monotonicity.}
With $\beta=2c/(1+c)$, the function below measures the difference
between the cumulative weights
of the amplified vacuum and the $m$th output after rescaling.
We establish its non-negativity through a representation that also yields
the tail monotonicity needed for majorization.
\begin{align}
F_m(T):=
\int_0^T
\left[1-(-c)^mL_m\!\left(\frac{2y}{\beta}\right)\right]e^{-y}\,\dif y
\ge0\ .
\label{eq:channel-F-detailed}
\end{align}
For $m=0$, this is an equality. Hence assume $m\ge1$ below.

\smallskip
\noindent\textit{Positivity for $0<c\le1$.}
Set $Q_n(z,c):=\sum_{j=0}^n(-c)^jL_j(z)$ and
$H_m(z,c):=Q_m(z,c)+cQ_{m-1}(z,c)-1$.
First, $H_m(0,c)=0$ because $L_j(0)=1$ and
$Q_n(0,c)=[1-(-c)^{n+1}]/(1+c)$.
Second, $F_m(T)=e^{-T}H_m(2T/\beta,c)$. Indeed, if
$\Phi(T)=e^TF_m(T)$, then
$\Phi'(T)=\Phi(T)+1-(-c)^mL_m(2T/\beta)$, with $\Phi(0)=0$.
It remains to verify that $\widetilde H(T):=H_m(2T/\beta,c)$ satisfies
the same initial-value problem. With $Q_{-1}:=0$, the derivative identity
\begin{align}
\partial_zH_m(z,c)=cQ_{m-1}(z,c)\ ,
\label{eq:channel-H-derivative}
\end{align}
follows from $L_k'(z)=-L_{k-1}^{(1)}(z)$ and
$L_n^{(1)}(z)=\sum_{j=0}^nL_j(z)$. The derivative formula is
Eq.~18.9.23 of Ref. \cite{dlmf24}. The finite sum follows by comparing
the ordinary and associated Laguerre generating functions. After interchanging the finite sums
in $\partial_zQ_n$ and combining the terms in
$\partial_zQ_m+c\,\partial_zQ_{m-1}$, all associated-Laguerre terms
cancel except $c(-c)^{m-1}L_{m-1}$, giving
$\partial_zH_m=cQ_{m-2}+c(-c)^{m-1}L_{m-1}=cQ_{m-1}$.
Hence $\widetilde H'(T)=(2c/\beta)Q_{m-1}(2T/\beta,c)
=(1+c)Q_{m-1}(2T/\beta,c)$.
Using $Q_m=Q_{m-1}+(-c)^mL_m$, one checks
$\widetilde H(T)+1-(-c)^mL_m(2T/\beta)
=(1+c)Q_{m-1}(2T/\beta,c)=\widetilde H'(T)$.
Uniqueness of the linear ODE then proves $F_m=e^{-T}H_m$.

It remains to show $H_m\ge0$. Abel summation gives
\begin{align}
Q_n(z,c)=c^nQ_n(z,1)+(1-c)\sum_{j=0}^{n-1}c^jQ_j(z,1)\ .
\label{eq:channel-abel}
\end{align}
The sums $Q_j(z,1)$ are non-negative directly from the Mehler
decomposition,
$Q_j(z,1)=\sum_{\ell=0}^j(-1)^\ell L_\ell(z)=S_j(z/2)\ge0$,
by \cref{eq:mehler-sos-app}. Since the coefficients in
\cref{eq:channel-abel} are non-negative for $c\in[0,1]$,
$Q_n(z,c)\ge0$.
Together with \cref{eq:channel-H-derivative} and $H_m(0,c)=0$, this
implies $H_m(z,c)\ge0$, and hence $F_m(T)\ge0$.

\smallskip
\noindent\textit{From tail inequalities to majorization.}
Define $\tilde\Psi_m(T)=e^T\int_T^\infty\tilde g_m(y)\,\dif y$.
Since
$\int_0^T\tilde g_m(y)\,\dif y=(1-e^{-T})-F_m(T)$,
$\tilde\Psi_m(T)=1+e^TF_m(T)=1+H_m(2T/\beta,c)$ and therefore
$\tilde\Psi_m'(T)=(1+c)Q_{m-1}(2T/\beta,c)\ge0$.
The last inequality is the positivity just proved for $0<c\le1$.

Let $\Pi$ have squared-radius endpoints $0=x_0<x_1<\cdots$, widths
$d_k=x_k-x_{k-1}$, and set $y_k=cx_k$ and
$\widetilde d_k:=y_k-y_{k-1}=cd_k$. Denote the rescaled partition by
$\widetilde\Pi$ and define the tails
$\widetilde T_m(k):=\int_{y_{k-1}}^\infty\tilde g_m(y)\,\dif y$. The
associated shell weights are
$\widetilde P_m(k):=\widetilde T_m(k)-\widetilde T_m(k+1)$.
Thus
$\widetilde{\boldsymbol P}_m(\widetilde\Pi)
:=(\widetilde P_m(k))_{k\ge1}$; by the change of variables $y=c\xi$,
$\widetilde P_m(k)=P_{\mathcal A_\kappa(\ket m\bra m),k}(\Pi)$ for every
$k$.
Since
$\tilde\Psi_m(y_{k-1})=e^{y_{k-1}}\widetilde T_m(k)$ and
$\tilde\Psi_m(y_k)=e^{y_k}\widetilde T_m(k+1)$,
monotonicity of $\tilde\Psi_m$ gives the tail inequality
\begin{align}
\widetilde T_m(k+1)
\ge e^{-\widetilde d_k}\widetilde T_m(k)\ .
\label{eq:channel-tail-detailed}
\end{align}
If $\widetilde{\boldsymbol P}_m(\widetilde\Pi)$ is componentwise
non-negative, the normalization lemma \cref{lem:shell-normalization-app},
applied to the physical output state before the change of variable, shows
that these weights sum to one. Hence $\widetilde T_m(k)$ are the tails of a
probability vector, and the subset-sum argument leading to
\cref{eq:variable-tail-majorization} gives
\begin{align}
\widetilde{\boldsymbol P}_m(\widetilde\Pi)
\prec
\widetilde{\boldsymbol P}_0(\widetilde\Pi^\downarrow)\ .
\label{eq:channel-fock-general}
\end{align}
The rescaling gives
\begin{align}
\widetilde D_\ell
:=\sup_{|I|=\ell}\sum_{i\in I}\widetilde d_i
=cD_\ell(\Pi),
\label{eq:channel-rescaled-order-statistics}
\end{align}
and therefore
\begin{align}
\widetilde P_{0,\ell}(\widetilde\Pi^\downarrow)
&=e^{-\widetilde D_{\ell-1}}-e^{-\widetilde D_\ell}\notag\\
&=e^{-cD_{\ell-1}(\Pi)}-e^{-cD_\ell(\Pi)}
=P_{\mathcal A_\kappa(\ket0\bra0),\ell}(\Pi^\downarrow),
\label{eq:channel-rescaled-benchmark}
\end{align}
which identifies the limiting benchmark even when no width supremum is
attained.

Phase covariance gives
$\boldsymbol P_{\mathcal A_\kappa(\rho)}(\Pi)=
\boldsymbol P_{\mathcal A_\kappa(\bar\rho)}(\Pi)$, where
$\bar\rho=\sum_m \rho_m\ket m\bra m$. Since the exterior region operator is
bounded, trace-norm convergence of this Fock decomposition gives
$\widetilde T_\rho(k)=\sum_m\rho_m\widetilde T_m(k)$ and
$P_{\mathcal A_\kappa(\rho),k}(\Pi)
=\widetilde T_\rho(k)-\widetilde T_\rho(k+1)$. Hence
\begin{align}
\widetilde T_\rho(k+1)
\ge e^{-\widetilde d_k}\widetilde T_\rho(k)
\qquad(k\ge1).
\label{eq:channel-mixed-tail}
\end{align}
If $\boldsymbol P_{\mathcal A_\kappa(\rho)}(\Pi)$ is componentwise
non-negative, \cref{lem:shell-normalization-app} makes it a probability
vector. Applying the subset-sum lemma with widths
$\widetilde d_k=cd_k$ and using \cref{eq:channel-rescaled-benchmark} gives
$\boldsymbol P_{\mathcal A_\kappa(\rho)}(\Pi)
\prec\boldsymbol P_{\mathcal A_\kappa(\ket0\bra0)}(\Pi^\downarrow)$.

Now let $\Phi$ be an arbitrary single-mode phase-insensitive Gaussian channel
and use the decomposition \cref{eq:channel-decomp}. Set
$\sigma:=\mathcal E_\eta(\rho)$. Since the attenuator fixes the vacuum,
$\mathcal E_\eta(\ket0\bra0)=\ket0\bra0$. If $\kappa=1$, then
$c_\Phi=1$ and \cref{eq:channel-output-benchmark} reduces to
$\boldsymbol P_0(\Pi^\downarrow)$; the claim is the main majorization
relation applied to $\sigma$. If $\kappa>1$, the amplifier result applied to
$\sigma$ gives
\begin{align}
\boldsymbol P_{\Phi(\rho)}(\Pi)
&=\boldsymbol P_{\mathcal A_\kappa(\sigma)}(\Pi)\notag\\
&\prec
\boldsymbol P_{\mathcal A_\kappa(\ket0\bra0)}(\Pi^\downarrow)
=\boldsymbol P_{\Phi(\ket0\bra0)}(\Pi^\downarrow)\ .
\end{align}
No pointwise non-negativity of the input or output Wigner function is required;
the sole sign assumption is componentwise non-negativity of
$\boldsymbol P_{\Phi(\rho)}(\Pi)$. This completes the proof
for the full phase-insensitive class.

\section{Radial negativity scale}\label{app:negativity-scale}

Throughout this section, the partition is the equal-area family
$\Pi_\Delta$. Thus, $\Delta_\rho^*$ is a resolution scale for equal-area
shells, not a parameterization of arbitrary unequal bins. This section
establishes the finiteness of $\Delta_\rho^*$ for several broad state families,
derives its semiclassical Airy asymptotics for Fock states, and shows that
finite mean energy alone does not ensure a finite scale.

\smallskip
\noindent\textit{Far-field criterion.} Let
$g_\rho(\xi)=\pi W_{\bar\rho}(\boldsymbol r)
|_{|\boldsymbol r|^2=\xi}$ be a normalized radial Wigner profile. If
$g_\rho(\xi)\ge0$ for all $\xi>\xi_{\rm ff}$ with some finite
$\xi_{\rm ff}$, then
$\Delta_\rho^*<\infty$.

Let $T(a):=\mathcal T_\rho(a)$. As shown in the proof of
\cref{eq:tail-abel-app}, $T(a)\ge0$ and
\begin{align*}
a_{\rm L}\int_0^\infty e^{-a_{\rm L}a}T(a)\,\dif a
\xrightarrow[a_{\rm L}\downarrow0]{}0.
\end{align*}
The Wigner bound makes $g_\rho$ locally integrable \cite{hillery84}. Since
$g_\rho(a)\ge0$ for $a>\xi_{\rm ff}$, the finite-interval identity
$T(a)-T(b)=\int_a^b g_\rho(u)\,\dif u$ shows that $T$ is non-increasing on
$[\xi_{\rm ff},\infty)$. Hence $T(a)\downarrow L$ for some $L\ge0$. For
every $a_{\rm L}>0$,
\begin{align*}
a_{\rm L}\int_0^\infty e^{-a_{\rm L}a}T(a)\,\dif a
\ge a_{\rm L}\int_{\xi_{\rm ff}}^\infty e^{-a_{\rm L}a}L\,\dif a
=Le^{-a_{\rm L}\xi_{\rm ff}}.
\end{align*}
Taking $a_{\rm L}\downarrow0$ and using \cref{eq:tail-abel-app} gives $L=0$.
Thus $T(a)\to0$ as $a\to\infty$.
Choose $\Delta_0>\xi_{\rm ff}$ such that $T(\Delta)<1/2$ whenever
$\Delta>\Delta_0$. Then $P_{\rho,1}(\Delta)=1-T(\Delta)>1/2$, and
every shell with $k\ge2$ lies in the non-negative region. Therefore
$\boldsymbol P_\rho(\Delta)$ is componentwise non-negative for all
$\Delta>\Delta_0$.

Every finite-stellar-rank pure state has a finite radial negativity scale. To
show this, let $\rho=\ket\psi\bra\psi$ and write
$\ket{\psi}=D(\alpha)S(\zeta)\ket{g_N}$ with
$\ket{g_N}=\sum_{n=0}^N c_n\ket{n}$ and
$c_N\neq0$ \cite{chabaud20}. This is precisely the pure-state class whose
Wigner nodal set is bounded \cite{abreu25}.

If $N=0$, then $\ket\psi$ is Gaussian and its Wigner function is pointwise
positive, so $\Delta_\rho^*=0$. Assume $N\ge1$ below.

The first step is to show that $W_{g_N}$ is positive at sufficiently large
radius. Set
$z=q+ip$ and $\xi=|z|^2=q^2+p^2$. The Hermite--Wigner formula gives \cite{hillery84}
$W_{g_N}(z)=\pi^{-1}e^{-|z|^2}[D(\xi)+X(z,\bar z)]$, with diagonal
part $D(\xi)=\sum_{n=0}^N|c_n|^2(-1)^nL_n(2\xi)$.
The interference term is
\begin{align*}
X(z,\bar z)
&=2\operatorname{Re}\sum_{k=1}^{N}\sum_{n=0}^{N-k}
c_{n+k}\bar c_n\,(-1)^n \\
&\quad{}\times
\sqrt{\frac{2^k n!}{(n+k)!}}\,\bar z^k L_n^{(k)}(2\xi)\ .
\end{align*}
The diagonal part $D$ is a polynomial of degree $N$ in $\xi$. Since
$(-1)^NL_N(2\xi)=2^N\xi^N/N!+(\text{lower order})$, its leading coefficient
is $\alpha_N:=|c_N|^2 2^N/N!>0$.
In particular, there exists $M_1>0$ such that
$|D(\xi)-\alpha_N\xi^N|\le M_1\xi^{N-1}$ for all $\xi\ge1$. Each cross
term contains $|\bar z^k|=\xi^{k/2}$ and a polynomial
$L_n^{(k)}(2\xi)$ of degree $n\le N-k$. Its modulus is therefore bounded
by $C\xi^{n+k/2}\le C\xi^{N-1/2}$ for $\xi\ge1$. Summing the finitely many
cross terms gives a constant $M_2>0$ such that
$|X(z,\bar z)|\le M_2\xi^{N-1/2}$ uniformly in the angular variable.
Consequently, for every
$\xi>R_1:=\max\bigl(1,\,(M_1+M_2)^2/\alpha_N^2\bigr)$,
\begin{align}
D(\xi)+X(z,\bar z)&\ge \alpha_N \xi^N-(M_1+M_2)\xi^{N-1/2} \notag\\
&=\xi^{N-1/2}\bigl[\alpha_N \xi^{1/2}-(M_1+M_2)\bigr]>0\ .
\end{align}

Covariance under displacement and squeezing transfers this far-field
positivity to $W_\psi$ through
$W_\psi(\boldsymbol r)=W_{g_N}\!\bigl(\mathsf S_\zeta^{-1}
(\boldsymbol r-\boldsymbol r_\alpha)\bigr)$,
where $\mathsf S_\zeta$ is the symplectic matrix associated with $S(\zeta)$ and
$\boldsymbol r_\alpha$ is the displacement vector. Let $\sigma_{\min}>0$ be
the smallest singular value of $\mathsf S_\zeta^{-1}$. Then
$|\mathsf S_\zeta^{-1}(\boldsymbol r-\boldsymbol r_\alpha)|^2
\ge\sigma_{\min}^2(|\boldsymbol r|-|\boldsymbol r_\alpha|)^2$.
Thus $W_\psi(\boldsymbol r)>0$ whenever
$|\boldsymbol r|>R_0:=|\boldsymbol r_\alpha|+\sqrt{R_1}/\sigma_{\min}$.

The phase-averaged profile inherits this property. For $R>R_0$, every
point on the circle $|\boldsymbol r|=R$ has positive $W_\psi$, and hence
$W_{\bar\rho}(R)=(2\pi)^{-1}\int_0^{2\pi}
W_\psi(R\cos\theta,R\sin\theta)\,\dif\theta>0$. Therefore
$g_\rho(\xi)>0$ for all $\xi>R_0^2$.
The far-field criterion above yields the conclusion.

For a finite mixture $\rho=\sum_iw_i\ket{\psi_i}\!\bra{\psi_i}$ of
finite-stellar-rank states, with $w_i\ge0$ and $\sum_iw_i=1$, linearity of
the shell weights gives
$\Delta_\rho^*\le\max_i\Delta_{\psi_i}^*<\infty$.

The radial negativity scale is also finite for every cat state
$\ket{\mathrm{cat}_\pm(\alpha)}\propto
(\ket{\alpha}\pm\ket{-\alpha})$ with $\alpha\ne0$. By the $U(1)$ covariance
of $\Delta_\rho^*$, assume $\alpha>0$ without loss of generality and write
$\rho_\pm:=\ket{\mathrm{cat}_\pm(\alpha)}
\bra{\mathrm{cat}_\pm(\alpha)}$.

\smallskip
\noindent\textit{Phase-averaged profile.}
Using the Fock-basis representation of the Wigner function of
$\ket n\!\bra m$ \cite{hillery84}, the cat-state Wigner function evaluates to
\begin{align}
W_{\mathrm{cat}_\pm}(\gamma)
&=\frac{1}{\pi\mathcal{N}_\pm^2}
\begin{aligned}[t]
\Bigl[&e^{-2|\gamma-\alpha|^2}+e^{-2|\gamma+\alpha|^2}\\[-0.35em]
&{}\pm2e^{-2|\gamma|^2}
\cos\bigl(4\operatorname{Im}(\gamma^*\alpha)\bigr)\Bigr]
\end{aligned}\ .
\label{eq:Wcat}
\end{align}
Here, $\mathcal{N}_\pm^2=2(1\pm e^{-2\alpha^2})$ and
$\gamma=(q+ip)/\sqrt{2}$. To retain the radial convention used throughout this
work, set $\xi=q^2+p^2=2|\gamma|^2$ and
$\gamma=\sqrt{\xi/2}\,e^{i\theta}$.
Phase averaging the Gaussian term uses $(2\pi)^{-1}\int_0^{2\pi}
e^{\pm2\sqrt{2}\alpha\sqrt{\xi}\cos\theta}\,\dif\theta
=I_0(2\sqrt{2}\alpha\sqrt{\xi})$ \cite[Eq.~10.32.1]{dlmf24}. The
interference term similarly gives $(2\pi)^{-1}\int_0^{2\pi}
\cos(2\sqrt{2}\alpha\sqrt{\xi}\sin\theta)\,\dif\theta
=J_0(2\sqrt{2}\alpha\sqrt{\xi})$ \cite[Eq.~10.9.1]{dlmf24},
where $I_0$ and $J_0$ are the modified and ordinary Bessel functions of the
first kind. Thus, the normalized radial profile
$g_\pm(\xi):=\pi W_{\bar\rho_\pm}(\boldsymbol r)|_{|\boldsymbol r|^2=\xi}$
is $g_\pm(\xi)=2[G(\xi)\pm\Xi(\xi)]/\mathcal N_\pm^2$, with
$G(\xi):=e^{-(\xi+2\alpha^2)}I_0(2\sqrt{2}\alpha\sqrt{\xi})$ and
$\Xi(\xi):=e^{-\xi}J_0(2\sqrt{2}\alpha\sqrt{\xi})$.

\smallskip
\noindent\textit{A lower bound on $G$.}
From the integral representation $I_0(z)=\tfrac{1}{\pi}\int_0^\pi e^{z\cos\theta}\,\dif\theta$ \cite[Eq.~10.32.1]{dlmf24}, the substitution $u=\sin(\theta/2)$ gives
\begin{align}
I_0(z)e^{-z}=\frac{2}{\pi}\int_0^1 \frac{e^{-2zu^2}}{\sqrt{1-u^2}}\,\dif u
\ge\frac{2}{\pi}\int_0^1 e^{-2zu^2}\,\dif u
=\frac{\operatorname{erf}(\sqrt{2z})}{\sqrt{2\pi z}}\ ,
\end{align}
valid for every $z>0$. The inequality uses $(1-u^2)^{-1/2}\ge 1$ on
$[0,1)$, followed by evaluation of the remaining Gaussian integral. Taking
$z=2\sqrt{2}\alpha\sqrt{\xi}$ and using
$-(\xi+2\alpha^2)+2\sqrt{2}\alpha\sqrt{\xi}
=-(\sqrt{\xi}-\sqrt{2}\alpha)^2$ gives
\begin{align}\label{eq:Gbound}
G(\xi)\ge
\frac{\operatorname{erf}(\sqrt{4\sqrt{2}\alpha\sqrt{\xi}})}
{\sqrt{4\sqrt{2}\pi\alpha\sqrt{\xi}}}
e^{-(\sqrt{\xi}-\sqrt{2}\alpha)^2}\ .
\end{align}

\smallskip
\noindent\textit{Interference bound.} Since $|J_0(z)|\le 1$ for
$z\in\RR$ \cite[Eq.~10.14.1]{dlmf24}, $|\Xi(\xi)|\le e^{-\xi}$.

\smallskip
\noindent\textit{Far-field dominance.} Combining \cref{eq:Gbound} with the
interference bound, write $r:=\sqrt{\xi}$ and define
$P(r):=\operatorname{erf}(\sqrt{4\sqrt{2}\alpha r})/
\sqrt{4\sqrt{2}\pi\alpha r}$.
Then
\begin{align}\label{eq:ratio}
\frac{G(\xi)}{e^{-\xi}}\ge
P(r)e^{2\sqrt{2}\alpha r-2\alpha^2}\ .
\end{align}
The prefactor is strictly positive for $r>0$. Moreover,
$P(r)\to2/\pi$ as $r\to0^+$, by the Taylor expansion of the error
function \cite[Eq.~7.6.1]{dlmf24}, whereas
$P(r)\sim(4\sqrt{2}\pi\alpha r)^{-1/2}$ as $r\to\infty$. Hence
$\log[G(\xi)e^\xi]\ge
\log P(r)+2\sqrt{2}\alpha r-2\alpha^2\to+\infty$.
Therefore there exists $r_*(\alpha)<\infty$ such that
$G(\xi)>e^{-\xi}\ge|\Xi(\xi)|$ whenever
$\sqrt{\xi}>r_*(\alpha)$. By
$g_\pm=2(G\pm\Xi)/\mathcal N_\pm^2$, the profiles are positive for all
$\xi>R_0(\alpha):=r_*(\alpha)^2<\infty$.

\smallskip
\noindent\textit{Shell non-negativity.} This eventual positivity satisfies the
far-field criterion. Hence
$\Delta_{\mathrm{cat}_\pm}^*<\infty$.

\smallskip\noindent\textit{Remark.}
Cat states show explicitly that the finiteness of $\Delta_\rho^*$ is not
restricted to finite-stellar-rank states. The argument rests on the eventual
dominance of the phase-averaged coherent peaks over the interference term.

For Fock states, Bracken et al. obtained the disk-region eigenvalue
$\lambda_n^D(a)=C_n(a^2)$ \cite{bracken99,bracken03}. The question left open
for the present purpose is the location of its last zero, which determines when
all shell weights become non-negative. Berry's Airy-edge asymptotics give
the semiclassical mechanism, and Ref. \cite{hanin20} establishes this local
scaling rigorously for harmonic-oscillator Wigner distributions. The proof
below proceeds directly from uniform Laguerre--Airy asymptotics. It derives
uniform asymptotics for the radial profile and its cumulative integral, then
locates the last zero that fixes the critical resolution in
\cref{eq:airy-law}.

For a Fock state, set
$g_n(\xi)=(-1)^nL_n(2\xi)e^{-\xi}$ and
$C_n(X)=\int_0^X g_n(\xi)\,\dif \xi$, and write
$\nu=4n+2$, $\xi_n=2n+1$, and $\ell_n=2^{-1/3}\nu^{1/3}$.
Let $\mathcal I_{\rm Ai}(t)=\int_{-\infty}^t\mathrm{Ai}(u)\,\dif u$.

\smallskip
\noindent\textit{Fock-state Airy asymptotics.} Let $t_0$ be the largest zero of
$\mathcal I_{\rm Ai}$. Then, as $n\to\infty$,
\begin{align}
\Delta_n^*=(2n+1)+2^{-1/3}t_0(4n+2)^{1/3}+O(n^{-1/3})\ .
\label{eq:airy-law-app}
\end{align}

\smallskip
\noindent\textit{Reduction to the cumulative profile.} Fix $n\ge1$. Suppose
that $C_n$ attains negative values, let
$X_0:=\sup\{X>0:C_n(X)<0\}$, and assume that $g_n\ge0$ on
$[X_0,\infty)$. Then $\Delta_n^*=X_0$.

Since $C_n$ is continuous and $C_n(X)\to1$ as $X\to\infty$, the number
$X_0$ is finite, $C_n(X_0)=0$, and $C_n(X)\ge0$ for $X\ge X_0$.
If $\Delta\ge X_0$, then the first shell has mass
$C_n(\Delta)\ge0$, and every later shell is contained in the region where
$g_n\ge0$. Thus, no shell mass is negative. Conversely, for every
$\epsilon>0$ there is $X\in(X_0-\epsilon,X_0)$ with $C_n(X)<0$, and
choosing $\Delta=X$ makes the first shell negative. This proves the claim.

The input asymptotics are the Airy transition formula for Laguerre polynomials.
Let $u=2\xi$, $\upsilon=u/\nu$, and define the Langer variable $\zeta_{\mathrm{L}}$ by
\begin{align}
\frac23\zeta_{\mathrm{L}}(\upsilon)^{3/2}
=\frac12\int_1^\upsilon\sqrt{1-\frac1w}\,\dif w
\qquad(\upsilon\ge1)\ ,
\label{eq:airy-detail-zeta-plus}\\[-0.35em]
\frac23[-\zeta_{\mathrm{L}}(\upsilon)]^{3/2}
=\frac12\int_\upsilon^1\sqrt{\frac1w-1}\,\dif w
\qquad(\upsilon<1)\ .
\label{eq:airy-detail-zeta-minus}
\end{align}
Set $\chi(\upsilon):=2^{1/2}\upsilon^{-1/4}
[\zeta_{\mathrm L}(\upsilon)/(\upsilon-1)]^{1/4}$, with the removable
singularity at $\upsilon=1$ filled in by continuity.

The only external asymptotic input is the uniform Frenzen--Wong expansion with
its error bound. Define the Airy derivative modulus
\begin{align}
\mathfrak A_1(z):=
\begin{cases}
|\mathrm{Ai}'(z)|,&z\ge0,\\
\bigl[\mathrm{Ai}'(z)^2+\mathrm{Bi}'(z)^2\bigr]^{1/2},&z<0,
\end{cases}
\label{eq:airy-derivative-modulus}
\end{align}
where $\mathrm{Bi}$ is the second Airy function. On every compact
$\upsilon$-interval contained in $(0,\infty)$,
\begin{align}
(-1)^ne^{-u/2}L_n(u)
&=\frac{\chi(\upsilon)}{\nu^{1/3}}
\mathrm{Ai}\!\left(\nu^{2/3}\zeta_{\mathrm L}(\upsilon)\right)
+\mathcal R_n(\upsilon)\ ,\notag\\
|\mathcal R_n(\upsilon)|
&\le C\nu^{-5/3}\beta_1(\upsilon)
\mathfrak A_1\!\left(\nu^{2/3}\zeta_{\mathrm L}(\upsilon)\right)\ ,
\label{eq:airy-detail-fw}
\end{align}
where $\beta_1$ is continuous and therefore bounded on that interval.
This is the $\alpha=0$ specialization of the Frenzen--Wong
formula \cite[Sec.~5, Eq.~(5.12)]{frenzenwong88}. The explicit error estimate
in this normalization is restated in Ref. \cite[Prop.~4.1]{hanin20}.

We also use the exterior estimate of
Ref. \cite[Prop.~1.7]{hanin20}. For each fixed $\delta>0$, it implies
\begin{align}
\frac{\nu}{2}\int_{1+\delta}^{\infty}
\left|(-1)^ne^{-\nu\upsilon/2}L_n(\nu\upsilon)\right|
\,\dif\upsilon
=O(e^{-c_\delta\nu})\ ,
\label{eq:airy-exterior-tail}
\end{align}
for some $c_\delta>0$. Indeed, the exterior rate in that result is strictly
positive at $1+\delta$, increases thereafter, and grows linearly as
$\upsilon\to\infty$.

Near the turning point, the Langer variable satisfies, uniformly for $t$ in
compact subsets of $\mathbb R$,
\begin{align}
Z_\nu(t):=\nu^{2/3}\zeta_{\mathrm{L}}\!\left(1+2^{2/3}\nu^{-2/3}t\right)
=t-\frac{2^{2/3}}{5}t^2\nu^{-2/3}+O(\nu^{-4/3})\ .
\label{eq:airy-detail-langer}
\end{align}
In particular, $\zeta_{\mathrm{L}}'(1)=2^{-2/3}$ and $\chi(1)=2^{1/3}$.
To verify \cref{eq:airy-detail-langer}, put $w=1+v$. Then
\begin{align*}
\sqrt{1-\frac1w}
&=\sqrt{v}\left(1-\frac v2+O(v^2)\right)\ ,\\[-0.35em]
\int_1^{1+\eta}\sqrt{1-\frac1w}\,\dif w
&=\frac23\eta^{3/2}\left(1-\frac{3\eta}{10}+O(\eta^2)\right)\ ,\\[-0.35em]
\zeta_{\mathrm{L}}(1+\eta)
&=2^{-2/3}\eta\left(1-\frac{\eta}{5}+O(\eta^2)\right)\ ,
\end{align*}
and the same expansion holds on the $\upsilon<1$ side. Taking
$\eta=2^{2/3}\nu^{-2/3}t$ proves \cref{eq:airy-detail-langer}. The limits
for $\zeta_{\mathrm{L}}'(1)$ and $\chi(1)$ follow from the same expansion and
the definition of $\chi$.

Combining the preceding formulas gives the following asymptotics at the
turning point.
For every compact $K\subset\mathbb R$, uniformly for $t\in K$,
\begin{align}
\ell_n g_n(\xi_n+\ell_n t)
&=\mathrm{Ai}(t)+O(\nu^{-2/3})\ .
\label{eq:airy-detail-density}
\end{align}
The cumulative profile obeys
\begin{align}
C_n(\xi_n+\ell_n t)
&=\mathcal I_{\rm Ai}(t)+O(\nu^{-2/3})\ .
\label{eq:airy-detail-cumulative}
\end{align}
For the first relation, take $u=2(\xi_n+\ell_n t)$. Then
$\upsilon=1+2^{2/3}\nu^{-2/3}t$. Combining
\cref{eq:airy-detail-fw,eq:airy-detail-langer} gives
$\ell_ng_n(\xi_n+\ell_nt)=2^{-1/3}\chi(\upsilon)
\mathrm{Ai}(Z_\nu(t))+\ell_n\mathcal R_n(\upsilon)$.
On compact $K$, $Z_\nu(t)$ remains bounded,
$\ell_n\mathcal R_n(\upsilon)=O(\nu^{-4/3})$,
$\chi(\upsilon)=2^{1/3}+O(\nu^{-2/3})$, and
$\mathrm{Ai}(Z_\nu(t))=\mathrm{Ai}(t)+O(\nu^{-2/3})$. This proves
\cref{eq:airy-detail-density}.

For the cumulative law, use normalization:
\begin{align}
1-C_n(X)=\int_X^\infty g_n(\xi)\,\dif \xi
=\frac12\int_{2X}^{\infty}(-1)^ne^{-u/2}L_n(u)\,\dif u\ .
\label{eq:airy-detail-norm}
\end{align}
For $X=\xi_n+\ell_n t$, the integration range lies in
$\upsilon\ge1-O(\nu^{-2/3})$. Fix a small $\delta>0$ and split the
tail integral at $\upsilon=1+\delta$. The second part is
$O(e^{-c_\delta\nu})$ by \cref{eq:airy-exterior-tail}. On the compact first
part, change variables to
$w=\nu^{2/3}\zeta_{\mathrm L}(\upsilon)$ and define the upper endpoint
$W_{\nu,\delta}:=\nu^{2/3}\zeta_{\mathrm L}(1+\delta)$.
The transformed Jacobian is
\begin{align}
J_\nu(w)&:=
\frac{\chi(\upsilon)}
{2\zeta_{\mathrm L}'(\upsilon)}
\bigg|_{\upsilon=\upsilon(w)}\ .
\label{eq:airy-detail-J}
\end{align}
Then \cref{eq:airy-detail-fw} gives
\begin{align}
\int_X^\infty g_n(\xi)\,\dif \xi
=\int_{Z_\nu(t)}^{W_{\nu,\delta}}
\mathrm{Ai}(w)J_\nu(w)\,\dif w
+\mathcal E_\nu+O(e^{-c_\delta\nu})\ ,
\label{eq:airy-tail-split}
\end{align}
where $\mathcal E_\nu=O(\nu^{-4/3})$, uniformly for $t\in K$.
Indeed, $\beta_1$ and $1/\zeta_{\mathrm L}'$ are bounded on the compact
$\upsilon$-interval, while
$\int_{\min K}^{\infty}\mathfrak A_1(w)\,\dif w<\infty$.

The removable singularity at $\upsilon=1$ gives
$J_\nu(w)=1+O(|\upsilon(w)-1|)$. On the interval above,
$|\upsilon(w)-1|\le C\nu^{-2/3}|w|$, and hence
\begin{align}
\int_{Z_\nu(t)}^{W_{\nu,\delta}}
\mathrm{Ai}(w)J_\nu(w)\,\dif w
&=\int_t^\infty\mathrm{Ai}(w)\,\dif w+O(\nu^{-2/3})\ .
\label{eq:airy-leading-tail}
\end{align}
Here, the displacement of the lower endpoint is
$Z_\nu(t)-t=O(\nu^{-2/3})$, while the Jacobian error is bounded by
$C\nu^{-2/3}\int_{\min K}^{\infty}|w\mathrm{Ai}(w)|\,\dif w$.
The omitted Airy tail beyond $W_{\nu,\delta}$ is $O(e^{-c_\delta\nu})$.
Since $\int_t^\infty\mathrm{Ai}(w)\,\dif w
=1-\mathcal I_{\rm Ai}(t)$, \cref{eq:airy-detail-norm,eq:airy-tail-split,eq:airy-leading-tail}
prove \cref{eq:airy-detail-cumulative}.

The remaining special-function input is standard. The largest zero of
$\mathrm{Ai}$ is $-2.338\ldots$, so $\mathrm{Ai}>0$ on
$[-2,\infty)$ \cite[\S9.9]{dlmf24}. Moreover,
$\mathcal I_{\rm Ai}(0)=2/3$ \cite[Eq.~9.10.11]{dlmf24}, and direct
evaluation of the convergent Airy
series gives $\mathcal I_{\rm Ai}(-2)=-0.2351\ldots$. Hence
$\mathcal I_{\rm Ai}$ has a unique zero $t_0\in(-2,0)$, numerically
$t_0=-1.38418\ldots$, and $\mathrm{Ai}(t_0)>0$. Because
$\mathcal I_{\rm Ai}$ is positive for $t\ge0$, this is its largest zero.

The positivity needed beyond the turning point follows from the location of
the Laguerre zeros:
\begin{align*}
\max\{u:L_n(u)=0\}\le4n-2\ .
\end{align*}
Consequently,
\begin{align}
g_n(\xi)>0\quad(\xi\ge\xi_n)\ .
\label{eq:airy-outer-positive}
\end{align}
The three-term recurrence
$uL_k(u)=-(k+1)L_{k+1}(u)+(2k+1)L_k(u)-kL_{k-1}(u)$
identifies the zeros of $L_n$ with the eigenvalues of the $n\times n$
symmetric tridiagonal Jacobi matrix whose diagonal entries are $2k+1$
$(0\le k\le n-1)$ and whose off-diagonal coupling between $k$ and $k+1$
has magnitude $k+1$. The Gershgorin row sums are bounded by
$(2k+1)+k+(k+1)=4k+2\le4n-2$.
Thus, every zero of $L_n$ is at most $4n-2$. For $L_n(2\xi)$ this gives
$\xi\le2n-1$, and since the leading sign of $g_n$ is positive after the
last zero, $g_n>0$ for $\xi\ge \xi_n=2n+1$.

\smallskip
\noindent\textit{Completion of the Airy proof.} Fix $t_{\rm c}=2$ and let
$t_0\in(-t_{\rm c},0)$ be the zero established above. Choose $\varepsilon>0$ with
$[t_0-\varepsilon,t_0+\varepsilon]\subset(-t_{\rm c},0)$, and set
$m_\varepsilon:=\min_{[t_0-\varepsilon,0]}\mathrm{Ai}>0$.
Let $E_\nu:=\sup_{t\in[-t_{\rm c},0]}
|C_n(\xi_n+\ell_n t)-\mathcal I_{\rm Ai}(t)|$ and
$\delta_\nu:=2(E_\nu+\nu^{-2/3})/m_\varepsilon$.
By \cref{eq:airy-detail-cumulative}, $E_\nu=O(\nu^{-2/3})$. Also,
\cref{eq:airy-detail-density} gives a uniform error $c_1\nu^{-2/3}$ on
$[-t_{\rm c},0]$. For all sufficiently large $n$,
$E_\nu<\min\{|\mathcal I_{\rm Ai}(-t_{\rm c})|,1/3\}$,
$\delta_\nu<\varepsilon$, and $c_1\nu^{-2/3}<m_\varepsilon$. Then
\begin{align*}
C_n(\xi_n-t_{\rm c}\ell_n)
\!\le \mathcal I_{\rm Ai}(-t_{\rm c})+E_\nu<0\ .
\end{align*}
At the lower displaced point,
\begin{align*}
C_n(\xi_n+\ell_n(t_0-\delta_\nu))
&\le -m_\varepsilon\delta_\nu+E_\nu
=-(E_\nu+2\nu^{-2/3})<0\ .
\end{align*}
On the interval $t_0+\delta_\nu\le t\le0$,
\begin{align*}
C_n(\xi_n+\ell_n t)
&\ge m_\varepsilon\delta_\nu-E_\nu
=E_\nu+2\nu^{-2/3}>0
\ .
\end{align*}
Beyond the turning point,
\begin{align*}
C_n(X)&\ge C_n(\xi_n)\ge
\mathcal I_{\rm Ai}(0)-E_\nu>1/3
\quad(X\ge\xi_n)\ .
\end{align*}
The first inequality shows that $C_n$ has negative values. The mean-value
theorem and monotonicity of $\mathcal I_{\rm Ai}$ near $t_0$ give the next
two inequalities. The last uses \cref{eq:airy-outer-positive}.
Therefore the last zero $X_0$ of the cumulative profile lies in the interval
\begin{align*}
X_0=\xi_n+\ell_n t_\nu\ ,
\qquad |t_\nu-t_0|\le\delta_\nu=O(\nu^{-2/3})\ .
\end{align*}
The density on the remaining inner interval satisfies
\begin{align*}
\ell_n g_n(\xi)\ge m_\varepsilon-c_1\nu^{-2/3}>0
\qquad(X_0\le\xi\le\xi_n)\ .
\end{align*}
In the interval $[X_0,\xi_n]$, the rescaled variable belongs to
$[t_0-\varepsilon,0]$, which gives the displayed lower bound.
For $\xi\ge \xi_n$, positivity again follows from
\cref{eq:airy-outer-positive}. Hence $g_n\ge0$ on
$[X_0,\infty)$, and the reduction to the cumulative profile gives
\begin{align*}
\Delta_n^*&=X_0
=\xi_n+\ell_n t_0+\ell_n O(\nu^{-2/3})\ ,\\[-0.35em]
\Delta_n^*&=(2n+1)+2^{-1/3}t_0(4n+2)^{1/3}+O(\nu^{-1/3})\ ,
\end{align*}
where the second line uses $\xi_n=2n+1$,
$\ell_n=2^{-1/3}\nu^{1/3}$, and $\nu=4n+2$. The remainder is equivalently
$O(n^{-1/3})$.

\smallskip
\noindent\textit{Finite energy does not imply a finite negativity scale.}
There exists a Fock-diagonal state $\rho$ such that
$\Tr(\rho\hat n)<\infty$ but $\Delta_\rho^*=\infty$.

Choose a compact interval $[t_-,t_+]$ strictly between the two largest
zeros of $\mathrm{Ai}$, where $\mathrm{Ai}<0$. The uniform density
asymptotics in \cref{eq:airy-detail-density} imply the existence of constants
$c_{\rm w},c_{\rm A}>0$ and $n_0$ such that, for every $n\ge n_0$, the
interval
\begin{align}
I_n:=[\xi_n+\ell_nt_-,\xi_n+\ell_nt_+]
\label{eq:infinite-scale-interval}
\end{align}
has length at least $c_{\rm w}n^{1/3}$ and
$g_n(\xi)\le-c_{\rm A}n^{-1/3}$ throughout $I_n$. For each fixed $m$,
$g_m(\xi)$ is $e^{-\xi}$ times a polynomial. Since
$\inf I_n=2n+O(n^{1/3})$, it follows that
\begin{align}
e^n n^{1/3}\sup_{\xi\in I_n}|g_m(\xi)|\longrightarrow0
\qquad(n\to\infty)\ .
\label{eq:infinite-scale-separation}
\end{align}

We can therefore choose strictly increasing positive integers $N_j\ge n_0$
recursively so that the following three conditions hold for every $j\ge1$.
First, the interval condition is
\begin{align}
|I_{N_j}|&\ge2j\ .
\label{eq:infinite-scale-width}
\end{align}
Second, the contribution of the previously chosen indices obeys
\begin{align}
\sum_{i<j}e^{-N_i}\sup_{\xi\in I_{N_j}}|g_{N_i}(\xi)|
&\le \frac{c_{\rm A}e^{-N_j}}{4N_j^{1/3}}\ .
\label{eq:infinite-scale-lower}
\end{align}
Finally, the next coefficient satisfies
\begin{align}
\frac{e^{-N_{j+1}}}{1-e^{-1}}
&\le \frac{c_{\rm A}e^{-N_j}}{4N_j^{1/3}}\ .
\label{eq:infinite-scale-upper}
\end{align}
Indeed, after $N_1,\ldots,N_j$ are fixed, taking $N_{j+1}$
sufficiently large enforces
\cref{eq:infinite-scale-width,eq:infinite-scale-lower} at index $j+1$ by
\cref{eq:infinite-scale-separation}, and simultaneously enforces
\cref{eq:infinite-scale-upper} at index $j$.

Set $Z:=\sum_{j=1}^\infty e^{-N_j}$.
The state is
\begin{align}
\rho:=\frac1Z\sum_{j=1}^\infty e^{-N_j}
\ket{N_j}\!\bra{N_j}\ .
\label{eq:infinite-scale-state}
\end{align}
Here $0<Z\le\sum_{n=1}^\infty e^{-n}<\infty$. Moreover,
$\sum_j e^{-N_j}|g_{N_j}|/Z\le1$, so the radial-profile series converges
uniformly and may be integrated term by term over every shell. This is a
finite-energy state because the $N_j$ are distinct positive integers and hence
\begin{align}
\Tr(\rho\hat n)
=\frac1Z\sum_{j=1}^\infty N_je^{-N_j}
\le\frac1Z\sum_{n=1}^\infty ne^{-n}<\infty\ .
\label{eq:infinite-scale-energy}
\end{align}

Let $a_j:=\inf I_{N_j}$ and $k_j:=\lfloor a_j/j\rfloor+2$. By
\cref{eq:infinite-scale-width}, the complete shell
$\mathcal J_j=[(k_j-1)j,k_jj)$ lies inside $I_{N_j}$ and belongs to the equal-area
partition with resolution $\Delta_j=j$. Define
$Q_j:=Z^{-1}c_{\rm A}j e^{-N_j}N_j^{-1/3}$. The $N_j$ component contributes
at most $-Q_j$ to this shell. The components below it obey
\begin{align}
\left|\frac1Z\sum_{i<j}e^{-N_i}\int_{\mathcal J_j}g_{N_i}(\xi)\,\dif\xi\right|
&\le\frac jZ\sum_{i<j}e^{-N_i}
\sup_{\xi\in I_{N_j}}|g_{N_i}(\xi)|
\le\frac{Q_j}{4}\ .
\label{eq:infinite-scale-lower-bound}
\end{align}
Because the $N_i$ are strictly increasing integers,
$\sum_{i>j}e^{-N_i}\le e^{-N_{j+1}}/(1-e^{-1})$. Using
$|g_n|\le1$ \cite[Eq.~18.14.8]{dlmf24} and
\cref{eq:infinite-scale-upper}, the components above $N_j$ similarly obey
\begin{align}
\left|\frac1Z\sum_{i>j}e^{-N_i}\int_{\mathcal J_j}g_{N_i}(\xi)\,\dif\xi\right|
\le\frac jZ\sum_{i>j}e^{-N_i}
\le\frac{Q_j}{4}\ .
\label{eq:infinite-scale-upper-bound}
\end{align}
Consequently, $P_{\rho,k_j}(\Delta_j)\le-Q_j/2<0$. Since $\Delta_j=j$ is
unbounded, $\Delta_\rho^*=\infty$.

\section{Shell-entropy bounds and fine-graining}

For any complete radial partition with bounded squared-radius widths,
$\boldsymbol P_0(\Pi^\downarrow)$ supplies a
finite-resolution lower bound on every Schur-concave shell entropy. When the
shell widths admit a non-increasing enumeration, it is obtained by placing the
same shell areas in that order from the origin and evaluating the shell
probabilities of the vacuum state; otherwise it denotes the limiting benchmark
defined by the order statistics $D_m$. The
equal-area partition is distinguished by a closed-form geometric distribution, whose
fine-graining limit yields the radial Wigner-entropy relation derived below.

Let $\Pi$ be complete, have bounded squared-radius widths, and suppose that
$\boldsymbol P_\rho(\Pi)$ is componentwise non-negative. By
\cref{eq:radial-shell-app} and Schur concavity,
\begin{align}
H_\alpha\!\left(\boldsymbol P_\rho(\Pi)\right)
\ge
H_\alpha\!\left(\boldsymbol P_0(\Pi^\downarrow)\right)\ ,
\label{eq:general-shell-entropy-app}
\end{align}
for every R\'enyi order $\alpha>0$ for which the entropy is defined, with the
Shannon case obtained at $\alpha=1$. The right-hand side depends only on the
shell areas. For each $m$, the first $m$ components of the vacuum benchmark
sum to the integral of $W_0$ over a centered disk whose area is the supremum
of the total area of any $m$ shells. This separates
state-dependent concentration from the purely geometric advantage of a wide
shell.

For the equal-area partition $\Pi_\Delta$, ordering the widths changes
nothing and the shell probabilities of the vacuum state form the geometric distribution
$P_{0,k}(\Delta)=(1-e^{-\Delta})e^{-(k-1)\Delta}$, $k\ge1$.
The geometric sum gives the Shannon entropy
\begin{align}
H(\boldsymbol P_0(\Delta))
&=-\ln(1-e^{-\Delta})
+\frac{\Delta e^{-\Delta}}{1-e^{-\Delta}}\ .
\label{eq:vacuum-entropies-app}
\end{align}
For $\alpha>0$, $\alpha\ne1$, the R\'enyi entropy is
\begin{align*}
H_\alpha(\boldsymbol P_0(\Delta))
&=\frac{\alpha\ln(1-e^{-\Delta})
-\ln(1-e^{-\alpha\Delta})}{1-\alpha}
\qquad(\alpha>0,\ \alpha\ne1)\ .
\end{align*}

\smallskip
\noindent\textit{Vacuum asymptotics.}
Expanding \cref{eq:vacuum-entropies-app} at small $\Delta$ gives
$H(\boldsymbol P_0(\Delta))=-\ln\Delta+1+\Delta^2/24+O(\Delta^3)$.

\smallskip
\noindent\textit{From shell entropy to radial differential entropy.}
Let $\bar\rho$ be phase invariant and Wigner-positive. With
$\xi=|\boldsymbol r|^2$, radial integration gives
$\int_{\mathbb R^2}f(|\boldsymbol r|^2)\,\dif\boldsymbol r
=\pi\int_0^\infty f(\xi)\,\dif\xi$. Define
$g_\rho(\xi):=\pi W_{\bar\rho}(\boldsymbol r)|_{|\boldsymbol r|^2=\xi}$,
with $\int_0^\infty g_\rho(\xi)\,\dif\xi=1$.
The universal Wigner bound gives $0\le g_\rho\le1$.

Let $g$ be a probability density on $[0,\infty)$ with
$0\le g\le1$, finite differential entropy, and finite Shannon entropy for its
unit-bin distribution. If
$P_{g,k}(\Delta):=\int_{(k-1)\Delta}^{k\Delta}g(\xi)\,\dif \xi$, then
\begin{align}
\lim_{n\to\infty}\bigl[H(\boldsymbol P_g(1/n))-\ln n\bigr]
=h(g):=-\int_0^\infty g(\xi)\ln g(\xi)\,\dif \xi\ .
\label{eq:GY-limit}
\end{align}

For completeness, let $X$ have density $g$, set $J=\lfloor X\rfloor$,
and let $g_j$ be the conditional density on $[j,j+1)$, with probability
$p_j$. The continuous chain rule gives
$h(g)=H(J)+\sum_jp_jh(g_j)$.
For $Q_n=\lfloor nX\rfloor$, the discrete chain rule gives
\begin{align}
H(Q_n)-\ln n
&=H(J)+\sum_jp_j[H(Q_n|J=j)-\ln n]\ .
\label{eq:entropy-chain-app}
\end{align}
Let $g_{j,n}$ be the piecewise constant density obtained by averaging
$g_j$ on the $n$ equal subintervals of $[j,j+1)$. Then
$H(Q_n|J=j)-\ln n-h(g_j)=D(g_j\Vert g_{j,n})\to0$.
Indeed, $g_{j,n}\to g_j$ almost everywhere and in $L^1$ by the
Lebesgue differentiation theorem. For
$\Phi(x)=x\ln x-x+1$, Jensen's inequality gives
$\Phi(g_{j,n})\le
\mathbb E[\Phi(g_j)\mid\mathcal P_n]$, where $\mathcal P_n$ is the
partition into $n$ equal subintervals. Since
$\int\Phi(g_j)=D(g_j\Vert1)<\infty$, these conditional averages are
uniformly integrable. Hence
$\int\Phi(g_{j,n})\to\int\Phi(g_j)$, which is equivalent to the stated
relative-entropy convergence. Jensen's
inequality also gives $h(g_j)\le H(Q_n|J=j)-\ln n\le0$.
Since $\sum_jp_j|h(g_j)|=H(J)-h(g)<\infty$, dominated convergence proves
\cref{eq:GY-limit}.

For a Wigner-positive state of finite mean photon number these hypotheses hold:
$\int_0^\infty \xi g_\rho(\xi)\,\dif \xi=2\langle\hat n\rangle_\rho+1<\infty$.
Finite mean gives $H(\boldsymbol P_\rho(1))<\infty$ and bounds the
differential entropy from above, while $0\le g_\rho\le1$ gives
$h(g_\rho)\ge0$.

The radial and two-dimensional Wigner entropies are related by
\begin{align}
h(W_{\bar\rho})
&=-\int_0^\infty g_\rho(\xi)
\ln\!\left[\frac{g_\rho(\xi)}{\pi}\right]\dif\xi
=h(g_\rho)+\ln\pi\ .
\end{align}
For the vacuum, $g_0(\xi)=e^{-\xi}$ and $h(g_0)=1$, so
$h(W_0)=1+\ln\pi$.

\smallskip
\noindent\textit{Wigner entropy bound.}
For every $n\ge1$, the specialization of
\cref{eq:radial-shell-app} gives
$H(\boldsymbol P_{\bar\rho}(1/n))\ge H(\boldsymbol P_0(1/n))$.
Subtracting $\ln n$ and taking $n\to\infty$, \cref{eq:GY-limit} and the
vacuum asymptotics yield
$h(g_\rho)\ge h(g_0)=1$.
Equivalently,
\begin{align}\label{eq:continuous_entropy}
h(W_{\bar\rho})\ge1+\ln\pi\ .
\end{align}
For phase-invariant finite-energy Wigner-positive states, this is the Wigner
entropy bound \cite{vanherstraeten23}.

\section{Finite-bin EPR entanglement detection}\label{app:epr-benchmarks}

This section proves \cref{cor:epr-criterion}, derives its entropic
and covariance corollaries, and gives the calculations underlying
\cref{fig:epr-nongaussian-benchmarks}. The Gaussian calibration is followed
by two non-Gaussian benchmarks: the first-order biphoton family and
Schmidt-correlated states with finite Fock support.

\smallskip
\noindent\textit{Proof of the entanglement criterion.}
Introduce the balanced quadratures
$\hat q_\pm=(\hat q_1\pm\hat q_2)/\sqrt2$ and
$\hat p_\pm=(\hat p_1\pm\hat p_2)/\sqrt2$.
They obey $[\hat q_\pm,\hat p_\pm]=i$ and
$[\hat q_\pm,\hat p_\mp]=0$. Hence the EPR pair
$\hat u=\hat q_-$ and $\hat v=\hat p_+$ is jointly measurable. The Wigner
marginal theorem \cite{leonhardt97} gives its normalized outcome density as
\begin{align}
F_\rho(u,v)
:=\int_{\mathbb R^2}W_{\rho_{12}}^{(\pm)}
  (q_+,v,u,p_-)\,\dif q_+\,\dif p_-\ ,
\end{align}
where $W_{\rho_{12}}^{(\pm)}$ is the two-mode Wigner function in the ordered
coordinates $(q_+,p_+,q_-,p_-)$. As a joint measurement density,
$F_\rho(u,v)\ge0$ and
$\int_{\mathbb R^2}F_\rho(u,v)\,\dif u\dif v=1$.

Let $\Gamma$ denote partial transposition with respect to the second mode,
taken in the position basis.
In phase space it acts as momentum reflection \cite{simon00},
$W_{\rho_{12}^\Gamma}(q_1,p_1,q_2,p_2)
=W_{\rho_{12}}(q_1,p_1,q_2,-p_2)$.
It therefore interchanges $p_+$ and $p_-$ while leaving $q_+$ and $q_-$
unchanged:
\begin{align}
W_{\rho_{12}^\Gamma}^{(\pm)}(q_+,p_+,q_-,p_-)
=W_{\rho_{12}}^{(\pm)}(q_+,p_-,q_-,p_+)\ .
\label{eq:epr-pt-wigner-app}
\end{align}
Choose the balanced beam-splitter unitary $U_{\rm BS}$ to map the local
quadratures to the $+$ and $-$ modes, and set
$\tau_\rho:=\Tr_+[U_{\rm BS}\rho_{12}^\Gamma U_{\rm BS}^\dagger]$.
The Wigner function of a reduced state is obtained by integrating over the
discarded mode. Using \cref{eq:epr-pt-wigner-app} gives
\begin{align}
W_{\tau_\rho}(u,v)
&=\int_{\mathbb R^2}W_{\rho_{12}^\Gamma}^{(\pm)}
  (q_+,p_+,u,v)\,\dif q_+\,\dif p_+ \notag\\[-0.35em]
&=\int_{\mathbb R^2}W_{\rho_{12}}^{(\pm)}
  (q_+,v,u,p_-)\,\dif q_+\,\dif p_-
=F_\rho(u,v)\ .
\label{eq:epr-wigner-marginal}
\end{align}
No factor of $2\pi$ appears in these marginals because the Wigner convention
in \cref{eq:wigner} is
normalized with respect to the Lebesgue phase-space measure.

Consequently,
$P_{\tau_\rho,k}(\Delta)=\int_{A_k(\Delta)}F_\rho(u,v)\,\dif u\dif v$
is the probability that the measured pair lies in the $k$th shell. From
$N$ independent outcomes it is estimated directly as
$\widehat P_{\tau_\rho,k}(\Delta)=N^{-1}\sum_{\ell=1}^N
\mathbbm1_{A_k(\Delta)}(u_\ell,v_\ell)$, without reconstructing a
continuous density.

If $\rho_{12}$ has a positive partial transpose (PPT), then
$\rho_{12}^\Gamma\ge0$. Unitary conjugation and the partial trace preserve
positivity and trace, so $\tau_\rho$ is a single-mode density operator. Its
Wigner function is the non-negative density $F_\rho$. The equal-area
relation \cref{eq:main} therefore gives
\begin{align}
\boldsymbol P_{\tau_\rho}(\Delta)
\prec\boldsymbol P_0(\Delta)\ ,\qquad \Delta>0\ .
\end{align}
This proves \cref{cor:epr-criterion}; equivalently,
$\sum_{i=1}^{k}P_{\tau_\rho,i}^{\downarrow}(\Delta)\le1-e^{-k\Delta}$ for
every $k\ge1$. Since every separable state is PPT, any violation certifies
negative-partial-transpose (NPT) entanglement. The converse is not implied:
satisfying all the inequalities does not establish separability, and
PPT-entangled states cannot violate this criterion. This is a bound on the
exact bin probabilities, so finite bin width cannot produce a violation for a
separable state.

\smallskip
\noindent\textit{Entropic and covariance consequences.}
Schur concavity gives
$H_\alpha(\boldsymbol P_{\tau_\rho}(\Delta))
\ge H_\alpha(\boldsymbol P_0(\Delta))$, which is
\cref{eq:general-shell-entropy-app} for equal-area shells.

The optimized MGVT variance-product consequence follows without assuming
Gaussianity. After removing the first moments of the measured EPR pair, write
its covariance matrix as $\Sigma$. In the single-mode representation
\cref{eq:epr-wigner-marginal}, $(u,v)$ are canonical phase-space
coordinates. For
$\mathsf S\in\operatorname{Sp}(2,\mathbb R)
=\operatorname{SL}(2,\mathbb R)$, define
\begin{align}
\begin{pmatrix}u_{\mathsf S}\\v_{\mathsf S}\end{pmatrix}
=\mathsf S\begin{pmatrix}u\\v\end{pmatrix}\ .
\label{eq:epr-symplectic-outcomes}
\end{align}
The transformed density is the Wigner function of the corresponding
Gaussian-unitary image of $\tau_\rho$, so
\cref{eq:epr-criterion} applies separately in every transformed frame.
Operationally, \cref{eq:epr-symplectic-outcomes} is a linear
post-processing of each measured outcome pair. Since
$\det\mathsf S=1$, it preserves $\dif u\,\dif v$ and maps circular shells in
the transformed coordinates to equal-area ellipses in the measured plane.
The transformed outcomes can therefore be rebinned without a new
measurement or density reconstruction. This constructs a family of valid
criteria; it does not assume symplectic invariance of a fixed shell
partition.

Let $g_{\mathsf S}$ denote the probability density of
$\xi_{\mathsf S}:=u_{\mathsf S}^2+v_{\mathsf S}^2$, whose mean is
$\langle\xi_{\mathsf S}\rangle
=\Tr(\mathsf S\Sigma\mathsf S^{\mathsf T})$.
If the state is PPT, the fine-graining limit of its Shannon-entropy bound for
the energy-shell probabilities gives
$h(g_{\mathsf S})\ge1$. Setting $m=\langle\xi_{\mathsf S}\rangle$,
positivity of
$D(g_{\mathsf S}\Vert m^{-1}e^{-\xi/m})=1+\ln m-h(g_{\mathsf S})$ gives
$h(g_{\mathsf S})\le1+\ln\Tr(\mathsf S\Sigma\mathsf S^{\mathsf T})$, with equality only for an
exponential radial density. Hence every PPT state satisfies
$\Tr(\mathsf S\Sigma\mathsf S^{\mathsf T})\ge1$ for every
$\mathsf S\in\operatorname{Sp}(2,\mathbb R)$. In particular,
$\mathsf S=I_2$ gives
\begin{align}
\operatorname{Var}(u)+\operatorname{Var}(v)=\Tr\Sigma\ge1\ ,
\label{eq:epr-balanced-dgcz}
\end{align}
which is the balanced Duan--Giedke--Cirac--Zoller variance bound in the
present normalization \cite{duan00}.

For a rotation $\mathsf R$ with
$\mathsf R\Sigma\mathsf R^{\mathsf T}=\operatorname{diag}(V_u,V_v)$, set
$\mathsf S_\mu=\operatorname{diag}(\mu,\mu^{-1})\mathsf R$. Then
\begin{align}
\Tr(\mathsf S_\mu\Sigma\mathsf S_\mu^{\mathsf T})
=\mu^2V_u+\mu^{-2}V_v\ge1\ .
\label{eq:epr-symplectic-trace}
\end{align}
The minimum, at $\mu^4=V_v/V_u$, is $2\sqrt{\det\Sigma}$. Hence
\begin{align}
\det\Sigma\ge\frac14\ .
\label{eq:epr-determinant-criterion}
\end{align}
In the principal axes of $\Sigma$, this is the MGVT variance-product
condition \cite{mancini02,giovannetti03}. The argument also parallels the
maximum-entropy reduction of the GHN criterion in
Ref.~\cite{garttner23}.

\smallskip
\noindent\textit{Gaussian calibration and finite-bin comparisons.}
Let $F(u,v)$ be a centered Gaussian EPR density with covariance matrix
$\Sigma$. A one-mode Williamson transformation, implemented by a symplectic
change of the canonical pair obtained after partial transposition, brings the
covariance to $\mathsf S\Sigma\mathsf S^{\mathsf T}=\nu I_2$ with
$\nu=\sqrt{\det\Sigma}$. In these optimized coordinates,
$F_\nu(u,v)=(2\pi\nu)^{-1}\exp[-(u^2+v^2)/(2\nu)]$, and the shell weights form
a geometric distribution:
\begin{align}
P_{\nu,k}(\Delta)
=\left(1-e^{-\Delta/(2\nu)}\right)
e^{-(k-1)\Delta/(2\nu)}\ .
\label{eq:epr-gaussian-geometric}
\end{align}
The vacuum corresponds to $\nu=1/2$. Since the first $K$ probabilities of
a geometric distribution with ratio $q$ sum to $1-q^K$,
\cref{eq:epr-gaussian-geometric} is majorized by
$\boldsymbol P_0(\Delta)$ exactly
when $\nu\ge1/2$. For $\nu<1/2$, the first component already violates the
order. Thus, for this Gaussian EPR pair and every $\Delta>0$, the optimized
form of criterion \eqref{eq:epr-criterion} is equivalent to the MGVT
variance-product condition $\sqrt{\det\Sigma}\ge1/2$. This equivalence does not by itself
recover the Simon covariance criterion, equivalently the PPT criterion for
two-mode Gaussian states, which uses the complete
$4\times4$ two-mode covariance matrix \cite{simon00}. Relating EPR variance
criteria to that full test requires optimization over a larger family of EPR
observables \cite{duan00,giovannetti03}.

For completeness, the finite-bin comparisons in
\cref{fig:epr-nongaussian-benchmarks}(a) use equal phase-space areas and
parametrize the two-mode squeezed vacuum by the real squeezing parameter
$\zeta\ge0$. A square
cell of area $\gamma$ in the canonical $(u,v)$ plane corresponds to
$\gamma=\pi\Delta$. In this convention, the TRGTW criterion
\cite{tasca13} reads
\begin{align}
H(U_\gamma)+H(V_\gamma)
\ge-\ln\min\!\left\{\frac{\gamma}{\pi e},
\lambda_0\!\left(\frac{\gamma}{4}\right)\right\}\ ,
\label{eq:tasca-shannon-app}
\end{align}
where $U_\gamma$ and $V_\gamma$ are the measured bin-probability vectors and
$\lambda_0(c)$ is the largest eigenvalue of the time--band-limiting operator
with kernel $\sin[c(x-y)]/[\pi(x-y)]$ on $[-1,1]$. The plotted threshold is
obtained by integrating the centered Gaussian marginals of variance
$e^{-2\zeta}/2$ over bins of width $\sqrt\gamma$.

For the regular square-grid discretization of Ref.~\cite{garttner23}, the
nonlocal coordinates obey
$r_\pm=\sqrt2u$ and $s_\mp=\sqrt2v$. Hence, a square of side $\delta$ has
area $\delta^2/2$ in the canonical $(u,v)$ plane. Equating this with the
shell area $\pi\Delta$ gives
$\delta=\sqrt{2\pi\Delta}$. If
$\boldsymbol p$ contains the one-dimensional Gaussian bin probabilities with
variance $1+e^{-2\zeta}$, Eqs.~(114)--(117) of Ref. \cite{garttner23} give
\begin{align}
\mathcal W_{\beta}^{\Delta}
=2H_\beta(\boldsymbol p)+\ln\Delta
-\frac{\ln\beta}{\beta-1}-\ln2\ .
\label{eq:garttner-grid-app}
\end{align}
The GHN curve in \cref{fig:epr-nongaussian-benchmarks}(a) is obtained by
numerically evaluating \cref{eq:garttner-grid-app} after the above change of
variables and reproduces the published benchmark. In particular, our numerical
curve passes through
$(\zeta,\delta_c)=(0.100335,2.650)$ and $(1.472219,6.190)$ and terminates at
$\delta_c\simeq6.47$, corresponding to $\Delta\simeq6.66$.

\smallskip
\noindent\textit{First-order Hermite--Gaussian biphoton benchmark.}
The biphoton family in \cref{eq:hg-state} was experimentally realized and analyzed
by Gomes \textit{et al.} \cite{gomes09pnas,gomes09prl}. Its linear prefactor is
the first Hermite polynomial in the collective coordinate $r_1+r_2$, which is
the origin of the term first-order Hermite--Gaussian mode. It was subsequently
used to benchmark STW \cite{saboia11} and GHN
\cite{garttner23prl,garttner23}:
\begin{align}
\psi_{\sigma_+,\sigma_-}(r_1,r_2)
=\frac{r_1+r_2}{\sqrt{\pi\sigma_-\sigma_+^3}}
\exp\!\left[-\frac14\left(
\frac{(r_1+r_2)^2}{\sigma_+^2}
+\frac{(r_1-r_2)^2}{\sigma_-^2}
\right)\right]\ ,
\label{eq:hg-state}
\end{align}
where $\sigma_+,\sigma_->0$. The state is entangled throughout this parameter
domain, including at $\sigma_+=\sigma_-$. It therefore permits a direct
comparison of criterion \eqref{eq:epr-criterion} with the earlier criteria.

In the convention of Refs.~\cite{garttner23prl,garttner23},
$r_\pm=r_1\pm r_2$ and $s_\pm=s_1\pm s_2$, where
$s_j$ is conjugate to $r_j$. The canonically normalized commuting pair is
$u=r_\pm/\sqrt2$ and $v=s_\mp/\sqrt2$.

Suppose first that $\sigma_-\ge\sigma_+$ and set
$a=\sigma_-/\sigma_+$. Choose the EPR pair $(r_+,s_-)$, align its principal
axes, and apply the determinant-one gain
$\mathsf S_{\rm loc}=\operatorname{diag}(\mu_{\rm loc},\mu_{\rm loc}^{-1})$
with $\mu_{\rm loc}^2=\sigma_+\sigma_-$. In the phase-space convention of
Refs. \cite{garttner23prl,garttner23}, the resulting EPR density is
$W_+(r_+,s_-)=a^2r_+^2
\exp[-a(r_+^2+s_-^2)/2]/(2\pi)$.
The rescaling $u=r_+/\sqrt2$ and $v=s_-/\sqrt2$ brings the vacuum to the
normalization used in the main text. Accounting for the Jacobian gives
\begin{align}
F_a(u,v)=\frac{2a^2}{\pi}u^2e^{-a(u^2+v^2)}\ .
\label{eq:hg-epr-density}
\end{align}
If $\sigma_+>\sigma_-$, choosing $(r_-,s_+)$ instead gives the same expression
up to interchanging the axes, with $a=\sigma_+/\sigma_-$. Thus, in both cases,
$a=\max\{\sigma_-/\sigma_+,\sigma_+/\sigma_-\}\ge1$.

With $\xi=u^2+v^2$, polar integration gives
$\varrho_a(\xi)=a^2\xi e^{-a\xi}$ for $\xi\ge0$.
The shell probabilities are
\begin{align}
P_{a,k}(\Delta)
&=\bigl[1+a(k-1)\Delta\bigr]e^{-a(k-1)\Delta}
-\bigl(1+ak\Delta\bigr)e^{-ak\Delta}\ .
\label{eq:hg-shell-probabilities}
\end{align}

The R\'enyi consequence of criterion \eqref{eq:epr-criterion} admits an explicit
$\alpha\to0^+$ limit at every fixed resolution. Set
$q_a=e^{-a\Delta}$ and rewrite
\cref{eq:hg-shell-probabilities} as
$P_{a,k}(\Delta)=q_a^{k-1}(A_ak+B_a)$, where
$A_a=a\Delta(1-q_a)$ and $B_a=1-a\Delta-q_a$.
Then
\begin{align}
\lim_{\alpha\to0^+}\alpha a\Delta
\sum_{k=1}^\infty P_{a,k}^\alpha(\Delta)=1\ .
\label{eq:hg-shell-abelian}
\end{align}
With $t=\alpha a\Delta$, the geometric factor restricts the sum to
$k=O(t^{-1})$, where $(A_ak+B_a)^\alpha\to1$. Multiplication by $t$
gives the Riemann sum for $\int_0^\infty e^{-y}\dif y=1$, with an
exponentially decaying majorant.

For $\boldsymbol P_0(\Delta)$,
$\sum_{k\ge1}P_{0,k}^\alpha(\Delta)
=(1-e^{-\Delta})^\alpha/(1-e^{-\alpha\Delta})$, and
$\lim_{\alpha\to0^+}\alpha\Delta\sum_kP_{0,k}^\alpha(\Delta)=1$.
Taking the logarithm of the ratio of the two power sums gives
\begin{align}
\lim_{\alpha\to0^+}\left[
H_\alpha(\boldsymbol P_a(\Delta))
-H_\alpha(\boldsymbol P_0(\Delta))\right]
=-\ln a\ .
\label{eq:hg-shell-zero-witness}
\end{align}
Thus, every $a>1$ violates criterion \eqref{eq:epr-criterion} at every fixed
$\Delta>0$ for all sufficiently small positive orders. The limiting entropy
difference is independent of $\Delta$.

For comparison, consider the GHN criterion \cite{garttner23} for the same
optimized state. In that convention, the
nonlocal $Q$ density is
$Q_a(r,s)=a^2(1+a+r^2)(1+a)^{-3}
\exp[-a(r^2+s^2)/(2(1+a))]$,
normalized with respect to $\dif r\dif s/(2\pi)$. Let
$I_\beta(a)=\int Q_a^\beta\dif r\dif s/(2\pi)$. Under the rescaling
$(r,s)=(R,S)/\sqrt\beta$, dominated convergence gives
$I_\beta(a)=(1+a)[1+o(1)]/(a\beta)$ and
$\ln I_\beta(a)/(1-\beta)
=-\ln\beta+\ln[(1+a)/a]+o(1)$ as $\beta\to0^+$.
\begin{samepage}
The vacuum R\'enyi--Wehrl bound is
$\ln\beta/(\beta-1)+\ln2=-\ln\beta+\ln2+o(1)$, so the optimized
R\'enyi--Wehrl entropy difference is
\begin{align}
\lim_{\beta\to0^+}\left[
\frac{\ln I_\beta(a)}{1-\beta}
-\frac{\ln\beta}{\beta-1}-\ln2\right]
=-\ln\frac{2a}{1+a}\ .
\label{eq:hg-q-zero-witness}
\end{align}
It is negative for every $a>1$, exactly the ideal detection region of the
small-$\alpha$ R\'enyi consequence of criterion
\eqref{eq:epr-criterion}.
\end{samepage}

For comparison, the marginals of \cref{eq:hg-epr-density} are
$f_a^{(u)}(u)=2a^{3/2}u^2e^{-au^2}/\sqrt\pi$ and
$f_a^{(v)}(v)=\sqrt{a/\pi}\,e^{-av^2}$.
For a probability density $f$, write
$h_\alpha(f):=(1-\alpha)^{-1}\ln\int f(x)^\alpha\,\dif x$, with the usual
continuous extensions at $\alpha=1$ and $\infty$.
The optimized STW criterion \cite{saboia11} uses orders
$1/2$ for $f_a^{(u)}$ and $\infty$ for $f_a^{(v)}$. Direct integration gives
$h_{1/2}(f_a^{(u)})=2\ln\int\sqrt{f_a^{(u)}(u)}\,\dif u
=\ln8-\tfrac12\ln(\pi a)$ and
$h_\infty(f_a^{(v)})=-\ln\max_v f_a^{(v)}(v)=\tfrac12\ln(\pi/a)$.
Their sum is $\ln(8/a)$, while the PPT bound in the present
normalization is $\ln(2\pi)$. Hence
\begin{align}
h_{1/2}(f_a^{(u)})+h_\infty(f_a^{(v)})-\ln(2\pi)
=-\ln\frac{\pi a}{4}\ ,
\label{eq:hg-stw-witness}
\end{align}
which becomes negative only for $a>4/\pi$.

The optimized MGVT threshold follows from the same density.
\Cref{eq:hg-epr-density} has vanishing cross covariance and
$\langle u^2\rangle=3/(2a)$,
$\langle v^2\rangle=1/(2a)$, and $\det\Sigma=3/(4a^2)$.
The condition $\det\Sigma<1/4$ thus detects only
$a>\sqrt3$.

\begin{figure}[t]
\centering
\includegraphics[width=0.42\textwidth]{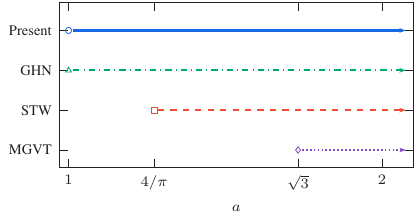}
\caption{Detection intervals for the first-order Hermite--Gaussian biphoton
family, with
$a=\max\{\sigma_-/\sigma_+,\sigma_+/\sigma_-\}$. Arrows point into the
detected regions, and open markers denote excluded thresholds. The present
criterion and GHN detect $a>1$, whereas STW and MGVT require $a>4/\pi$ and
$a>\sqrt3$, respectively. The entangled state at $a=1$ is a blind point of
these linear-EPR criteria.}
\label{fig:supp-biphoton-benchmark}
\end{figure}

At~$a=1$, \cref{eq:hg-state} is, up to equal local squeezing, the first NOON
state. Although entangled, it remains a blind point for all the criteria
considered here because they use linear EPR variables.

\smallskip
\noindent\textit{Finite-support Schmidt-correlated states.}
Consider $\ket{\psi_{\boldsymbol c}}=\sum_{k\in\mathcal K}c_k\ket{k,k}$,
where $\mathcal K$ is finite and $\sum_{k\in\mathcal K}|c_k|^2=1$. Apply the
local phase rotation to obtain
$\ket{\psi_\Theta}=(e^{i\Theta\hat n}\otimes\mathbb I)
\ket{\psi_{\boldsymbol c}}$. This local unitary leaves entanglement unchanged
and is equivalent to changing the phase of the first homodyne setting, so
varying $\Theta$ provides a local optimization of the criterion. Let
$\varphi_k$ be the normalized oscillator wave function. For the commuting
EPR pair $u=q_-$ and $v=p_+$, the
contribution of $\ket{k,k}$ to the joint amplitude is
\begin{align}
\mathcal A_k(u,v)
=\frac{1}{\sqrt{2\pi}}\int_{\mathbb R}e^{-ivx}
\varphi_k\!\left(\frac{x+u}{\sqrt2}\right)
\varphi_k\!\left(\frac{x-u}{\sqrt2}\right)\dif x\ .
\label{eq:schmidt-amplitude}
\end{align}
The Mehler formula \cite[Eq.~18.18.28]{dlmf24}, followed by the Gaussian Fourier
integral and the Laguerre generating function \cite[Eq.~18.12.13]{dlmf24},
gives
\begin{align}
\sum_{k=0}^\infty z^k\mathcal A_k(u,v)
=\frac{e^{-\xi/2}}{\sqrt\pi(1-z)}
\exp\!\left(-\frac{z\xi}{1-z}\right)
=\frac{e^{-\xi/2}}{\sqrt\pi}\sum_{k=0}^\infty L_k(\xi)z^k\ ,
\label{eq:schmidt-amplitude-generating}
\end{align}
where $\xi=u^2+v^2$. Hence
$\mathcal A_k(u,v)=\pi^{-1/2}e^{-\xi/2}L_k(\xi)$.
Taking the modulus square of the superposed amplitude gives
\begin{align}
F_{\psi_\Theta}(u,v)
=\frac{e^{-\xi}}{\pi}
\left|\sum_{k\in\mathcal K}c_ke^{ik\Theta}L_k(\xi)\right|^2\ ,
\qquad \xi=u^2+v^2\ .
\label{eq:schmidt-epr-density}
\end{align}
The square form is manifestly non-negative, as required by its direct
interpretation as the joint probability density of the commuting EPR pair.
For a PPT input, it is also the Wigner function of the single-mode density
operator $\tau_{\psi_\Theta}$ defined above.

For brevity, write
$P_{\psi_\Theta,j}(\Delta):=P_{\tau_{\psi_\Theta},j}(\Delta)$.
Integrating \cref{eq:schmidt-epr-density} over the $j$th radial shell gives
\begin{align}
P_{\psi_\Theta,j}(\Delta)
&=\sum_{m,n\in\mathcal K}c_m c_n^*e^{i(m-n)\Theta}
\mathcal I_{mn}\!\left((j-1)\Delta,j\Delta\right)\ .
\label{eq:schmidt-shell}
\end{align}
The auxiliary integral is
\begin{align}
\mathcal I_{mn}(a,b)
&:=\int_a^b e^{-\xi}L_m(\xi)L_n(\xi)\,\dif \xi =\sum_{r=0}^{m}\sum_{\ell=0}^{n}
\frac{(-1)^{r+\ell}\binom mr\binom n\ell}{r!\,\ell!}
\left[\Gamma(r+\ell+1,a)-\Gamma(r+\ell+1,b)\right]\ .
\label{eq:schmidt-integral}
\end{align}
Here, $\Gamma(\,\cdot\,,\,\cdot\,)$ is the upper incomplete gamma function.
The second equality
follows by expanding both Laguerre polynomials and using
$\int_a^b \xi^\ell e^{-\xi}\,\dif \xi=\Gamma(\ell+1,a)-\Gamma(\ell+1,b)$ term by term.

For the first shell, write
$\mathcal A_\Theta:=\sum_{k\in\mathcal K}c_ke^{ik\Theta}$.
Since $L_k(\xi)=1-k\xi+O(\xi^2)$ and
$e^{-\xi}=1-\xi+O(\xi^2)$ uniformly over the finite set $\mathcal K$,
\begin{align}
P_{\psi_\Theta,1}(\Delta)
=\int_0^\Delta e^{-\xi}
\left|\sum_{k\in\mathcal K}c_ke^{ik\Theta}L_k(\xi)\right|^2\,\dif \xi=\Delta|\mathcal A_\Theta|^2+O(\Delta^2)\ .
\end{align}
Since $1-e^{-\Delta}=\Delta+O(\Delta^2)$,
\begin{align}
P_{\psi_\Theta,1}(\Delta)-(1-e^{-\Delta})
=\Delta(|\mathcal A_\Theta|^2-1)+O(\Delta^2)\ .
\label{eq:schmidt-fine}
\end{align}
The phase average of
$|\mathcal A_\Theta|^2$ is
$\sum_{k\in\mathcal K}|c_k|^2=1$. If at least two coefficients are nonzero, this
trigonometric polynomial cannot be constant. A polynomial of constant modulus
on the unit circle is a monomial. It therefore exceeds $1$ for some $\Theta$.
Hence every entangled state in this Schmidt-correlated family with finite
Fock support
violates the first partial-sum inequality at sufficiently small $\Delta$.
The required resolution may depend on the coefficients $\{c_k\}$.

The MGVT comparison follows directly from the same
coefficients. If the
occupied photon numbers differ pairwise by at least two, then
$\langle\hat a_1\hat a_2\rangle=\sum_{k\ge0}(k+1)c_k^*c_{k+1}=0$, where
coefficients outside $\mathcal K$ are understood to vanish. The other
intermode second moments also vanish, and each reduced state is Fock diagonal.
Thus the covariance matrix of $(u,v)$ is
$\Sigma=(\bar n+1/2)I_2$, where
$\bar n=\sum_{k\in\mathcal K}k|c_k|^2$. Every entangled member of this
subclass has $\det\Sigma=(\bar n+1/2)^2>1/4$, so even the optimized MGVT
criterion does not detect its entanglement.

For the benchmark in \cref{fig:epr-nongaussian-benchmarks}, set
$\ket{\psi_p}=\sqrt{1-p}\ket{0,0}+\sqrt p\ket{2,2}$ and $\Theta=0$. The
radial EPR density is then
$g_p(\xi)=e^{-\xi}[\sqrt{1-p}+\sqrt p\,L_2(\xi)]^2$, so all shell
probabilities follow analytically from \cref{eq:schmidt-integral}. For the
balanced transformation,
$U_{\rm BS}\ket{2,2}=\sqrt{3/8}(\ket{0,4}+\ket{4,0})-\ket{2,2}/2$.

The two reduced output states of $U_{\rm BS}\ket{\psi_p}$ coincide:
\begin{align}
\omega_p
&=\left(1-\frac{5p}{8}\right)\ket0\!\bra0
+\frac p4\ket2\bra2+\frac{3p}{8}\ket4\!\bra4 \notag\\[-0.35em]
&\quad+\sqrt{\frac{3p(1-p)}8}
\bigl(\ket0\!\bra4+\ket4\!\bra0\bigr)\ .
\label{eq:schmidt-physical-marginal-p}
\end{align}
Consequently, the distributions of $\hat u=\hat q_-$ and
$\hat v=\hat p_+$ are the position and momentum distributions of the same
physical one-mode state $\omega_p$. For every conjugate pair of R\'enyi orders,
$1/\alpha+1/\beta=2$, the Babenko--Beckner
relation therefore gives \cite{beckner75,saboia11}
\begin{align}
h_\alpha(u)+h_\beta(v)
&\ge\mathcal B_{\alpha,\beta}\ .
\label{eq:schmidt-stw-balanced}
\end{align}
Here $\mathcal B_{\alpha,\beta}:=\ln\pi-
\ln\alpha/[2(1-\alpha)]-\ln\beta/[2(1-\beta)]$.
The endpoint values are understood by continuity. This is the STW separability
bound \cite{saboia11} in the present normalization, so the balanced STW
criterion detects no $\ket{\psi_p}$.

The same conclusion holds after gain optimization. Up to an overall scale, the
nonlocal variables are
$\hat u_t=(t\hat q_1-\hat q_2)/\sqrt{1+t^2}$ and
$\hat v_t=(\hat p_1+t\hat p_2)/\sqrt{1+t^2}$, with $t>0$.
They correspond to a beam splitter with $T=t^2/(1+t^2)$, $R=1-T$, and
\begin{align}
U_T\ket{2,2}=\sum_{k=0}^4b_k(T)\ket{k,4-k}\ .
\label{eq:schmidt-unbalanced-bs}
\end{align}
Its coefficients are
\begin{align*}
b_0=b_4=\sqrt6\,TR,\qquad
b_1=-b_3=\sqrt{6TR}\,(R-T),\qquad
b_2=1-6TR\ .
\end{align*}
The symmetry $b_k^2=b_{4-k}^2$, together with $b_0=b_4$, makes the two
output reductions identical, including their vacuum--four-photon coherence.
Their quadratures therefore obey \cref{eq:schmidt-stw-balanced}, whereas the
PPT separability bound is only
$\mathcal B_{\alpha,\beta}+\ln(2\sqrt{TR})\le\mathcal B_{\alpha,\beta}$.
Rotations change only the coherence phase. Thus gains, rotations, and entropy
orders cannot produce an STW violation for $0<p<1$.

For GHN \cite{garttner23}, the same
partial-transpose reduction associates the balanced mixed-variable marginal
with the Husimi function of the trace-one operator
\begin{align*}
\sigma_p
&:=\Tr_+\!\left[
U_{\rm BS}\bigl(\ket{\psi_p}\!\bra{\psi_p}\bigr)^\Gamma
U_{\rm BS}^\dagger\right]\\[-0.35em]
&=d_0\ket0\!\bra0+d_1\ket1\!\bra1
+d_2\ket2\!\bra2+d_4\ket4\!\bra4\ .
\end{align*}
The coefficients are
\begin{align*}
d_0&=1-\frac{5p}{8}+\frac{\chi_p}{2},\qquad
d_1=-\chi_p,\qquad
d_2=\frac{p}{4}+\frac{\chi_p}{2},\\[-0.35em]
d_4&=\frac{3p}{8},\qquad
\chi_p=\sqrt{p(1-p)}\ .
\end{align*}
For an entangled input, $\sigma_p$ need not be positive. Nevertheless, its
Husimi function below equals the non-negative measured marginal.
In the convention of Ref. \cite{garttner23},
$Q_p(\alpha)=e^{-y}(d_0+d_1y+d_2y^2/2+d_4y^4/24)$, where
$y=|\alpha|^2$ and the measure is $\dif^2\alpha/\pi$. Its continuous
majorization witness is
\begin{align*}
\mathcal M_Q(p):=\sup_{t\ge0}\left\{
\int\frac{\dif^2\alpha}{\pi}\,[Q_p(\alpha)-t]_+-\mathcal V_0(t)
\right\}\ ,\qquad [z]_+:=\max\{z,0\}\ .
\end{align*}
The vacuum function is
\begin{align*}
\mathcal V_0(t)=
\begin{cases}
1-t+t\ln t,&0<t<1\ ,\\
0,&t\ge1\ ,
\end{cases}
\qquad \mathcal V_0(0)=1\ .
\end{align*}
Here, $\mathcal M_Q(p)>0$ signals a continuous-majorization violation and
therefore a violation of the GHN criterion. The boundary for
criterion \eqref{eq:epr-criterion} follows from
\begin{align*}
I(\Delta)&:=\int_0^\Delta e^{-\xi}L_2(\xi)\,\dif \xi
=e^{-\Delta}\Delta\left(1-\frac{\Delta}{2}\right)\ .
\end{align*}
The second auxiliary integral is
\begin{align*}
J(\Delta)&:=\int_0^\Delta e^{-\xi}\left[L_2^2(\xi)-1\right]\,\dif \xi
=-e^{-\Delta}\Delta^2\left(2-\Delta+\frac{\Delta^2}{4}\right)\ .
\end{align*}
Indeed, the first-shell gap is
\begin{align}
P_{\psi_p,1}(\Delta)-P_{0,1}(\Delta)
=2\sqrt{p(1-p)}\,I(\Delta)+pJ(\Delta)\ .
\label{eq:schmidt-first-shell-gap}
\end{align}
For $0<\Delta<2$, it is positive precisely when $0<p<p_*(\Delta)$, where
\begin{align}
p_*(\Delta)=
\frac{(2-\Delta)^2}
{(2-\Delta)^2+\Delta^2\left(2-\Delta+\Delta^2/4\right)^2}\ .
\label{eq:schmidt-shell-boundary}
\end{align}
For $\Delta\ge2$ the first-shell gap is non-positive. Numerical sorting of all
shell weights on a $25\times49$ grid over $0.02\le p\le0.98$ and
$0.04\le\Delta\le1.96$, with radial cutoff $\xi=60$, found that, on this grid,
no partial sum with $K>1$ enlarged the first-shell violation region in
\cref{fig:epr-nongaussian-benchmarks}. The optimized MGVT criterion detects no member
with $p\in(0,1)$.

The balanced GHN endpoint can be found analytically. The radial
function $Q_p(y)$ is strictly decreasing for every $p\in[0,1]$. To see
this for $p>0$, set $a_p=\sqrt{(1-p)/p}$ and write
$Q_p'(y)=-e^{-y}R_p(y)$. Direct simplification gives
\begin{align}
64(1+a_p^2)R_p(y)
={}&y^4-4y^3+(8+16a_p)y^2-(16+96a_p)y
+24+96a_p+64a_p^2>0\ .
\label{eq:schmidt-q-monotone}
\end{align}
Indeed, with $h(y)=y^2-6y+6$, the right-hand side is
$y^2(y-2)^2+4[(y-2)^2+2]+16a_p h(y)+64a_p^2$. If $h(y)\ge0$, its minimum over
$a_p\ge0$ occurs at $a_p=0$. If $h(y)<0$, the minimum occurs at
$a_p=-h(y)/8$ and equals
$4(2y^3-10y^2+14y-3)$, whose minimum on
$(3-\sqrt3,3+\sqrt3)$ is $68/27$.
Continuous majorization can therefore be tested by centered disks. Their
cumulative difference from the vacuum is
\begin{align}
\int_0^Y[Q_p(y)-e^{-y}]\,\dif y
=Ye^{-Y}\bigg[
\frac{\chi_p}{2}-\frac{5p}{8}
-\left(\frac{5p}{16}+\frac{\chi_p}{4}\right)Y
-\frac{pY^2}{16}-\frac{pY^3}{64}\bigg]\ .
\label{eq:schmidt-q-cumulative}
\end{align}
It is positive for some $Y>0$ exactly when
$\chi_p/2-5p/8>0$, or $0<p<16/41$.

We also searched $\mathcal M_Q$ numerically over the relative scaling, local
phase, and nonlocal squeezing allowed in Ref. \cite{garttner23}. The
non-negative measured $Q$ marginal was evaluated directly from the two-mode Husimi
distribution; its complementary Gaussian integrations were performed by
Gauss--Hermite quadrature. Complete phase periods, logarithmic scaling and
squeezing parameters in $[-1.6,1.6]$, and phase-space grids up to $201$
points per axis were used, with independent Fock-space calculations through
cutoff $31$ as a convergence check. The search reproduces the balanced
violation below $p=16/41$. After refinement of the grids and cutoff, no
violation above this endpoint was observed. This conclusion is numerical; the
balanced endpoint itself is exact. This endpoint is the GHN threshold plotted in
\cref{fig:epr-nongaussian-benchmarks}(b).

\section{Extension to Cahill--Glauber
\texorpdfstring{$s$}{s}-ordered distributions}
\label{app:sparam}

The finite-shell relation extends from the Husimi $Q$ function at $s=-1$ to
the Wigner function at $s=0$ \cite{cahill69a,cahill69b}. For an arbitrary
single-mode input, rotational covariance makes its radial shell weights
identical to those of its phase average $\bar\rho$. Hence, for every complete
radial partition $\Pi$ with bounded squared-radius widths and every
$s\in[-1,0]$, componentwise
non-negativity of $\boldsymbol P_\rho^{(s)}(\Pi)$ implies
$\boldsymbol P_\rho^{(s)}(\Pi)\prec
\boldsymbol P_0^{(s)}(\Pi^\downarrow)$. A separate continuous statement is
established here for the one-dimensional radial distribution of the phase-averaged
state: if $g_{\bar\rho}^{(s)}$ is pointwise non-negative, then
$g_{\bar\rho}^{(s)}\prec_{\mathrm{cont}}g_0^{(s)}$.

Let $W_s(\alpha;\rho)$ denote the normalized Cahill--Glauber
$s$-ordered quasiprobability \cite{cahill69a,cahill69b}, with
$\alpha=(q+ip)/\sqrt2$. The values $s=-1$ and $s=0$ correspond to the
Husimi $Q$ and Wigner representations, respectively.
For $s\le0$, $W_s$ is well defined for every density operator.
Representations of the same state at $s_1>s_2$ satisfy
\begin{align}
W_{s_2}(\alpha;\rho)
&=\frac{2}{\pi(s_1-s_2)}
\int W_{s_1}(\beta;\rho)\,
e^{-2|\alpha-\beta|^2/(s_1-s_2)}\,\dif^2\beta\ .
\label{eq:s-gauss}
\end{align}
The convolution kernel is strictly positive.

Set $\xi=q^2+p^2=2|\alpha|^2$. For a phase-invariant state
$\bar\rho=\sum_n\rho_n\ket n\bra n$, define
$g_{\bar\rho}^{(s)}(\xi):=(\pi/2)W_s(\alpha;\bar\rho)|_{|\alpha|^2=\xi/2}$.
The normalization for a general countable Fock mixture will be established
below. For $s\in[-1,0]$, introduce
$\chi_s=(1+s)/(1-s)$ and $y=\xi/(1-s)$. For $-1<s\le0$, additionally set
$a_s=2/(1+s)$.
Write the rescaled radial profile as
$\widetilde g_{\bar\rho}^{(s)}(y):=
(1-s)g_{\bar\rho}^{(s)}((1-s)y)$. For $-1<s\le0$, the Fock-state profile is
given by \cite{cahill69b}
\begin{align}
\widetilde g_n^{(s)}(y)=(-\chi_s)^nL_n(a_s y)e^{-y},
\qquad \widetilde g_0^{(s)}(y)=e^{-y}\ .
\label{eq:gn-y}
\end{align}
Thus the vacuum is independent of $s$ in the rescaled coordinate. The
identity $1+\chi_s t(1-a_s)=1-t$, used below, follows from
$\chi_s(1-a_s)=-1$.

\smallskip
\noindent\textit{Core identity.}

Define
$\Psi_n^{(s)}(y):=e^{y}\int_y^\infty
\widetilde g_n^{(s)}(u)\,\dif u$, normalized so that
$\Psi_n^{(s)}(0)=1$.

For every $s\in[-1,0]$, $n\ge1$, and $y\ge0$, one has
\begin{align}
\Psi_n^{(s)\prime}(y) = n(1+\chi_s)\, e^{y}\,
\widetilde g_{\sigma_n}^{(s)}(y)\ ,
\label{eq:core}
\end{align}
where $\sigma_n:=n^{-1}\sum_{j=0}^{n-1}\ket j\bra j$. At the Wigner
endpoint, \cref{eq:core} reduces to the Laguerre derivative identity used in
the radial proof above.

For $-1<s\le0$, the Laguerre generating function with $w=-\chi_s t$ and
$\chi_s(1-a_s)=-1$ gives
\begin{align*}
\sum_{n=0}^\infty \widetilde g_n^{(s)}(y)t^n
&=\frac{e^{-\omega_s(t)y}}{1+\chi_s t}\ ,\qquad
\omega_s(t):=\frac{1-t}{1+\chi_s t}\ .
\end{align*}
The corresponding generating function for $\Psi_n^{(s)}$ is
\begin{align*}
\sum_{n=0}^\infty \Psi_n^{(s)}(y)t^n
&=\frac{e^{y\phi_s(t)}}{1-t}\ ,\qquad
\phi_s(t):=1-\omega_s(t)=\frac{t(1+\chi_s)}{1+\chi_s t}\ .
\end{align*}
Using $\Psi_n'=\Psi_n-e^y\widetilde g_n$ and
$n\widetilde g_{\sigma_n}=\sum_{j=0}^{n-1}\widetilde g_j$, the two sides
of \cref{eq:core} have generating functions
\begin{align*}
\sum_{n=0}^\infty \Psi_n^{(s)\prime}(y)t^n
&=\frac{\phi_s(t)e^{y\phi_s(t)}}{1-t}\ .
\end{align*}
For the right-hand side,
\begin{align*}
\sum_{n\ge1}\bigl[n(1+\chi_s)e^y
\widetilde g_{\sigma_n}^{(s)}(y)\bigr]t^n
&=(1+\chi_s)e^y\frac{t}{1-t}
\frac{e^{-\omega_s(t)y}}{1+\chi_s t}\ .
\end{align*}
They coincide because
$\phi_s(t)(1+\chi_s t)=t(1+\chi_s)$ and
$\phi_s(t)+\omega_s(t)=1$, proving
\cref{eq:core} coefficientwise.

At $s=-1$, $\widetilde g_n^{(-1)}(y)=e^{-y}y^n/n!$ and
$\Psi_n^{(-1)}(y)=e^y\int_y^\infty e^{-u}u^n/n!\,\dif u
=\sum_{j=0}^{n}y^j/j!$. Differentiation verifies \cref{eq:core} directly.

\smallskip
\noindent\textit{Non-negativity of the $s$-ordered radial quasiprobability of
$\sigma_n$.}

The sign needed in the majorization argument is
\begin{align}
\widetilde g_{\sigma_n}^{(s)}(y)\ge0
\qquad(s\in[-1,0],\ y\ge0,\ n\ge1)\ .
\label{eq:sigma-pos-s}
\end{align}
At $s=0$, \cref{eq:gn-y,eq:S-def-app} give
$\widetilde g_{\sigma_n}^{(0)}(y)=e^{-y}\sum_{j=0}^{n-1}
(-1)^jL_j(2y)/n=e^{-y}S_{n-1}(y)/n\ge0$,
by \cref{eq:mehler-sos-app}. For $s<0$, \cref{eq:s-gauss} is convolution
with a positive Gaussian kernel, so it preserves this non-negativity.

For a countable Fock expansion, the required interchanges follow from a
uniform phase-point-operator bound. The Cahill--Glauber representation gives
\begin{align}
\widetilde g_{\bar\rho}^{(s)}(y)
&=\Tr\!\left[
\bar\rho\,D(\alpha_y)(-\chi_s)^{\hat n}D^\dagger(\alpha_y)
\right],
\qquad
|\alpha_y|^2=\frac{1-s}{2}y ,
\label{eq:s-phase-point}
\end{align}
where the phase of $\alpha_y$ is immaterial and, at $s=-1$, the operator
$0^{\hat n}$ is understood as $\ket0\bra0$. Since
$\|(-\chi_s)^{\hat n}\|\le1$ for $s\in[-1,0]$,
$|\widetilde g_n^{(s)}(y)|\le1$.
Moreover,
\begin{align}
\widetilde g_{\bar\rho}^{(s)}(y)
&=\sum_{n=0}^\infty\rho_n\widetilde g_n^{(s)}(y)\ .
\label{eq:s-profile-mixture}
\end{align}
The latter series converges uniformly for $y\ge0$.

\Cref{eq:core,eq:sigma-pos-s} imply
$\Psi_n^{(s)\prime}(y)\ge0$ and hence
$\Psi_n^{(s)}(y)\ge\Psi_n^{(s)}(0)=1$. Define
$\mathcal T_n^{(s)}(y):=e^{-y}\Psi_n^{(s)}(y)
=1-\int_0^y\widetilde g_n^{(s)}(u)\,\dif u$.
Then $\mathcal T_n^{(s)}(y)\ge0$, while \cref{eq:s-phase-point} gives
$\mathcal T_n^{(s)}(y)\le1+y$. Therefore, for every phase-invariant state
$\bar\rho=\sum_n\rho_n\ket n\bra n$,
\begin{align}
\mathcal T_{\bar\rho}^{(s)}(y)
&:=\sum_{n=0}^\infty\rho_n\mathcal T_n^{(s)}(y)\notag\\
&=1-\int_0^y\widetilde g_{\bar\rho}^{(s)}(u)\,\dif u
=e^{-y}\sum_{n=0}^\infty\rho_n\Psi_n^{(s)}(y)
\ge e^{-y}.
\label{eq:s-tail-dom}
\end{align}
Here uniform convergence on finite intervals justifies exchanging the Fock
sum and the integral.

For $0<a_{\rm L}<1$, the Laguerre Laplace transform, with the Gamma expression used
at $s=-1$, gives
\begin{align}
\widehat g_{\bar\rho}^{(s)}(a_{\rm L})
&:=\int_0^\infty e^{-a_{\rm L}y}
\widetilde g_{\bar\rho}^{(s)}(y)\,\dif y \notag\\
&=\frac{1}{1+a_{\rm L}}
\sum_{n=0}^\infty\rho_n
\left(\frac{1-\chi_s a_{\rm L}}{1+a_{\rm L}}\right)^n
\xrightarrow[a_{\rm L}\downarrow0]{}1 .
\label{eq:s-radial-laplace}
\end{align}
\Cref{eq:s-phase-point} justifies exchanging the Fock sum and the Laplace
integral, and dominated convergence in $n$ gives the limit. Since
$\mathcal T_{\bar\rho}^{(s)}(y)\le1+y$, integration by parts yields
\begin{align}
a_{\rm L}\int_0^\infty e^{-a_{\rm L}y}
\mathcal T_{\bar\rho}^{(s)}(y)\,\dif y
=1-\widehat g_{\bar\rho}^{(s)}(a_{\rm L})
\xrightarrow[a_{\rm L}\downarrow0]{}0 .
\label{eq:s-tail-abel}
\end{align}
If $\widetilde g_{\bar\rho}^{(s)}$ is pointwise non-negative, then
$\mathcal T_{\bar\rho}^{(s)}$ is non-increasing.
\Cref{eq:s-tail-abel} therefore forces
$\mathcal T_{\bar\rho}^{(s)}(y)\to0$, proving
$\int_0^\infty\widetilde g_{\bar\rho}^{(s)}(y)\,\dif y=1$ in the
pointwise-positive case.

\smallskip
\noindent\textit{Continuous majorization.}

Assume that $\widetilde g_{\bar\rho}^{(s)}$ is pointwise non-negative.
The normalization proved above turns $\mathcal T_{\bar\rho}^{(s)}$ into its
survival function,
$\mathcal T_{\bar\rho}^{(s)}(y)=
\int_y^\infty\widetilde g_{\bar\rho}^{(s)}(u)\,\dif u$.
Moreover, \cref{eq:core,eq:sigma-pos-s} give
\begin{align}
\mathcal T_{\bar\rho}^{(s)}(y)-\widetilde g_{\bar\rho}^{(s)}(y)
&=\sum_{n=0}^\infty\rho_n
\bigl[\mathcal T_n^{(s)}(y)-\widetilde g_n^{(s)}(y)\bigr]\ge0,
\label{eq:s-survival-density}
\end{align}
because each bracket equals $e^{-y}\Psi_n^{(s)\prime}(y)$. The sum is
justified by \cref{eq:s-profile-mixture} and the bounds above.

Set $f(y):=\widetilde g_{\bar\rho}^{(s)}(y)$,
$T(y):=\mathcal T_{\bar\rho}^{(s)}(y)$, and define the hazard rate
$h(y):=f(y)/T(y)$. Since
\cref{eq:s-tail-dom,eq:s-survival-density} imply $T(y)>0$ and
$0\le h(y)\le1$, the identity $T'=-f$ gives $T(y)=e^{-H(y)}$, where
$H(y):=\int_0^y h(u)\,\dif u$.
For any measurable $A\subset[0,\infty)$ of finite Lebesgue measure $|A|$,
let $H_A(y):=\int_0^y\mathbbm{1}_A(u)h(u)\,\dif u$. Since $H\ge H_A$,
\begin{align}
\int_A f(y)\,\dif y
&=\int_0^\infty\mathbbm{1}_A(y)h(y)e^{-H(y)}\,\dif y\notag\\
&\le\int_0^\infty\mathbbm{1}_A(y)h(y)e^{-H_A(y)}\,\dif y\notag\\
&=1-\exp\!\left[-\int_A h(y)\,\dif y\right]
\le1-e^{-|A|}.
\label{eq:s-set-concentration}
\end{align}
Taking the supremum over measurable sets of measure $m$ and using the
characterization of the decreasing rearrangement by maximal set integrals
\cite{hardy29,ryff65,chong74} yields
\begin{align}
\int_0^m f^\downarrow(y)\,\dif y
&=\sup_{|A|=m}\int_A f(y)\,\dif y
\le1-e^{-m}
=\int_0^m e^{-y}\,\dif y.
\label{eq:s-rearrangement-concentration}
\end{align}
Both cumulative integrals tend to unity as $m\to\infty$. Hence
\begin{align}
\widetilde g_{\bar\rho}^{(s)}(y)
&\prec_{\mathrm{cont}}e^{-y},\qquad s\in[-1,0].
\label{eq:s-cont-majorization}
\end{align}
This proves the continuous statement for every countable Fock mixture whose
radial distribution is pointwise non-negative. Undoing the common positive rescaling gives
$g_{\bar\rho}^{(s)}\prec_{\mathrm{cont}}g_0^{(s)}$.

\smallskip
\noindent\textit{Finite-shell majorization.}
The finite-shell relation follows under the weaker condition that the
integrated shell weights are non-negative.
Let $\Pi$ be a complete radial partition
$0=r_0<r_1<\cdots$ with $r_k\to\infty$ and bounded squared-radius widths
$d_k=r_k^2-r_{k-1}^2$, and let
$s\in[-1,0]$. Define
$P_{\rho,k}^{(s)}(\Pi):=\int_{r_{k-1}^2}^{r_k^2}
g_{\bar\rho}^{(s)}(\xi)\,\dif\xi$, and let $D_m(\Pi)$ be defined by
\cref{eq:Dm-app}.
If $\boldsymbol P_\rho^{(s)}(\Pi)$ is componentwise non-negative,
then
\begin{align}
\boldsymbol P_\rho^{(s)}(\Pi)
&\prec
\boldsymbol P_0^{(s)}(\Pi^\downarrow)\ .
\label{eq:s-general-majorization}
\end{align}
The ordered vacuum components are
\begin{align}
P_{0,m}^{(s)}(\Pi^\downarrow)
&=e^{-D_{m-1}(\Pi)/(1-s)}
-e^{-D_m(\Pi)/(1-s)}\ .
\label{eq:s-ordered-vacuum}
\end{align}

Rotational covariance reduces the proof to
$\bar\rho=\sum_n\rho_n\ket n\bra n$. In the coordinate
$y=\xi/(1-s)$, set $y_k=r_k^2/(1-s)$,
$\widetilde d_k=d_k/(1-s)$, and
$T_k=\mathcal T_{\bar\rho}^{(s)}(y_{k-1})$.
\Cref{eq:core,eq:sigma-pos-s} make
$\Psi_{\bar\rho}^{(s)}=\sum_n\rho_n\Psi_n^{(s)}$ non-decreasing, which gives
$T_{k+1}\ge e^{-\widetilde d_k}T_k$. The identity
$P_{\rho,k}^{(s)}=\int_{y_{k-1}}^{y_k}
\widetilde g_{\bar\rho}^{(s)}(y)\,\dif y=T_k-T_{k+1}$ holds independently of
sign. Componentwise non-negativity therefore makes the boundary tails $T_k$
decrease to some $L\ge0$. The rescaled widths have the finite upper bound
$\widetilde D:=\sup_k\widetilde d_k<\infty$. For
$y\in[y_{k-1},y_k]$, monotonicity of $\Psi_{\bar\rho}^{(s)}$ gives
\begin{align*}
\mathcal T_{\bar\rho}^{(s)}(y)
\ge e^{-(y-y_{k-1})}T_k
\ge e^{-\widetilde D}L .
\end{align*}
\Cref{eq:s-tail-abel} forces $L=0$. Consequently,
$\sum_{k=1}^\infty P_{\rho,k}^{(s)}=T_1-\lim_{k\to\infty}T_k=1$, and the
non-negative shell sequence is a probability vector. The subset-sum
argument then gives
$\sum_{j=1}^m(P_{\rho,j}^{(s)})^\downarrow
\le1-e^{-D_m(\Pi)/(1-s)}$, which is the cumulative form of
\cref{eq:s-general-majorization}.

For equal-area shells $d_k=\Delta$, ordering the widths leaves the partition
unchanged and \cref{eq:s-general-majorization} reduces to
\begin{align}
\boldsymbol P_\rho^{(s)}(\Delta)
&\prec\boldsymbol P_0^{(s)}(\Delta)\ .
\label{eq:s-equal-area-majorization}
\end{align}
The vacuum shell probabilities are
$P_{0,k}^{(s)}(\Delta)=[1-e^{-\Delta/(1-s)}]
e^{-(k-1)\Delta/(1-s)}$.

At $s=0$, the continuous statement agrees with the
phase-invariant result of Ref. \cite{vanbever21}, and the finite-resolution
statement reduces to the main-text relation.

\end{document}